\pdfoutput=1
\documentclass[12pt]{article}
\usepackage{amsmath, amssymb}
\usepackage{graphicx}
\usepackage{enumerate}
\usepackage{natbib}
\usepackage{overpic}
\usepackage{subfiles}

\newcommand{\blind}{1}

\usepackage{enumitem}
\usepackage{mathrsfs, booktabs, amsthm, amsfonts, bbold}
\usepackage{hyperref}
\usepackage{accents}

\usepackage{xcolor}
\protected\def\[#1\]{\begin{equation}\begin{aligned}#1\end{aligned}\end{equation}}
\protected\def\(#1\){\begin{equation*}\begin{aligned}#1\end{aligned}\end{equation*}}

\newtheorem{theorem}{Theorem}
\newtheorem{proposition}{Proposition}
\newtheorem{lemma}{Lemma}
\newtheorem{corollary}{Corollary}

\theoremstyle{remark}
\newtheorem{remark}{Remark}

\theoremstyle{definition}
\newtheorem{definition}{Definition}
\newtheorem{assumption}{Condition}

\usepackage{subcaption}
\usepackage{float}

\begin{document}

\def\spacingset#1{\renewcommand{\baselinestretch}
{#1}\small\normalsize} 
\spacingset{1}

\if1\blind
{
  \title{\bf Statistical Modeling of \\
  Combinatorial Response Data}
  \author{
  	Yu Zheng \\
	Department of Statistics, University of Florida\\
    and \\
	Malay Ghosh\\
	Department of Statistics, University of Florida\\
		    and \\
	Leo Duan\thanks{Corresponding author. \href{email:li.duan@ufl.edu}{li.duan@ufl.edu}}\\
	Department of Statistics, University of Florida
    }
    \date{}
  \maketitle
  \vspace{-0.5in}
} \fi

\if0\blind
{
  \bigskip
  \bigskip
  \bigskip
  \begin{center}
    {\LARGE\bf Statistical Modeling of \\
  Combinatorial Response Data}
\end{center}
  \medskip
} \fi

\bigskip

\begin{abstract}
There is a rich literature for modeling binary and polychotomous responses. However, existing methods are inadequate for handling combinatorial responses, where each response is an integer array under additional constraints. Such data are increasingly common in modern applications, such as surveys collected under skip logic, event propagation on a network, and observed matching in ecology. Ignoring the combinatorial structure leads to biased estimation and prediction. The fundamental challenge is the lack of a link function that connects a linear or functional predictor with a probability respecting the combinatorial constraints. In this article, we propose a novel augmented likelihood that views combinatorial response as a deterministic transform of a continuous latent variable. We specify the transform as the maximizer of integer linear program, and characterize useful properties such as dual thresholding representation. When taking a Bayesian approach and considering a multivariate normal distribution for the latent variable, our method becomes a direct generalization to the celebrated probit data augmentation, and enjoys straightforward computation via Markov chain Monte Carlo. We provide theoretical justification, including consistency and applicability, at an interesting intersection between duality and probability. We demonstrate the effectiveness of our method through simulations and a data application on the seasonal matching between waterfowl.
\end{abstract}
\noindent
{\it Keywords:} Implicit link function, Strong duality, Truncated Normal, Continuous Embedding \vfill

\addtolength{\textheight}{-0.8in}
\addtolength{\topmargin}{0.5in}

\newpage

\spacingset{1.8} 

\section{Introduction}
In statistical data analysis, combinatorial response data are becoming increasingly common. For example, in online surveys, since the user can terminate the survey more easily than in in-person questionnaires, survey designers often use skip logic to present only a subset of questions to each user. Depending on the user's answer to each branching question, the vector record for each user's answers has some elements marked as skipped (recorded as $0$). Those zeros are often referred to as {\em structural zeros} in the literature \citep{manrique2014bayesian}, as they appear not due to random mechanism, but due to the combinatorial constraints. Similarly, in ecological studies, seasonally monogamous animals like ducks form pair bonds gradually over the course of a year. Each observed matching graph at a given time is a subgraph of a bipartite graph, which describes the possible pair bonds (for example, between male and female, belonging to the same species), and each node in the matching graph can have at most one edge, corresponding to the pairing constraint.

In applications such as above, the multivariate data are combinatorial since they are not only discrete, but also subject to additional constraints. That means the response space is in fact much smaller than the product space of all possible marginal outcomes. Under a likelihood approach, the probability of those outcomes not satisfying the constraints is effectively zero. It is of great interest to characterize the probability distribution of such data and develop regression models characterizing their relationship with covariates.

The existing strategies to model such combinatorial response can be put into two categories. First, one may enumerate all possible outcomes in the combinatorial space and reindex each response as a univariate multinomial outcome instead. Such a strategy is suitable when the response is in a low dimension, and the number of possibilities in each dimension is finite. On the other hand, such a reindexing strategy becomes impractical as the dimension increases exponentially \citep{baptista2018bayesian}, and reindexing may lead to ignoring ordinal information between two combinatorial outcomes. 

Second, one may simply ignore the constraints and model each data point as an array of discrete outcomes. To elucidate the relation between two likelihoods, $\mathcal L(y;\theta)$ upholding and $\bar {\mathcal L}(y;\theta)$ ignoring the constraints, with $y$  the data and $\theta$ the parameter, we can think of $\mathcal L(y;\theta)$ as a conditional probability given that discrete vector falling inside the constrained space denoted by $\mathcal Y$ hence $\mathcal L(y;\theta)=\bar L(y;\theta)1(y\in\mathcal Y)/P(y \in \mathcal Y;\theta)$. Clearly, omitting $P(y \in \mathcal{Y};\theta)$ results in a biased estimation of $\theta$, which leads to inconsistency as we demonstrate later. Similarly, as has long been recognized in the literature on log-linear and latent class models, failure to properly account for structural zeros introduces bias into the estimated distribution \citep{goodman1968analysis,fienberg1972multiple,bishop1975discrete,vermunt1997general}; for a more recent discussion, see \cite{manrique2014bayesian,carota2015bayesian}.

To appropriately model combinatorial data under constraints, we propose a generative approach that first samples a latent vector from an unconstrained continuous distribution, and then deterministically maps it to a combinatorial outcome in $\mathcal{Y}$.
This continuous latent representation, known as an {\em embedding}, builds on foundational work by \cite{albert1993bayesian} who showed that binary outcomes can be modeled by thresholding a univariate normal at zero. Their key insight was that marginalizing out this latent variable yields probabilities equivalent to those from a generalized linear model (GLM) with a link function. However, the embedding framework offers greater flexibility than traditional GLMs.

Several extensions demonstrate the versatility of latent embeddings. Ordinal data can be handled by introducing multiple thresholds to partition the latent space into ordered categories \citep{albert1993bayesian}. For network data with binary symmetric adjacency matrices, low-rank matrix embeddings with normal noise have been employed \citep{hoff2002latent}. Rank data correlations can be modeled using transformed multivariate normal vectors \citep{hoff2007extending}. In spatial statistics, spatial dependence between categorical data can be captured through latent Gaussian processes \citep{rue2009approximate}. These examples highlight how embeddings can naturally incorporate complex dependencies and constraints.

Our approach stems from reinterpreting the zero-thresholding transform of \cite{albert1993bayesian} as the solution to a univariate integer linear program. We extend this insight by generalizing the transform to a multivariate integer linear program, which naturally accommodates combinatorial constraints. Since we observe the combinatorial data, our objective is not to solve an optimization problem, but to determine the continuous embedding that corresponds to the observed integer solution. We demonstrate that this continuous embedding is analytically tractable through strong duality, and we leverage this tractability to develop a Metropolis-Hastings-Within-Gibbs sampler with excellent computational performance.

\vspace{-0.8cm}
\section{Method}
Our focus is on regression with $d$-dimensional integer response $y_i$ in the space
$
\{ z\in \mathbb{Z}^d: A z\le b\},
$
where $A$ is an $m\times d$ matrix, $b$ is a
$m$-element vector, $\mathbb{Z}$ is the integer set, and $\le$ applies elementwise. The coefficients in $A$ and $b$ are pre-determined (up to common scaling) by known constraints for the data. To be concrete, we  now list a few simple examples, and will introduce a few more complex constraints in Section 2.4.

First, for  binary constraint $y_i\in \{0,1\}^d$, $Ay_i\le b$ corresponds to $0\le y_{i,j}\le 1$ with $A=  [I, -I]$ and $b=[1_d, 0_d]$. Second, for equality constraint $a_k^\top y_i=b_k$, such as zero-contrast $y_{i,1}+ y_{i,2}-y_{i,3}-y_{i,4}=0$, $A y_i\le b$ contains two inequalities $a_k^\top y_i\le b_k$ and $-a_k ^\top y_i\le -b_k$. Third, for ordered or partially ordered integer data, such as non-decreasing sequence $y_{i,1}\le y_{i,2}\le \ldots \le y_{i,d}$, $Az\le b$ contains inequality $y_{i,j}-y_{i,k}\le 0$ for suitable $(j,k)$.

\vspace{-0.5cm}

\subsection{Data augmentation via integer linear program}\label{subsec:method}

In order to facilitate the introduction of our data augmentation (DA) strategy, we first review the DA for the probit model for simple binary data $y_i\in \{0,1\}$. \cite{albert1993bayesian} shows that for $y_i \sim \text{Bernoulli}[\Phi(\mu_i)]$, with $\mu_i\in \mathbb{R}$ (such as in the form of linear predictor) and $\Phi$ the distribution function for standard normal, there is a latent normal $\zeta_i\sim \text{N} (\mu_i,1)$ which has 
$P(\zeta_i >0) = \Phi(\mu_i), \;P(\zeta_i<0) = 1-\Phi(\mu_i).$
In another words, $\zeta_i$ is a latent variable for $y_i$, via thresholding at zero: $y_i=1(\zeta_i>0).$ We now provide a new perspective using integer linear program, with coefficient $\zeta_i\neq 0$ almost surely:
\(
y_i = \underset{z\in \{0,1\}}{\arg\max}\; \zeta_i z.
\)
It is not hard to see that the above has the same solution: $y_i=1$ if $\zeta_i>0$, and $y_i=0$ if $\zeta_i<0$. Taking this idea further, we now consider $y_i\in \mathbb{Z}^d$ as the solution to multivariate integer linear program:  
\[\label{eq:ilp}
 y_i = \underset{z}{\arg\max}\; \zeta_i^\top z,\qquad \text{subject to } z \in \mathbb{Z}^d, A z\le b.
\]
 We denote the continuous polyhedron as $\mathcal P=\{z^*: Az^*\le b \}$; hence, the feasible region of the above problem is equivalent to $\mathbb{Z}^d\cap \mathcal P$. The parameterization \eqref{eq:ilp} enables us to model combinatorial data $y_i$ under affine constraints $A y_i\le b$. 

\textit{Shopper's selection problem:}
To make the integer linear program more intuitive and interpretable, we now describe a thought experiment involving a shopper's selection. Imagine a shopper $i$ is buying several items from a store, the store has $d$ different items, and each item has a utility score $\zeta_{i,j}$ that represents how desirable the item $j$ is to the shopper $i$ ($\zeta_{i,j}>0$  desirable, and $\zeta_{i,j}<0$  undesirable), and $y_i\in \mathbb{Z}^d$ records the units of each item the shopper selects and buys. The shopper has a set of constraints on the selection. For example, the shopper may have a budget constraint that limits the total cost of the selected items to be at most $b_1$; the shopper may also have a need constraint that mandates at least $b_2$ items in total to be selected under some category; last but not least, the shopper cannot select negative units of items, $y_{i,j}\ge 0$ for all $j$. Clearly, the integer linear program \eqref{eq:ilp} corresponds to the shopper's maximization of total utility score subject to the constraints.

Since we observe $y_i$ only, we need to model $\zeta_i$. In our approach, any continuous distribution on $\mathbb{R}^d$ can be used for $\zeta_i$. We choose the normal distribution for its convenience for posterior computation. We now show that the solution to the integer linear program \eqref{eq:ilp} is unique, when $\zeta_i$ is from a continuous distribution. Consider two $z_1\neq z_2$, both in $\mathbb{Z}^d\cap \mathcal P$. For $\zeta_i$ from a distribution absolutely continuous with respect to Lebesgue measure, we know $\zeta_i^\top (z_1-z_2) =0$ with probability zero. Therefore, for the maximum value $\zeta_i^\top \hat z$, we know $\hat z$ is unique almost surely with respect to the continuous distribution of $\zeta_i$.

The almost sure unique solution means that we can consider \eqref{eq:ilp} as a mapping:
\[
\label{eq:T}
T: \mathbb{R}^{d}\setminus C \mapsto \mathbb{Z}^d, \qquad y_i=T(\zeta_i)=\underset{z \in \mathbb{Z}^d \cap \mathcal P}{\arg\max}\; \zeta_i^\top z,
\]
where $C$ is the set that does not lead to a unique solution, and $C$ has measure zero under the distribution of $\zeta_i$. The maximizer may not have a closed-form and hence $T$ is an {\em implicit function} equivalent to $y_i: \zeta_i^\top y_i = \max_{z\in \mathbb{Z}^d\cap \mathcal P} \zeta_i^\top z$. {In the rest of the article, all of the statements about $T(\zeta_i)$ are meant in the almost sure sense.}

 Taking a generative view of continuous $\zeta_i \sim \mathcal F(\cdot;\mu_i)$ and transform $y_i =T(\zeta_i)$, we have an augmented likelihood:
\[\label{eq:aug_lik}
\mathcal L(y_i, \zeta_i ; \mu_i) = \mathcal F(\zeta_i;\mu_i) 1(y_i = \underset{z \in \mathbb{Z}^d \cap \mathcal P}{\arg\max}\; \zeta_i^\top z).
\]
Marginally, we can integrate over $\zeta_i$ and find the probability of $y_i$ taking one specific combinatorial value $e$:
{
\[\label{eq:lp_link}
P(y_i = e ; \mu_i) = \int_{T^{-1}(e)}
 \mathcal F(\zeta_i;\mu_i) \textup{d} \zeta_i=\int_{\mathbb R^d}\mathcal F(\zeta_i;\mu_i)1(\zeta ^\top e=\max_{z\in\mathbb R^d\cap\mathcal P}\zeta ^\top z)d\zeta_i,
\]
where $T^{-1}(e):=\{\zeta_{i}\in\mathbb R^d:T(\zeta_{i})=e\}$ is the pre-image of $e$.
}
Following the nomenclature in GLM, we refer to \eqref{eq:lp_link} as the Inverse Integer Linear Program (IILP)-link, since it is associated with the link function between $\mu_i$ and $\mathbb E(y_i \mid \mu_i)$.

We now provide an illustrative toy example. Consider two-dimensional $y_i\in \{0,1\}^2$ with $y_{i,1}+y_{i,2}\le 1$. This integer linear program has a closed-form solution
$y_i =
  (0,0),  \text{ when } \zeta_{i,1}<0, \zeta_{i,2}<0;
  (0,1),  \text{ when } \zeta_{i,1}<\zeta_{i,2}, \zeta_{i,2}>0;
  (1,0),  \text{ when } \zeta_{i,1}>\zeta_{i,2}, \zeta_{i,1}>0.
$
If we further specify $\zeta_i\sim \text{N}[ (\mu_0,\mu_0)^\top, I_2]$, then we see the three possible outcomes have likelihood valued at $\Phi^2(-\mu_0)$, $0.5[1-\Phi^2(-\mu_0)]$ and $0.5[1-\Phi^2(-\mu_0)]$, respectively.

While $\zeta_i$ is interpretable via the shopper's selection problem, we treat it as an auxiliary variable as our interest is on the marginal distribution of combinatorial response. There may be alternative parameterizations for the same marginal probability. For example, one could consider $p(y_i; \gamma) = m^{-1}(\gamma) 1(y_{i,1} + y_{i,2} \leq 1) \prod_{i=1}^2 \Phi(\gamma)^{y_{i,j}}\Phi(-\gamma)^{1-y_{i,j}}$, the distribution of two independent Bernoulli random variables conditioned on their sum being less than or equal to 1, with $m(\gamma)$ being the normalization constant. It is not hard to see its equivalence to our model via $\Phi^2(-\mu_0)=\Phi(-\gamma)/[2-\Phi(-\gamma)]$.

  On the other hand, a major advantage of our proposed approach over the alternative is that the augmented likelihood does not involve $m^{-1}(\cdot)$, which may become intractable as the number of constraints increases. Although there are a plethora of algorithms for dealing with intractable normalizing constants, such as exchange algorithms \citep{murray2006mcmc, moller2006efficient,liang2016adaptive} and pseudo-marginal Metropolis-Hastings \citep{andrieu2009pseudo,lee2019unbiased}, our model without intractable normalizing constant enjoys algorithmic simplicity and good scalability to relatively high dimension.

Our next task is to find out $T^{-1}(e)$, the appropriate set of input $\zeta_i$ for each possible outcome $e$. In the next two subsections, we first introduce a condition on $\mathcal P$, known as {\em integrality}, that leads to an equivalent solution between integer linear program and continuous linear program, then we describe how to find $T^{-1}(e)$ when $\mathcal P$ is integral.
\vspace*{-0.5cm}
\subsection{Integral polyhedron and image of integer linear program}

The feasible region of the integer linear program \eqref{eq:ilp} is an integer set. We now consider a relaxation that extends the feasible region to a polyhedron $\mathcal P=\{z\in \mathbb{R}^d: Az\le b\}$, which makes the problem a continuous linear program:
\[\label{eq:clp}
\max_{z\in \mathbb{R}^d} \zeta_i^\top z, \;  \text{subject to } z\in \mathcal P.
\]
A useful property of continuous linear program is the vertex optimality: if (i) $\mathcal P$ has at least one vertex ($z$ that cannot be written as $\lambda z_1 + (1-\lambda) z_2$ for any distinct $z_1,z_2\in \mathcal P$ and $\lambda\in (0,1)$), and (ii) $\max \zeta_i^\top  z <\infty$, then at least one of the optimal solutions is on a vertex.

When $\mathcal P$ has all vertices integer-valued, $\mathcal P$ is referred to as an integral polyhedron, or simply, that $\mathcal P$ is integral. Therefore, under the integrality of $\mathcal P$, the continuous linear program \eqref{eq:clp} has at least one solution being integer-valued. 
Since $\max_{z\in \mathbb{Z}^d \cap \mathcal P} \zeta_i^\top z \le \max_{z\in \mathbb{R}^d \cap \mathcal P}  \zeta_i^\top z$ becomes equality, the solution of \eqref{eq:ilp} belongs to the solution set of \eqref{eq:clp}.

We now show a stronger result: as long as $\mathcal P$ has more than one vertex, the solution of \eqref{eq:clp} is unique almost surely, hence the solution of \eqref{eq:clp} being guaranteed to be the same as the one of \eqref{eq:ilp} with probability one.
\begin{theorem}
\label{thm:unique_solution_as}
  When $\zeta_i$ is from a continuous distribution, for an integral $\mathcal P$ with more than one vertex, the solution to the continuous linear program \eqref{eq:clp} is integer-valued and unique almost surely.
\end{theorem}
\vspace*{-0.5cm}

We provide the proofs of theorems in the Supplementary Materials S1.
We will describe the sufficient conditions on $A$ and $b$ that guarantee the integrality of $\mathcal P$; for now, we assume integral $\mathcal P$ throughout the rest of the article. Figure \ref{plot:ilp_intuition}(a) illustrates how a linear program with integral polyhedron transforms continuous $\zeta_i$ into an integer-valued $y_i$.

The above uniqueness result also means that any non-vertex point in $\mathcal P$ is not optimal. Therefore, one should exercise caution that the range of $T$ may not cover some combinatorial outcome in $\mathbb{Z}^d\cap \mathcal P$. Fortunately, when $y_i$ is binary, we know any $z\in \{0,1\}^d: Az\le b$ must be a vertex (as it cannot be written as a nontrivial convex combination of two other binary points).
For ease of discussion and with some loss of generality, we reduce our focus to binary linear program $y_i\in \{0,1\}^d \cap \mathcal P$, such that 
\[\label{eq:binary_ilp}
 y_i = \underset{z}{\arg\max}\; \zeta_i^\top z, \;  \text{subject to } z\in [0,1]^d, A z\le b.
\]
In general, the integer linear program may lack a closed-form solution. Nevertheless, we show that given $y_i$, we can find the pre-image $T^{-1}(y_i)$ using the dual form.

\vspace*{-0.5cm}
\subsection{Dual thresholding representation}

For ease of notation, we now omit the subscript $i$ and use $\eta = \zeta_i$ and $\hat z=y_i$, and consider equivalent primal problem $\min_z (-\eta ^\top z)$ subject to $z\in [0,1]^d, Az\le b$.

Using $u\ge 0$, $v\ge 0$, $w\ge 0$, we set up the Lagrangian function
\(
L(z, u,v,w)& = -\eta^\top z + u^\top(Az-b) + v^\top(z-1) + w^\top(-z)\\
& = (-\eta^\top +u^\top A+v^\top-w^\top)z -u ^\top b-v^\top 1.
\)
Minimizing over $z$ leads to the dual problem:
\(
\min_{u,v,w} \; &  u^\top b+v^\top 1\\
\text{subject to } & -\eta+A^\top u+v=w\ge 0,\; u\ge 0, v\ge 0.
\)
It is not hard to see the optimal
$
\hat v =  (\eta-A^\top u)_+$, with $(x)_+=x$ if $x >0$ and $0$ otherwise, hence simplifying the dual loss to $ u^\top b+ 1^\top (\eta-A^\top u)_+$.

We use complementary slackness and strong duality to reveal a geometric characterization between $\eta$ and $\hat z$. First, on complementary slackness, we have three possible cases:
\begin{itemize}
	\item Case (i): when $\eta_j>A^{\top}_{.j}\hat u$, the optimal $\hat v_j=  (\eta_j-A^{\top}_{.j}\hat u )>0$, we have optimal $\hat w_j=0$ and $\hat z_j=1$;
	\item Case (ii): when  $\eta_j<A^{\top}_{.j}\hat u$, the optimal $\hat v_j=0$ and $\hat w_j= A^{\top}_{.j}\hat u-\eta_j>0$ hence $\hat  z_j=0$;
	\item Case (iii): when $\eta_j=A^{\top}_{.j}\hat u$, the optimal $\hat w_j=\hat v_j=0$. One may find $\hat z_j$ via the other complementary slackness conditions $\hat u_l(A \hat z- b)_l=0$ for $l=1\ldots,m$. 
\end{itemize}
Second, the linear program enjoys strong duality
\(
\max_{z\in [0,1]^d:Az\le b} \eta^\top z = \min_{u\ge 0} u^\top b+ 1^\top (\eta-A^\top u)_+,
\)
which means that a feasible $(z,u)$ achieves $\eta^\top z=u^\top b+1^\top (\eta-A^\top u)_+$ if and only if $(z,u)$ are the primal and dual optimal, respectively. 
Combining the above, we have a useful characterization between $\eta$ and $\hat u$.
\begin{theorem}
\label{thm:z_and_mu}
  For a given $\tilde z$ in the primal feasible region, 
  consider a set $C_{\tilde z}$ of $(\eta,u)$ with $\eta\in \mathbb{R}^d,u\ge 0$ that satisfy,
  \begin{itemize}
  \item $u^\top(A{\tilde z}-b)=0$,
  \item $\eta_j \ge A_{\cdot j}^\top u \text{ for } j: \tilde  z_j=1,$ and $\eta_j \le A_{\cdot j}^\top u \text{ for } j:\tilde  z_j=0,$
  \end{itemize}
Then $\tilde z$ is a solution of $\max_{z\in [0,1]^d:Az\le b}  \eta^\top z$ if and only if $(\eta,u)\in C_{\tilde z}$.
\end{theorem}
We now switch to notation $y_i$ and $\zeta_i$, and index $u$ as $u^{(i)}$. Theorem~\ref{thm:z_and_mu} gives a concrete characterization of the pre-image of the projection $T$ defined in (\ref{eq:T}). For any $y_i$ in the combinatorial space, the pre-image $T^{-1}(y_i)$ is given by
\(
\begin{aligned}
T^{-1}(y_i)&=\bigcup_{u^{(i)}\ge 0: u^{(i)}_k=0 \text{ if } (Ay_i-b)_k<0}\left\{\zeta_{i}\in\mathbb R^d: 
\zeta_{i,j} \ge A_{\cdot j}^\top u^{(i)} \text{ if } y_{ij}=1,\;
\zeta_{i,j} < A_{\cdot j}^\top u^{(i)} \text{ if } y_{i,j}=0
\right\}.
\end{aligned}
\)
Equivalently, we have a simple interpretation: there exists some $u^{(i)}\ge 0$ such that
\[
\label{eq:dual_thresholding}
y_{i,j} = 1(\zeta_{i,j} \ge A_{.j}^\top u^{(i)}).
\]
We refer to \eqref{eq:dual_thresholding} as the {\em dual thresholding representation}, and illustrate this in Figure \ref{plot:ilp_intuition}(b).
When there are no constraints on the binary vector $y_i$, we recover the zero thresholding $y_i=1(\zeta_i>0)$ as in \cite{albert1993bayesian}.

To be clear, the augmented likelihood $\mathcal L(y_i, \zeta_i ; \mu_i)$ describes the distribution of $y_i$ and $\zeta_i$, and the dual form \eqref{eq:dual_thresholding} should be understood as the existence of a separating set:
\(
\mathcal U(y_i,\zeta_i)=\{u:
& \zeta_{i,j} \ge A_{.j}^\top u^{(i)}: y_{i,j}=1,  
 \zeta_{i,j} \le A_{.j}^\top u^{(i)}: y_{i,j}=0,
\text{ for }j=1\ldots d,\\
& u_k=0 \text{ for }k:(Ay_i-b)_k<0,\; u_k\ge 0 \text{ for }k:(Ay_i-b)_k=0
\}.
\)

\begin{figure}[H]
    \centering
  \begin{subfigure}[t]{0.45\textwidth}
    \centering
      \begin{overpic}[width=\textwidth]{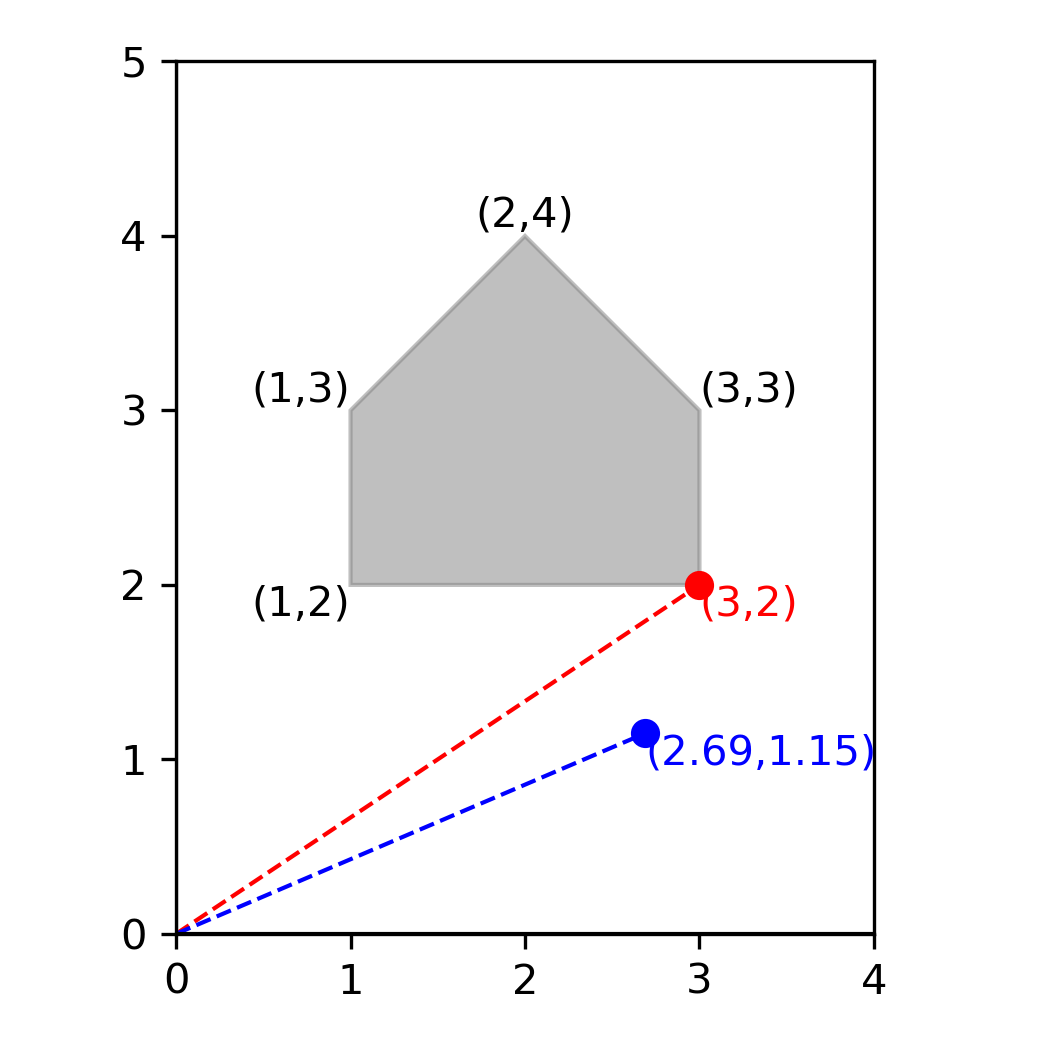}
      \end{overpic}
        \caption{Solving a continuous linear program in an integral polyhedron (gray, whose vertices are all integer-valued) is equivalent to transforming a continuous coefficient vector $\zeta_i$ (blue) into an integer coordinate $y_i$ (red).} 
  \end{subfigure}\qquad
    \begin{subfigure}[t]{0.45\textwidth}
      \centering
      \begin{overpic}[width=\textwidth]{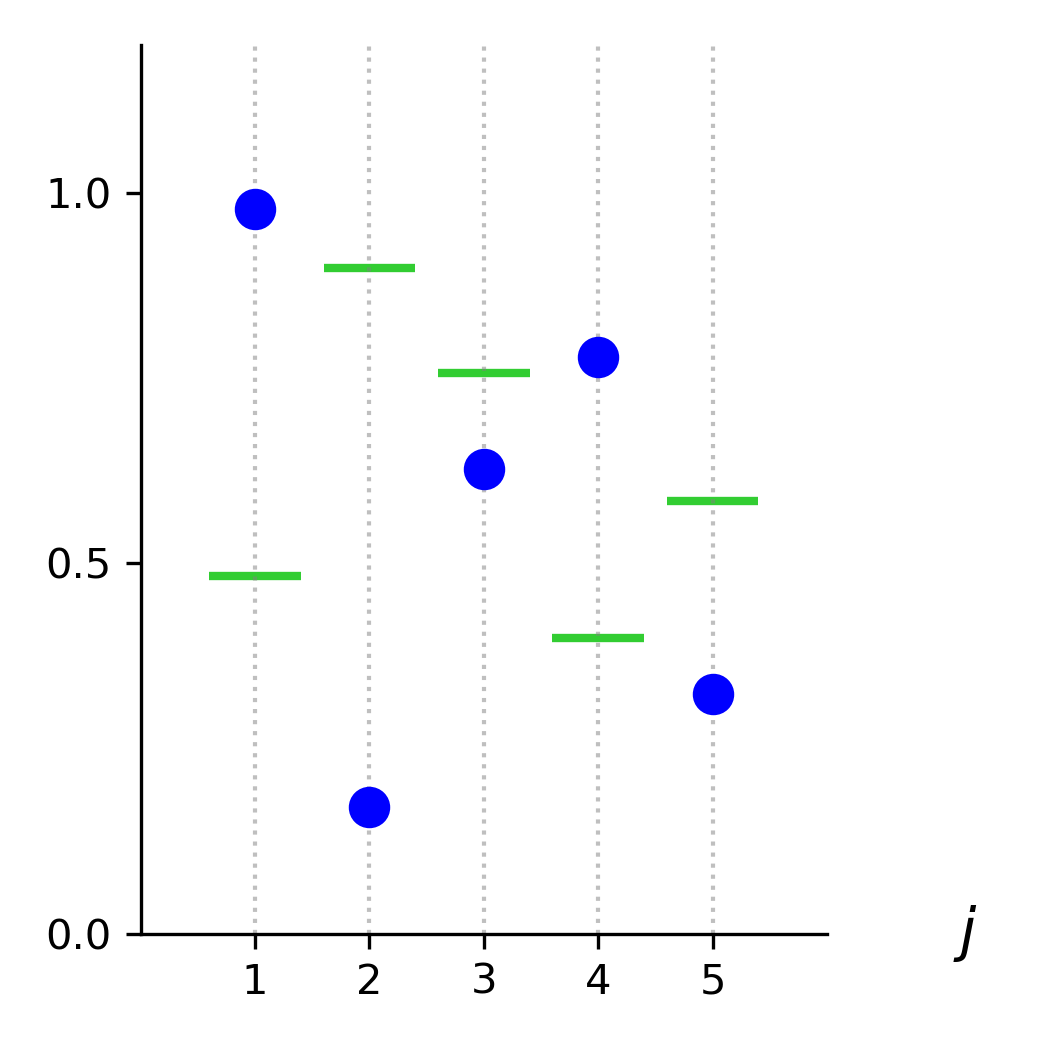}
      \end{overpic}
              \caption{When $y_i$ is binary, there is a dual thresholding representation corresponding to whether $\zeta_{i,j}$ (blue) is greater or equal to $A_{.j}'u^{(i)}$ (green bar), with $u^{(i)}$ the dual variable. In the figure, the associated $y_i$ is $(1,0,0,1,0)$.} 
  \end{subfigure}
              \caption{Illustration of using the linear program as a mapping (Panel a), and its equivalent representation when the output is restricted to be binary (Panel b).} 
  \label{plot:ilp_intuition}
  \end{figure}

That is, $
y_i = \underset{z\in \{0,1\}^d \cap \mathcal P}{\arg\max}\; \zeta_i^\top z\;
\Leftrightarrow \; \mathcal U(y_i,\zeta_i)\neq \varnothing.
$
To enforce $\mathcal U(y_i,\zeta_i)\neq \varnothing$ in the posterior computation, we use another latent variable $u^{(i)}\sim \Pi_u(\cdot\mid y_i,\zeta_i)$ supported on $\mathcal U(y_i,\zeta_i)$:
\[\label{eq:aug_lik2}
\mathcal L(y_i, \zeta_i ; \mu_i)&\Pi_u(u^{(i)}\mid y_i,\zeta_i) =  \mathcal F(\zeta_i;\mu_i) \prod_{j=1}^{d}\bigg\{
  1(y_{i,j}=1,\zeta_{i,j} > A_{.j}^\top u^{(i)})  +
  1(y_{i,j}=0,\zeta_{i,j} < A_{.j}^\top u^{(i)}) \bigg\}\\
&\cdot  1[ u^{(i)}\in \mathcal U(y_i,\zeta_i), \mathcal U(y_i,\zeta_i)\neq  \varnothing] \; \cdot\; \Pi_u(u^{(i)}\mid y_i,\zeta_i),
\]
where $\Pi_u(u^{(i)}\mid y_i,\zeta_i)$ is an arbitrary proper density over $\mathcal U(y_i,\zeta_i)$, and we have replaced $\ge$ by $>$ in $\mathcal U(y_i,\zeta_i)$ since hence the equality event has zero probability when $\zeta_i$ and $u^{(i)}$ are continuous latent variables.

\subsection{Applicability of data augmentation}\label{sec:theory}
\vspace*{-0.5cm}
For our data augmentation to work, one key condition is that the polyhedron $\mathcal P=\{z\in \mathbb{R}^d: Az\le b\}$ is integral. Now we present a few useful results from the combinatorial optimization literature for checking this condition, and we illustrate the condition with a few examples commonly seen in applications. The results in this section apply to general integer $z$ that is not necessarily binary; whereas the binary setting can be viewed as having constraint $0\le z\le 1$ in addition to $\tilde Az\le \tilde b$, so we view $A=[\tilde A , I , -I]$ and $b=[\tilde b, 1, 0]$.

A useful sufficient (though not necessary) condition for a polyhedron to be integral is total unimodularity. A matrix $A$ is said to be totally unimodular (TUM) if every square submatrix of $A$ has determinant in $\{-1, 0, 1\}$. Then we have the following \citep{schrijver1998theory}:
\begin{theorem}\label{thm:TUM}
  If $A$ is TUM, then for any $b\in \mathbb{Z}^m$, the polyhedron $\mathcal P=\{z\in \mathbb{R}^d: Az\le b\}$ is an integral polyhedron.
\end{theorem}\vspace*{-0.5cm}
The converse holds as well \citep{hoffman2010integral}, that if $\mathcal P$ is integral for all $b\in \mathbb{Z}^m$, then $A$ is TUM. It is not hard to see that $A_{i,j}\in \{-1,0,1\}$. The total unimodularity can be easily checked via software, such as \texttt{lintools} package in R. On the other hand, we list a few detailed properties that could be useful for model extension.

\begin{enumerate}
  \item If $A$ is TUM, then both $-A$ and $A^\top$ are TUM.
  \item If $A$ is TUM, then $[A, I]$ is TUM.
  \item If $A$ is TUM, then $[A , -A]$ is TUM.
  \item If $A_{i,j}\in \{-1,0,1\}$, each column of $A$ has at most one -1 and at most one 1, then $A$ is TUM.
  \item  If $A_{i,j}\in \{-1,0,1\}$, each column of $A$ has at most two non-zeros, and there exists a row index subset $S\subseteq [m]$: $\sum_{i\in S} A_{i,j} = \sum_{i\not\in S} A_{i,j}$ for all $j$, then $A$ is TUM.
\end{enumerate}
The proof can be found in \cite{wolsey1999integer}.
We give a few useful examples.

\noindent\textbf{Example 1} (Binary constraint) In our focused binary linear program $z\in [0,1]^d$, using Properties 1 and 2, and the fact that multiplying a row of square submatrix by $-1$ keeps the determinant in $\{-1,0,1\}$, we see $A=[\tilde A, I, -I]$ is TUM if $\tilde A$ is TUM. 

\noindent\textbf{Example 2} (Equality) For equality constraints $\tilde A z =b$ with a TUM $\tilde A$, we can represent equality by $[\tilde A, -\tilde A] z\le [b,-b]$, hence we see the polyhedron is integral via Property 3. 

\noindent\textbf{Example 3} (Partial ordering)
For partial order constraints $z_j\le z_k$, for a set of pair indices $(j,k)$, we can represent them by $A_{i,j}=1$, $A_{i,k}=-1$ and $A_{i,l}=0$ $l\neq j$ or $k$, with $b=0$, hence Property 4 leads to $A$ being TUM.

\noindent\textbf{Example 4} (Matching)
Suppose we have a bipartite graph $G=(V, E_G)$, with $V$ the node set divided into two subsets $V_1$ and $V_2$, and $E_G$ the edges formed only across $V_1$ and $V_2$. A matching is a subgraph $M=(V,E_M)$ with $E_M\subseteq E_G$, such that each node has zero or one edge connecting to another node. We can use $z\in \{0,1\}^{|E_G|}$ to represent if the edge of $E_G$ is chosen in $E_M$. Using a matrix $\tilde A$ for $G$, with $\tilde A_{v,e}=1$ if the $e$-th edge in $E_G$ is incident to node $v$, otherwise $\tilde A_{v,e}=0$, we see that $\tilde A$ is TUM using Property 5 with $S=V_1$. The matching representation $z$ is associated with a TUM  $A=[\tilde A,  -I]$ and integer $b=[1,0]$.

The above could be extended to perfect matching where each node has exactly one edge, or $\tilde b$-matching where each node's edge upper bound a varying integer.

\noindent\textbf{Example 5} (Flow) Consider a flow network with $V$ the node set and $E$ the directed edge set, with $s$ the source node and $t$ the sink node. The flow on the network is represented by $z_{e}\ge 0$ for directed edge $e\in E$, subject to capacity constraint $z_e \le \tilde  b_e$, and flow conservation for each node $i\in V\setminus \{s,t\}$, $\sum_{e\in E: e=(i \to j)} z_e = \sum_{e\in E: e=(k \to i)} z_e$. We see that the incidence matrix $\tilde A$ for directed graph is TUM, where $\tilde A_{i,e}=1$, $\tilde A_{j,e}=-1$ for each $e=(i\to j)$, and $A_{l,e}=0$ for $l\neq i$ or $j$. The flow network representation vector $z$ is associated with a TUM matrix $A=[\tilde A, -\tilde A, I, -I]$ and $b=[0,0,\tilde b,0]$. If we have all capacity constraints $\tilde b\in \mathbb{Z}^{|E|}$, then we know the produced flow under a linear program is integer-valued.

We want to clarify that TUM is just one of the useful sufficient conditions to verify the integrality of a polyhedron, and there are other conditions often used in practice. For example, the total dual integrality condition and Edmonds-Giles' theorem could be used \citep{edmonds1977min}, which may apply even when $A$ is not TUM. We refer the readers to \cite{schrijver1998theory} for a comprehensive review.

\vspace*{-0.5cm}
\section{Sampling algorithms for Bayesian regression with combinatorial response}\label{sec:algorithm}

We now focus on a Bayesian regression model with combinatorial response $y_i\in \{0,1\}^d$ and linear predictor $\mu_i = \beta x_i$, for $x_i\in\mathbb{R}^p$ and $\beta\in\mathbb{R}^{d\times p}$. 
We assign a matrix-normal prior $\beta\sim \text{Mat-N}(0, I_d, I_p\tau)$ with $\tau>0$ a hyper-parameter, and use $\zeta_i\sim \text{N}(\mu_i, I_d)$. We have the augmented likelihood independent for $i=1,\ldots, n$:
\(
  \mathcal L(y_i, \zeta_i;\beta x_i)  \propto &\prod_{j=1}^{d}\bigg\{
  1(y_{i,j}=1,\zeta_{i,j} >  (UA)_{i,j})  +
  1(y_{i,j}=0,\zeta_{i,j} <  (UA)_{i,j}) \bigg\}\\
  & \times \exp( - 0.5 \|\zeta_i-\beta x_i\|^2_2) 1[ u^{(i)}\in \mathcal U(y_i,\zeta_i), \mathcal U(y_i,\zeta_i)\neq  \varnothing],
\)
For simpler notation, matrix $U\in\mathbb R^{n\times m}$ has $(i,j)$-th entry $u^{(i)}_j$, hence $(UA)_{i,j} = A^{\top}_{.j}u^{(i)}$. 

We now present a Metropolis-Hastings-within-Gibbs sampler. Continuing from our discussion of \eqref{eq:aug_lik2}, we now view $u^{(i)}$ as a random variable and specify $\Pi_u(u^{(i)}\mid y_i,\zeta_i)$. We assign exponential kernel for those $u^{(i)}_k>0$, leading to
\[\label{eq:cond_u}
\Pi_u(u^{(i)}\mid y_i,\zeta_i) = \frac{\prod_{k:(Ay_i-b)_k=0}\exp(-{u^{(i)}_k})}{c[\mathcal U(y_i,\zeta_i)]},
\]
where $c[\mathcal U(y_i,\zeta_i)]$ is the normalizing constant. Since $c[\mathcal U(y_i,\zeta_i)]$ depends on $\zeta_i$, whose value is likely intractable, we propose an accept--reject step to bypass its computation. For two sets $\mathcal U(y_i,\zeta_i)$ and $\mathcal U(y_i,\zeta^*_i)$,  consider a third set $\mathcal U(y_i,\tilde\zeta_i)$, with $\tilde\zeta_{i,j} = \max(\zeta_{i,j}, \zeta^*_{i,j})$ if $y_{i,i}=1$ and $\tilde\zeta_{i,j} = \min(\zeta_{i,j}, \zeta^*_{i,j})$ if $y_{i,i}=0$. It is not hard to see both $\mathcal U(y_i,\zeta_i)$ and $\mathcal U(y_i,\zeta^*_i)$ are subsets of $\mathcal U(y_i,\tilde\zeta_i)$.

Consider the event of drawing $u^{(i)*} \sim \Pi_u(\cdot\mid y_i,\tilde\zeta_{i})$ in $\mathcal U(y_i,\tilde\zeta_i)$, and having $u^{(i)*}$ falling in the polyhedron $\mathcal U(y_i,\zeta_i)$, it has probability:
\(
P[u^{(i)*} \in \mathcal U(y_i,\zeta_i)] = \frac{
c[\mathcal U(y_i,\zeta_i)]
}{
c[\mathcal U(y_i,\tilde\zeta_i)]
}.
\)
Therefore, having an additional acceptance step using the above probability is essential in order to cancel out $c[\mathcal U(y_i,\zeta_i)]$ in the acceptance ratio.
{Similar scenarios of dealing with intractable normalizing constants were considered by \cite{moller2006efficient,lee2011auxiliary}.} {
In theory, any proper distribution having support $\mathcal U(y,\zeta)$ is a valid distribution choice for $u^{(i)}$. In the Supplementary Materials S6, we show a simulation using alternative half-Gaussian kernel for $u^{(i)}$, and find no empirical difference from exponential in mixing performance.
}

Each MCMC iteration of our algorithm consists of the following steps:\vspace*{-0.5cm}
\begin{itemize}
  \item Update $\zeta_i$ for each $i$:
\begin{enumerate}\vspace*{-0.5cm}
	\item Generate $u^{(i)}$ from $\Pi(u^{(i)}\mid y_i,\zeta_i)$ as in \eqref{eq:cond_u}. 
Propose  $\zeta_i^*$ based on truncated normal over elements  $j\in \{1,\ldots,d\}$:
\(
\zeta^*_{i,j}\sim 
\left\{
\begin{array}{ll}
  \text{TN}( \beta_j^\top x_i, 1)1[\cdot > (UA)_{i,j}], & \text{ if } y_{i,j}=1;\\
  \text{TN}( \beta_j^\top x_i, 1)1[\cdot < (UA)_{i,j}], & \text{ if }  y_{i,j}=0.
\end{array}
\right.
\)
\item Based on $\mathcal U(y_i,\zeta_i)$ and $\mathcal U(y_i,\zeta^*_i)$, find $\mathcal U(y_i,\tilde\zeta_i)$, then draw another $u^{(i)*}$ from $\Pi(u^{(i)}\mid y_i,\tilde\zeta_i)$.
\item If  $u^{(i)*} \in \mathcal U(y_i,\zeta_i)$, accept $\zeta^*_i$ as the updated value of $\zeta_i$. Otherwise, keep the current value of $\zeta_i$.
\end{enumerate}
\item Update $\beta$ from matrix normal (where $X\in\mathbb{R}^{p\times n}$, $\zeta\in\mathbb{R}^{n\times d}$)
\(
(\beta \mid -) \sim \text{Mat-N}[ \zeta^\top X^\top(XX^\top+I_p/\tau)^{-1}, I_d,  (XX^\top+I_p/\tau)^{-1} ].
\)
\end{itemize}
To explain the algorithm, we examine the first step as a Metropolis-Hastings step. On the transition from $\zeta_i$ to $\zeta_i^*$, we have the joint probability kernel of the current $\zeta_i$ and proposed $\zeta^*_i$, marginalized over $u^{(i)}$:
\(
&  \mathcal L(y_i, \zeta_i ;\beta x_i) 
  \int   \frac{\prod_{k:(Ay_i-b)_k=0}e^{-{u^{(i)}_k}}}{c[\mathcal U(y_i,\zeta_i)]} { 1(u^{(i)}\in \mathcal U_{y_i,\zeta_i})} \phi_{TN}(\zeta^*_i\mid u^{(i)})1(u^{(i)}\in \mathcal U_{y_i,\zeta^*_i}) d u^{(i)} \frac{
c[\mathcal U(y_i,\zeta_i)]
}{
c[\mathcal U(y_i,\tilde\zeta_i)]
}\\
& \propto  
\exp( - 0.5 \|\zeta_i-\beta x_i\|^2_2)
\exp( - 0.5 \|\zeta^*_i-\beta x_i\|^2_2)\\
&\qquad\times
\int  
\bigg (
\prod_{j=1}^d   \bigg\{\Phi[(UA)_{i,j}- \beta_j^\top x_i]1(y_{i,j}=0)+  \Phi[\beta_j^\top x_i-(UA)_{i,j}]1(y_{i,j}=1)\bigg\}^{-1}\\
& \qquad \times \big[\prod_{k:(Ay_i-b)_k=0}\exp(-{u^{(i)}_k})\big]1(u^{(i)}\in \mathcal U_{y_i,\zeta^*_i} )1(u^{(i)}\in \mathcal U_{y_i,\zeta_i})\bigg)  d u^{(i)} \frac{1}{c[\mathcal U(y_i,\tilde \zeta_i)]},
\)
where we use $\phi_{TN}$ to denote the truncated normal density.
Since $\mathcal U(y_i,\tilde \zeta_i)$ is invariant to the ordering between $\zeta_i$ and $\zeta^*_i$, we see the backward joint probability kernel from $\zeta^*_i$ to $\zeta_i$ is exactly the same. Therefore, the Metropolis-Hastings acceptance ratio is one.

{

\begin{remark}
At a relatively high dimension $d$, the large difference between $\mathcal U(y_i,\zeta_i)$ and $\mathcal U(y_i,\tilde\zeta_i)$ may cause
$P[u^{(i)*} \in \mathcal U(y_i,\zeta_i)]\approx 0$. To avoid such an issue, one can update a subvector of $\zeta_i$ for all $i$ in each iteration. Empirically, we find that updating a subvector of length $\le 100$ keeps $P[u^{(i)*} \in \mathcal U(y_i,\zeta_i)]$ away from zero. For high dimensional problems, the above algorithm can lead to slow mixing performance; therefore, we present an alternative No-U-Turn Hamiltonian Monte Carlo algorithm in the Supplementary Materials S5.
\end{remark}

For drawing $u^{(i)}$ from \eqref{eq:cond_u}, when $u^{(i)}$ is low-dimensional, we can use rejection sampling. When $u^{(i)*}$ is relatively high-dimensional, we use the hit-and-run algorithm \citep{smith1984efficient}. Specifically, with $u^{(i)}$ initialized at $\omega\in \mathcal U(y_i,\zeta_i)$: (i) draw a uniform direction $v:\|v\|_2=1$, subject to {$\sum_{k:(Ay_i-b)_k=0}v_k>0$ and} $v_k=0: (A y_i-b)_k<0$; (ii) find the longest interval $(\tilde a,\tilde b)$ such that $\omega+ \alpha v\in \mathcal U(y_i,\zeta_i)$ for all $\alpha\in(\tilde a,\tilde b)$, then draw $\alpha^* \sim \text{Exp}_{(\tilde a,\tilde b)}( \sum_{k:(Ay_i-b)_k=0} v_k)$, which is a truncated exponential distribution, with rate $\sum_{k:(Ay_i-b)_k=0} v_k$ and supported in $(\tilde a,\tilde b)$; (iii) update $\omega$ to $\omega+  \alpha^* v$. We repeat (i)-(iii) 100 times. The hit-and-run algorithm is efficient and shown to mix rapidly for drawing samples in a polyhedron \citep{lovasz2006fast,chen2022hit}. In case one has an unbounded polyhedron, one can truncate $\tilde a$ and $\tilde b$ at a large magnitude such as $10^5$. We see that \eqref{eq:cond_u} is always proper regardless if $\mathcal U(y_i,\zeta_i)$ is bounded or not.

\begin{remark}
When hit-and-run is used to draw $u^{(i)}$, 
we effectively draw from an approximate distribution hence making
 the algorithm approximate MCMC \citep{johndrow2015optimal}. 
\end{remark}

In the Supplementary Materials S5, we present two alternative samplers, the Hamiltonian Monte Carlo (HMC) sampler and pseudo-marginal sampler, and show their computational performance in simulation.
}

\vspace*{-0.5cm}

\section{Consistency theory}
This section develops posterior consistency theory for Bayesian inference on combinatorial response models. We establish conditions under which the posterior distribution concentrates around the true parameter as sample size grows. We consider a setting in which the response variables $Y_i$ take values in the vertex set $\mathcal E\subseteq \{0,1\}^d$ of an integral polytope
$\mathcal P= \{ z \in [0,1]^d : Az \le b \}$, with $d$ fixed and finite, while the sample size $n$ diverges. We focus on the following two models:
    \[\label{eq:intercept_only}\text{Intercept-only model: }
    Y_i = T(\zeta_i), \quad \zeta_i \mid \mu \stackrel{i.i.d.}{\sim} N(\mu, I_d), \quad \mu \sim \pi, \quad \forall i = 1, \ldots, n.
    \]
    \[\label{eq:regression}\text{Regression model: }
    Y_i = T(\zeta_i), \quad \zeta_i \mid \beta \stackrel{indep}{\sim} N(\beta x_i, I_d), \quad \beta \sim \pi, \quad \forall i = 1, \ldots, n.
    \]
where $T:\mathbb R^d\setminus C\to\mathcal E$ is the argmax mapping $T(\zeta)=\arg\max_{e \in \mathcal E}\, \zeta^\top e$, and $\pi$ represents the prior distribution on either $\mu$ or $\beta$, depending on the model under consideration.

To simplify notation, we denote by $Y^n := (Y_1, \ldots, Y_n)$ the response data and by $P^n_{\mu}$ the joint distribution of $Y^n$ given $\mu$, with 
\(
P_{\mu}(e):=P(Y=e|\mu)=\int_{\mathbb R^d}\phi(\zeta;\mu,I_d)1(\zeta^\top e=\max_{e'\in\mathcal E}\zeta^\top e')d\zeta,\quad\forall e\in\mathcal E.
\)

\subsection{Inconsistency of unconstrained model}

For combinatorial responses, it might appear tempting to simply ignore the constraints and model each data point as an array of independent discrete outcomes. However, we show this model will lead to inconsistent estimates.

We illustrate the inconsistency using intercept-only models. Suppose the ground-truth data generating distribution is $P^0(Y=e)=p^0_e$ for each combinatorial outcome $e\in \{0,1\}^d \cap \mathcal P$, and we use $q_{j1}= \sum_{\text{all }e} p^0_e 1(e_j=1)$ to denote the marginal $P^0(Y_j=1).$ We consider the case that each $q_{j1}\in (0,1)$, so that both 1 and 0 can appear in $j$-th dimension with positive probability.

Consider an unconstrained model that ignores the constraint $Az \le b$, and hence models the responses as $y_{ij} \stackrel{i.i.d.}{\sim} \text{Bernoulli}[\Phi(\mu_j)]$, with a prior $\mu \sim \pi(\mu)$. The resulting likelihood-prior product is
\[
\label{eq:unconstrained_model}
\exp \left[ n \sum_{j=1}^{d} \left\{ \frac{n_{j1}}{n} \log \Phi(\mu_j) + \frac{n - n_{j1}}{n} \log \Phi(-\mu_j) \right\} + \log \pi(\mu) \right],
\]
where $n_{j1} = \sum_{i=1}^n \mathbf{1}(y_{ij} = 1)$ denotes the number of ones in the $j$-th coordinate. Assuming the prior $\pi(\mu)$ has full support on $\mathbb{R}^d$ and is invariant to $n$, the log-likelihood term dominates as $n \to \infty$. Consequently, the posterior for $\Phi(\mu_j)$ converges almost surely to a point mass: $\Pi(\Phi(\mu_j) \mid y_1, \ldots, y_n) \stackrel{a.s.}{\to} \delta_{q_{j1}}$, where $\delta$ denotes a Dirac measure; this is analogous to the way that the maximum likelihood estimator (MLE) $\hat{\Phi}(\mu_j)=n_{j1}/n$.

Consequently, the fitted value under the unconstrained model has:
\(
\tilde P(Y=e)  \stackrel{a.s.}{\to} \prod_{j=1}^{d} q_{j1}^{e_j} (1- q_{j1})^{1-e_j}.
\)
Note that the limiting law of $Y$ above is not guaranteed to be the same as $p^0_e$. Perhaps the most striking difference is that for a value outside the combinatorial space, $e_1 \in \{0,1\}^d \setminus \mathcal P$, we have $\tilde P(Y=e_1) > 0$, but the ground-truth is $P^0(Y=e_1)=0$.
For the example in Section \ref{subsec:method}, if we have the ground truth $\mu_0=0$, then
$
{n_{11}}/{n}\stackrel{a.s.}{\to}{3}/{8},\quad {n_{21}}/{n}\stackrel{a.s.}{\to}{3}/{8},
$
hence 
$
\tilde P[Y=(1,1)]\stackrel{a.s.}{\to} 9/64,
$ but we know $Y=(1,1)$ cannot happen due to the constraint.

\begin{remark}
In the above  exposition, we focus on the unconstrained models with conditional independence across the dimensions of $y_i$. As an alternative, an unconstrained model could potentially avoid the inconsistency issue if one considers a latent correlated $\tilde\zeta_{i}\sim \text{N}(\mu, \Omega^{-1})$ and $y_{i,j}=1(\tilde\zeta_{i,j}>0)$, where $\Omega$ characterizes the conditional dependence structure of $y_i$. On the other hand, a caveat of such a latent Gaussian model is that ignoring the combinatorial constraint could lead to misleading conditional dependence estimate, for two outcomes that are marginally independent. For example, consider three binary outcomes, $y_{i,1}{\sim} \text{Bernoulli}(0.5)$, $y_{i,3}{\sim} \text{Bernoulli}(0.5)$, $y_{i,1}$ and $y_{i,3}$ are independent; $y_{i,2}$ is subject to combinatorial constraint $y_{i,2}=0$ if $y_{i,1}=0$, otherwise $y_{i,2} \sim \text{Bernoulli}[(y_{i,3}+1)/3]$. Now $y_{i,1}$ and $y_{i,3}$ become conditionally dependent given $y_{i,2}$ due to the combinatorial constraint.
\end{remark}

\subsection{Consistency under linear program link}
We now discuss the consistency results. We focus on the scenario when the data generating distribution is \eqref{eq:intercept_only} or \eqref{eq:regression}, with ground truth and fixed $\mu_0$ or $\beta_0$. Our exposition progresses over three stages: (i) the consistent recovery of the law $P^0(Y = e)$ under the intercept-only model, regardless of whether we can recover $\mu_0$; (ii) the consistent recovery of $\mu_0$, under additional identifiability conditions for the intercept-only model; (iii) the consistent recovery of $\beta_0$ in the regression model. Although (ii) may be viewed as a special case of (iii), establishing posterior consistency in regression is much more challenging than in the intercept-only model because the data are no longer identically distributed. Therefore, we keep (ii) and (iii) as two separate results.

\subsubsection{Recovering probability law via intercept-only model}

We first focus on recovering the groud-truth law of $Y$, $P_{\mu_0}$, under an intercept-only generative distribution.
For recovering the above, there is not much requirement needed for posterior consistency other than prior support, as described below. 
\begin{assumption}[Prior support]
\label{ass:prior}
The prior $\pi$ is a probability measure on the parameter space 
($\mathbb{R}^d$ for the intercept-only model or $\mathbb{R}^{d \times p}$ 
for the regression model) satisfying $\pi(B_\varepsilon(\theta_0)) > 0$ 
for every $\varepsilon > 0$, where $\theta_0$ denotes the true parameter.
\end{assumption}

Condition~\ref{ass:prior} ensures that the prior places positive mass in every Kullback--Leibler (KL) neighborhood of the true parameter. This is the usual KL support requirement for posterior consistency \citep{schwartz1965onbayes} and is satisfied by commonly used continuous priors.

\begin{theorem}
\label{thm:law_consistency}
Under Condition~\ref{ass:prior}, the posterior of the intercept-only model \eqref{eq:intercept_only} is strongly consistent at estimating $P_{\mu_0}$. That is,
\(
\left\|\int P(Y\mid \mu) \Pi(\mu \mid Y^n) d\mu-P(Y\mid\mu_0)\right\|_1\stackrel{n\to\infty}{\to}0\quad P^\infty_{\mu_0}-\text{almost surely}.
\)
\end{theorem}

\subsubsection{Recovering coefficients in intercept-only model}

	We want to clarify that the consistent recovery for $P_{\mu_0}$ does not imply the recovery for $\mu_0$. An identifiability issue may arise, for example, in a model where the response is a one-hot vector determined by the Gaussian latent variable that attains the maximum value:
 \(
y_i = \arg\max_z \zeta_i^\top  z \quad \text{subject to } z \in \{0,1\}^d \text{ and } \sum_{j=1}^d z_j = 1,
\)	
	shifting all coordinates of $\mu$ by a constant scalar does not change the probability distribution of the outcomes.
	For such cases, a common strategy to ensure identifiability in generalized linear model is to impose equality constraint on the parameter $\mu$, such as setting $\mu_1=0$ --- equivalently, one assigns a degenerate prior $\pi(\mu)$.

On the other hand, there are cases under which $\mu$ is identifiable without the need of degenerate prior. It turns out that the geometry of the polytope $\mathcal P$ plays a key role in ensuring identifiability of $\mu$. In fact, we show that the non-degeneracy of the polytope $\mathcal P$ is equivalent to the identifiability of $\mu$ in Theorem~\ref{thm:mu_identifiability}.

\begin{assumption}[Polytope structure]
\label{ass:polytope}
The polytope $\mathcal P=\{z\in\mathbb R^d:Az\le b\} \subset \mathbb R^d$ has nonempty interior in $\mathbb R^d$.
\end{assumption}

\begin{theorem}[Intercept-only model identifiability]
\label{thm:mu_identifiability}
Consider the intercept-only model \eqref{eq:intercept_only}. For any $\mu,\mu'\in\mathbb R^d$,
\[
P_{\mu}=P_{\mu'} 
\quad\text{if and only if}\quad 
\mu-\mu'\in \text{Span}\{e-e':e,e'\in\mathcal E\}^{\perp}.
\]
Consequently, the model \eqref{eq:intercept_only} is identifiable (i.e., $P_\mu$ is injective in $\mu$) if and only if Condition~\ref{ass:polytope} holds. When Condition~\ref{ass:polytope} fails, the model is identifiable only up to shifts in directions orthogonal to $\text{Span}\{e-e':e,e'\in\mathcal E\}$.
\end{theorem}

 A commonly used approach for establishing posterior consistency is Schwartz's Theorem \citep{schwartz1965onbayes}, which involves the construction of an exponentially consistent sequence of tests as introduced below. 
\begin{definition}
\label{def:exp_consistent_test_mu}
Let $V^c\subset\mathbb R^{d}$ and $\mu_0\in\mathbb R^d$. A sequence of tests $\psi_n(Y^n):\mathcal E^n\to[0,1]$ is said to be \textit{exponentially consistent} for testing
$
H_0:\mu=\mu_0\text{ vs }H_1:\mu\in V^c,
$
if there exists $C_1, C_2, b>0$ such that
\begin{enumerate}
\item[\rm(a)] $\mathbb E_{P^n_{\mu_0}}[\psi_n]\le C_1e^{-nb}$;
\item[\rm(b)] $\inf_{\mu\in V^c}\mathbb E_{P^n_\mu}[\psi_n]\ge 1-C_2e^{-nb}$.
\end{enumerate}
\end{definition}

With these ingredients, we now establish the posterior consistency of the intercept-only model for recovering $\mu_0 \in \mathbb{R}^d$.
\begin{theorem}[Posterior consistency for the intercept-only model]
\label{thm:consistency_mu}
Under Conditions~\ref{ass:prior} and \ref{ass:polytope}, the posterior of the  model \eqref{eq:intercept_only} is consistent at estimating $\mu_0$, i.e., for any open neighborhood of $\mu_0$, denoted by $V_{\mu_0}$, the following holds:
\(
\Pi( V^c_{\mu_0}\mid Y^n)\stackrel{n\to\infty}{\to}0\quad P^\infty_{\mu_0}-\text{almost surely}.
\)
\end{theorem}

\subsubsection{Recovering coefficients in regression model}
We now establish posterior consistency for the regression model \eqref{eq:regression}.
Following \cite{amewou2003posterior} and \cite{sriram2013posterior}, we treat the covariates $x_i$ as fixed.
A key distinction from the intercept-only model is that the observations $Y_1,\ldots,Y_n$ are independent but not identically distributed, since the distribution of each $Y_i$ depends on the fixed covariate $x_i$. 
Let $Q_{i,\beta}$ denote the conditional law of $Y_i$ given $\beta x_i$. Let $Q^n_\beta$ denote the conditional law of $Y^n=(Y_1,\ldots,Y_n)$ given $\{\beta x_i\}_{i=1}^n$. We have $Q^n_\beta=\prod_{i=1}^n Q_{i,\beta}$.
We have the posterior distribution of $\beta$ given $Y^n$, for any $\pi$-measuable set $\mathcal{B} \subset \mathbb{R}^{d \times p}$:
\(
\Pi(\mathcal B|Y_1,\ldots,Y_n)=\frac{\int_{\mathcal B}Q^n_\beta(Y_1,\ldots,Y_n)d\pi(\beta)}{\int_{\mathbb R^{p\times d}}Q^n_\beta(Y_1,\ldots,Y_n)d\pi(\beta)}.
\)
For two discrete probability distributions $G$ and $G'$ on a common support $\mathcal E_0$, we define the KL divergence
$
\mathrm{KL}(G \,\|\, G')
:=
\sum_{e\in\mathcal E_0} G(e)\log\!\left(\frac{G(e)}{G'(e)}\right),
$
and the associated second-moment quantity
$
V(F,G)
:=
\sum_{e\in\mathcal E_0}
G(e)\Bigl[\,0 \vee \log\!\bigl(G(e)/G'(e)\bigr)\Bigr]^2 .
$

With the above notations, we now employ an adapted version of Schwartz's Theorem showed by \cite{amewou2003posterior}. For this article to be self-contained, we first state the theorem below using our notations.
\begin{theorem}\citep{amewou2003posterior}
\label{thm:adapt_sch}
 Suppose the following conditions hold:
\begin{enumerate}
    \item[\rm(i)] There exists an exponentially consistent sequence of tests $\psi_n(Y^n) : \mathcal E^n \to [0,1]$ for testing
    $
    H_0: \beta = \beta_0 \text{ vs. } H_1: \beta \in V^c_{\beta_0}.
    $
That is, there exist constants $C_1, C_2, b > 0$:
\begin{enumerate}
    \item[\rm(a)] $\mathbb{E}_{Q^n_{\beta_0}}[\psi_n] \leq C_1 e^{-nb}$;
    \item[\rm(b)] $\inf_{\beta \in V^c_{\beta_0}} \mathbb{E}_{Q^n_\beta}[\psi_n] \geq 1 - C_2 e^{-nb}$.
\end{enumerate}
    
    \item[\rm(ii)] For every $\delta_0 > 0$,
    \(
    \pi\left(\left\{\beta : \sup_i KL(Q_{i,\beta_0} \| Q_{i,\beta}) < \delta_0, \sum_{i=1}^\infty \frac{V(Q_{i,\beta_0}, Q_{i,\beta})}{i^2} < \infty \right\}\right) > 0.
    \)
\end{enumerate}
Then,
$
\Pi(V^c_{\beta_0} \mid Y^n) \stackrel{n\to\infty}{\to} 0\quad Q^\infty_{\beta_0}-\text{almost surely}.
$
\end{theorem}

Condition~(ii) in Theorem~\ref{thm:adapt_sch} is the adapted KL-support condition for the non-i.i.d.\ case, which requires Condition~\ref{ass:prior}, and the boundedness of covariates (Condition~\ref{ass:covariates_bound}). The latter condition is ubiquitously assumed in establishing posterior consistency for regression models, and is needed to bound the supremum of the Kullback-Leibler divergence.

\begin{assumption}[Boundedness of Covariates]
\label{ass:covariates_bound}
There exists $\ell>0$ such that $\sup_i\|x_i\|_2\le \ell$.
\end{assumption}

To distinguish $\beta\neq\beta_0$ from $\beta_0$, the covariates must satisfy additional directional condition. We introduce the following definition. Then we present the additional condition on the covariates.

\begin{definition}[Face selection mapping]
For a polyhedron $\mathcal P$, a face is a subset $F= \{z\in \mathcal P : v^\top z = \max_{z\in \mathcal P} v^\top z\}$ for some $v\in \mathbb R^d$.  A \textit{proper face} is one that is not equal to $\mathcal P$. 
We define \textit{face selection} mapping as:
\(
\mathcal F(v):=\{e^\star\in\mathcal P:v^\top e^\star=\max_{e\in\mathcal P}v^\top e\}.
\)
\end{definition}
\noindent We note that a face can be an empty set, a vertex, a line segment, or a higher-dimensional subset of $\mathcal P$. Our primary focus below is on $\mathcal F(v)$ that could yield a vertex subset in $\mathcal E$.

\begin{assumption}[Covariates richness]
\label{ass:anti_boundary}
For every $u\in\mathbb S^{dp-1}:=\{\beta\in\mathbb R^{d\times p}:\|\beta\|_F=1\}$ where $\|\cdot\|_F$ denotes the Frobenius norm, there exist a proper face set $F_u: \mathcal E\cap F_u\subsetneq\mathcal E$ and a constant $\delta_u>0$, such that 
\(
\liminf_{n\to\infty}\frac{1}{n}\sum_{i=1}^n1\Big\{\mathcal F(ux_i)=F_u,\min_{e\in \mathcal E\cap F_u,\ e'\in\mathcal E\setminus F_u}(e-e')^\top ux_i\ge\delta_u\Big\}>0.
\)
\end{assumption}

\begin{remark}
Intuitively, Condition~\ref{ass:anti_boundary} requires that for every unit
direction $u$ in the parameter space, there exists a nontrivial subset of
covariates along which the perturbation $\beta_0+u$ produces a detectable separation between $e^\top ux_i$ and $e'^\top ux_i$ for $e\in \mathcal E\cap F_u$ and $e'\in\mathcal E\setminus F_u$. 
{This ensures that any deviations from $\beta_0$ can be detected by a subset $\mathcal X$ of covariates such that $|\mathcal X\cap \{x_i:i\in[n]\}|/n$ is not vanishing as $n$ diverges.}
\end{remark}

\begin{remark}
Condition~\ref{ass:anti_boundary} is an essentially minimal requirement
for posterior consistency in our regression model. 
When it fails, $\beta_0+u_0$ and $\beta_0$ are distinguishable only on a increasingly smaller proportion of covariates as $n$ diverges, since  there exists a direction $u_0\in\mathbb S^{dp-1}$ such that 
$
\liminf_{n\to\infty}{(1/n)}\sum_{i=1}^n1\left\{\mathcal F(u_0x_i)=F,\min_{e\in \mathcal E\cap F,\ e'\in\mathcal E\setminus F}(e-e')^Tu_0x_i\ge\delta\right\}=0
$
for every proper face set $\mathcal E\cap F\subsetneq \mathcal E$ and every $\delta>0$,
\end{remark}

\begin{theorem}[Posterior consistency for the regression model]
\label{thm:consistency_beta}
Under Conditions~\ref{ass:prior}, \ref{ass:covariates_bound}, and \ref{ass:anti_boundary}, the posterior of the model \eqref{eq:regression} is consistent at estimating $\beta_0$, i.e., for any open neighborhood of $\beta_0$, denoted by $V_{\beta_0}$, the following holds:
\(
\Pi(V^c_{\beta_0}\mid Y^n)\stackrel{n\to\infty}{\to}0\quad Q^\infty_{\beta_0}-\text{ almost surely}.
\)
\end{theorem}

\begin{remark}
Throughout this section we treat the covariates $\{x_i\}_{i\ge1}$ as fixed. Nevertheless, Condition~\ref{ass:anti_boundary} has a natural interpretation under random design.
Suppose that $x_i$ are i.i.d.\ draws from a continuous distribution $P_X$ supported on a compact subset of $\mathbb R^p$. If the empirical frequencies in Condition~\ref{ass:anti_boundary} are replaced by probabilities under $P_X$, namely if for every $u\in\mathbb S^{dp-1}$ there exist a proper face set $F_u$ such that $\mathcal E\cap F_u\subsetneq\mathcal E$ and $\delta_u>0$ such that
\(
P_X\!\left(
\mathcal F(uX)=F_u,\;
\min_{e\in \mathcal E\cap F_u,\ e'\in\mathcal E\setminus F_u}(e-e')^\top uX\ge\delta_u
\right)>0,
\)
then Condition~\ref{ass:anti_boundary} holds almost surely by the law of large numbers. In this case, the posterior consistency result in Theorem~\ref{thm:consistency_beta} is expected to remain valid under random design.
\end{remark}

\section{Simulation studies}

\subsection{Fitted value curve}
We first conduct a simple simulation to demonstrate the importance of incorporating combinatorial constraints in the model. We revisit the example from Section~\ref{subsec:method} where responses $y_i \in \{0,1\}^2$ must satisfy the constraint $y_{i,1} + y_{i,2} \le 1$, meaning at most one component can be 1. The data generating process follows our latent variable framework where $\zeta_i \sim \text{N}(\mu_i, I_2)$ and the mean $\mu_i$ has a linear structure $\mu_i = \beta_0 + \beta_1 x_i$. Here $\beta_0, \beta_1 \in \mathbb{R}^2$ are coefficient vectors and $x_i \in \mathbb{R}$ is a scalar covariate. We use a one-dimensional covariate $x_i$ in order to clearly illustrate the fitted value curves.

We compare two models: (1) a standard bivariate probit regression that ignores the constraint (unconstrained model), and (2) our proposed model that incorporates the constraint (constrained model). Figure~\ref{plot:stack_comparison} shows the fitted value curves, as the estimated class probabilities from both models alongside the ground truth probabilities computed using the true coefficients $\beta_0$ and $\beta_1$.

Figure \ref{plot:stack_comparison}(b,c) show that our model successfully recovers the true probability distribution. On the other hand, panel (a) shows that the unconstrained model produces biased probability estimates.
Clearly, the unconstrained model incorrectly assigns positive probability to the outcome $(1,1)$ that should have had zero probability.

\begin{figure}[H]
\centering
\begin{subfigure}[t]{0.95\textwidth}
    \begin{overpic}[width=\textwidth, height = 3.8cm]{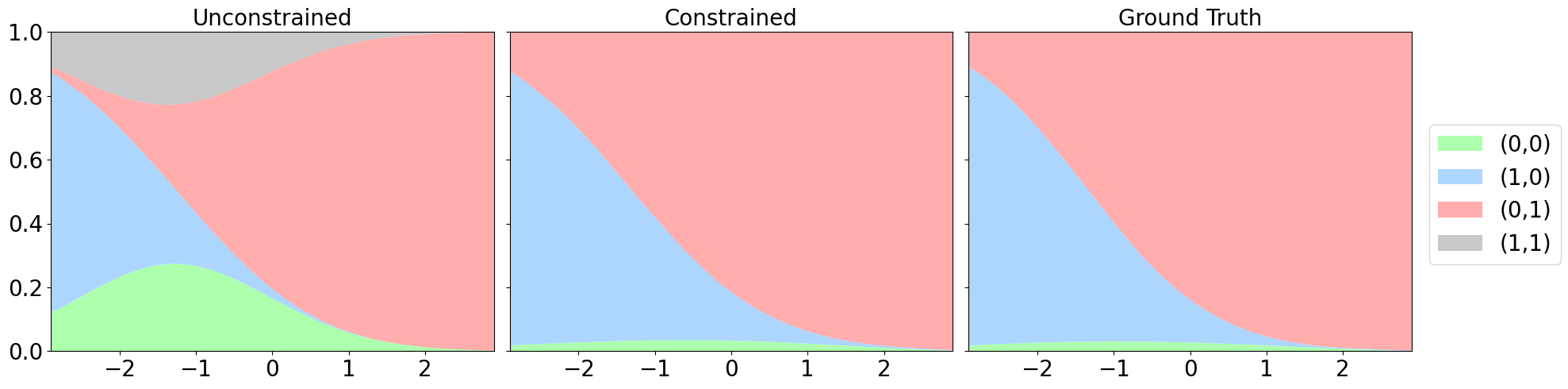}
    \put(46, -1){\scriptsize $x_i$}
    \put(16, -1){\scriptsize $x_i$}
    \put(76, -1){\scriptsize $x_i$}
    \put(-2,8){\rotatebox{90}{\scriptsize Probability}}
    \end{overpic}
\end{subfigure}
\caption{Comparison of predicted class probabilities under unconstrained and constrained models. Each panel shows the stacked predicted probabilities for the four response patterns as a function of the covariate $x_i$.}
\label{plot:stack_comparison}
\end{figure}
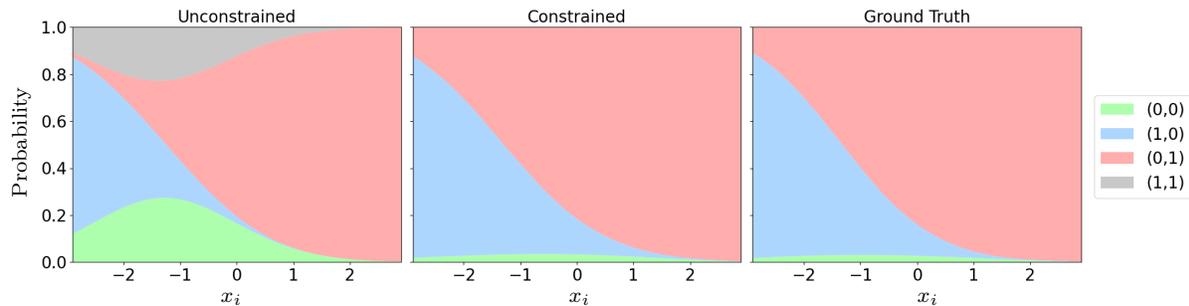

A natural question is if one could correct the bias of the unconstrained model by post-processing the predicted probabilities via conditioning on the constraint. Unfortunately, the answer is negative. As shown in Figure~\ref{plot:stack_comparison}(a), at $x_i \approx -1.4$, the unconstrained model incorrectly assigns almost equal probabilities to the three outcomes $(0,0)$, $(1,0)$, and $(0,1)$, whereas the ground truth probability of $(0,0)$ is close to zero.

\subsection{Signal recovery and computational efficiency under various dimensions}
{
We evaluate the performance of our method with simulation mimicking skip-logic surveys, a common source of combinatorial response data with structural zeros. To illustrate,
we first consider a short survey consisting of $d=5$ binary questions, where later questions are conditionally presented based on earlier responses.
For instance, respondents are first asked a screening question (Question~1), and only those answering positively are presented with follow-up questions (Questions~2 and~3).
Additional questions (Questions~4 and~5) are shown only if specific combinations of earlier responses are observed.
As a result, certain response patterns are infeasible by design, and show as structural zeros in the response vector. Table~\ref{tb:skip_logic_example} shows a short online health survey for studying preventive care behavior and its skip-logic conditions.

\begin{table}[H]
\centering
\begin{tabular}{clp{6.5cm}}
\toprule
Index & Question & Skip-logic condition \\
\midrule
$1$ & Has health insurance &
Always asked \\

$2$ & Had a routine checkup in past year &
Asked only if $y_{i,1} = 1$ \\

$3$ & Had preventive screening tests &
Asked only if $y_{i,1} = 1$ \\

$4$ & Checkup identified health risk &
Asked only if $y_{i,2} = 1$ \\

$5$ & Followed up with a specialist &
Asked only if $y_{i,4} = 1$ \\
\bottomrule
\end{tabular}
\caption{Skip-logic survey example with structural zeros.}
\label{tb:skip_logic_example}
\end{table}

This skip logic can be encoded using partial order constraints on the response vector.
For the example in Table~\ref{tb:skip_logic_example}, we can use $y_{i,2}\le y_{i,1}, y_{i,3}\le y_{i,1}, y_{i,4}\le y_{i,2}, y_{i,5}\le y_{i,4}$, which corresponds to $Az\le 0 \; (A\in\{-1,0,1\}^{4\times 5})$, where each row of $A$ has exactly one $-1$ and one $1$. Equivalently, the matrix $A$ can be viewed as the directed incidence matrix of a directed graph, hence guaranteed to be TUM.

Our simulation study adopts this skip-logic survey mechanism to form constraints and simulate data over different dimensions. Specifically, we generate a random graph containing $m$ nodes and $d$ edges, and use its directed incidence matrix as $A$. We use a binary vector for $b$. 
We simulate data from a regression model with a true coefficient matrix $\beta\in\mathbb R^{d\times p}$. Each simulation uses $n=1000$, $p=5$ and $d$ from a varying range. We examine how the method scales with response dimension by considering nine different values of $d$, ranging from small ($d=2,5,10,20,50$) to large ($d=100,200,500,1000$). The number of constraints $m$ is choosen from $\{1, 5, 10, 20, 50, 100\}$. For prior specification, we assign $\beta\sim\text{Mat-N}(0,I_d,I_p\tau)$ with $\tau$ set as $10$.
}

For each simulation, we perform MCMC sampling with 50000 total iterations, discarding the first 5000 as burn-in and retaining every 25th sample as thinnings thereafter. The resulting trace plots and autocorrelation function (ACF) plots shown in Figure~\ref{plot:beta_ACF_and_trace} indicate efficient mixing of the Markov chain, with quick decay in the autocorrelation structure.

\begin{table}[H]
  \centering
  \small
  \begin{tabular}{cccccccccc}
    \hline
        & \multicolumn{9}{c}{d}                                                 \\ \cline{2-10} 
    m   & 2     & 5     & 10    & 20    & 50    & 100   & 200   & 500   & 1000  \\ \hline
    1   & 0.046 & 0.066 & 0.076 & 0.065 & 0.086 & 0.079 & 0.090 & 0.094 & 0.091 \\
    5   & -     & - & 0.066 & 0.084 & 0.084 & 0.081 & 0.082 & 0.088 & 0.087 \\
    10  & -     & -     & - & 0.088 & 0.099 & 0.084 & 0.095 & 0.807 & 0.088 \\
    20  & -     & -     & -     & -     & 0.116 & 0.104 & 0.090 & 0.091 & 0.091 \\
    50  & -     & -     & -     & -     & -     & {0.209}     & 0.171 & 0.125 & 0.104 \\
    100 & -     & -     & -     & -     & -     & -     & -     & {0.263}     & 0.128 \\ \hline
    \end{tabular}
    \caption{Root mean squared error (RMSE) for estimating $\beta$ using posterior mean under different dimensionality $(d,m)$. We choose $d>m$ so that the number of constraints is smaller than the number of elements in $y_i$.}
  \label{tb:RMSE}
  \end{table}

For the running time per $1000$ iterations, it takes 0.07 minutes for $(d,m)=(2,1)$, 8 minutes for $(d,m)=(1000,1)$, and 52 minutes for $(d,m)=(1000,100)$, using Rcpp on MacBook Pro with a 8-core CPU. We provide the R source code in the Supplementary Materials.

To assess estimation accuracy, we calculate the root mean squared error (RMSE) between the posterior mean estimates and the ground-truth values, presented in Table~\ref{tb:RMSE}. Figure~\ref{plot:beta_violin} displays the posterior distribution for each $\beta_{i,j}$ parameter, demonstrating that the credible intervals effectively capture the true parameter values. {We provide additional details in the Supplementary Materials S4.}

\begin{figure}[H]
\centering
\begin{subfigure}[t]{0.3\textwidth}
    \begin{overpic}[width=\textwidth, height = 3.8cm]{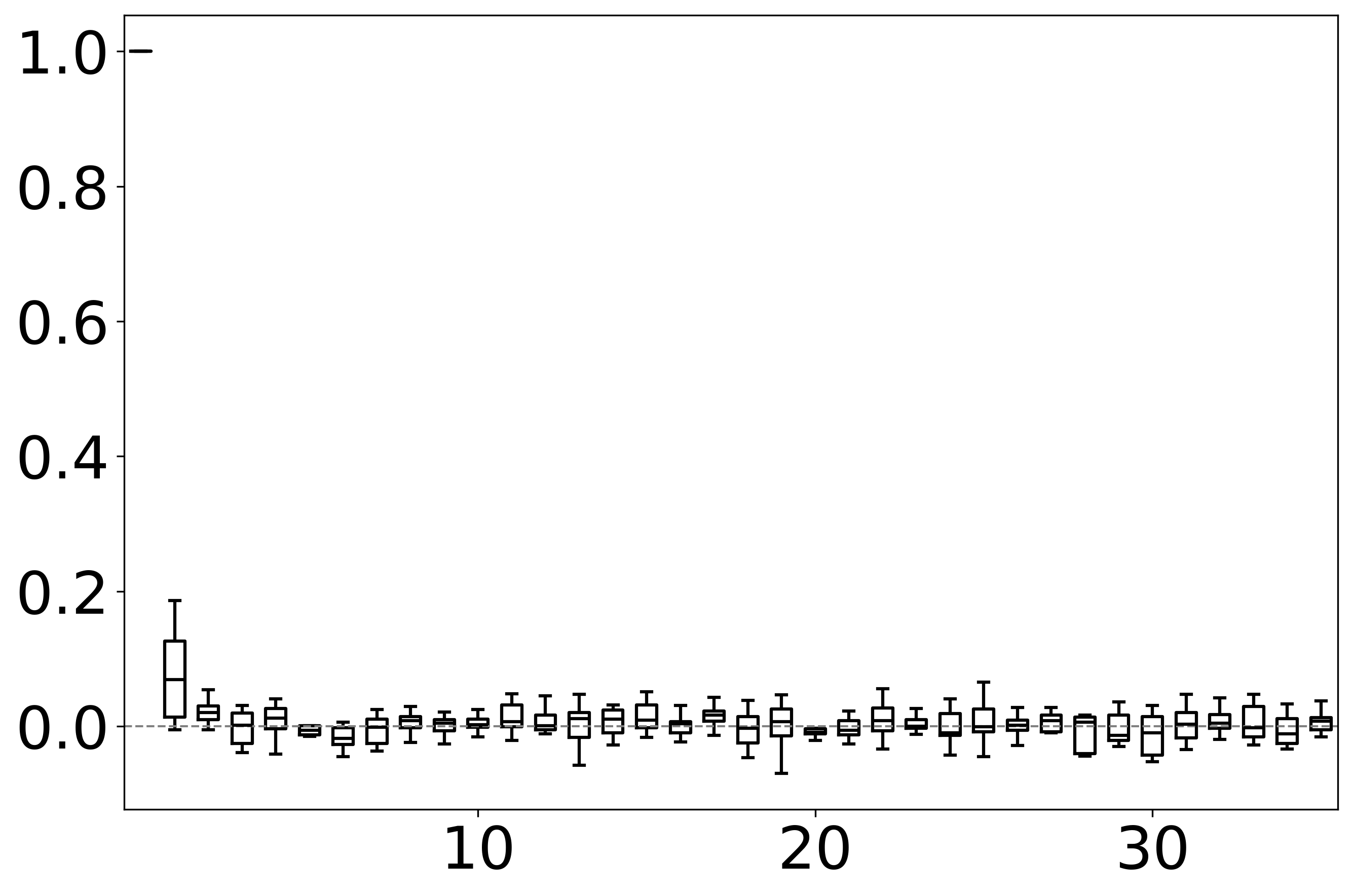}
    \put(50, -4){\scriptsize Lag}
    \end{overpic}
    \caption{$(d,m) = (2,1)$.}
\end{subfigure}\;
\begin{subfigure}[t]{0.3\textwidth}
    \begin{overpic}[width=\textwidth, height = 3.8cm]{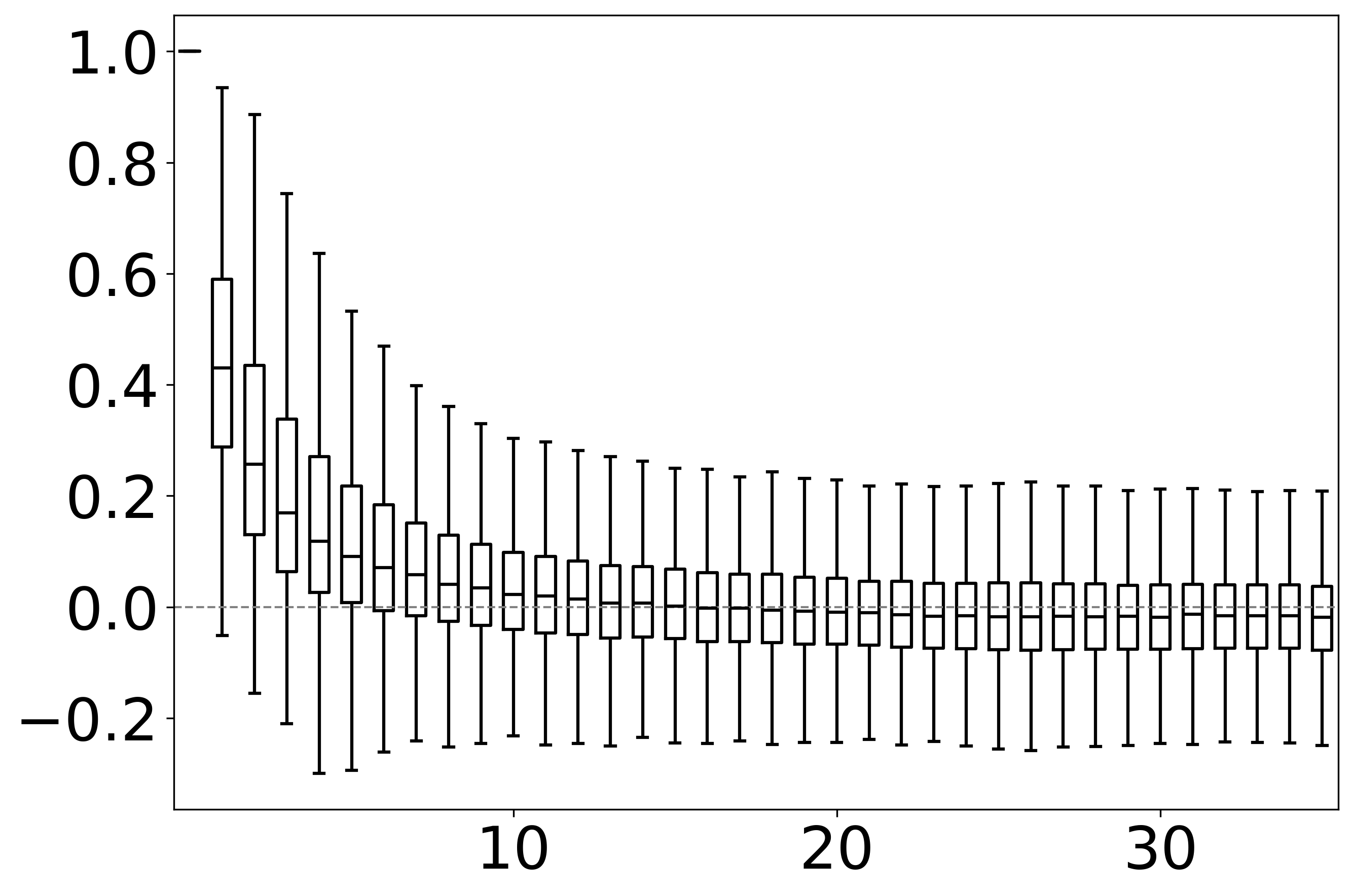}
    \put(50, -4){\scriptsize Lag}
    \end{overpic}
    \caption{$(d,m) = (500,50)$.}
\end{subfigure}\;
\begin{subfigure}[t]{0.3\textwidth}
    \begin{overpic}[width=\textwidth, height = 3.8cm]{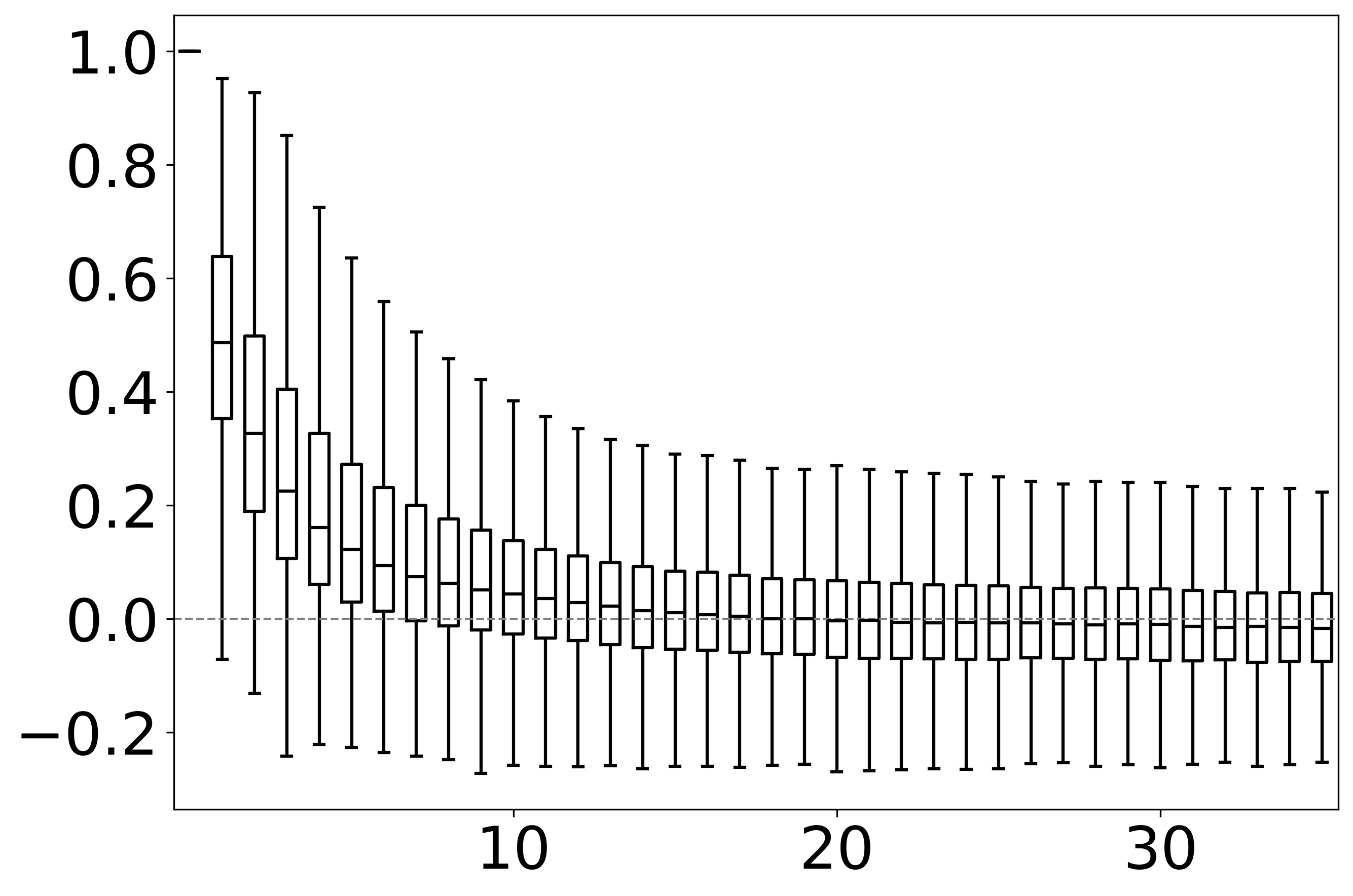}
    \put(50, -4){\scriptsize Lag}
    \end{overpic}
    \caption{$(d,m) = (10^3,10^2)$.}
\end{subfigure}\\
\begin{subfigure}[t]{0.3\textwidth}
    \begin{overpic}[width=\textwidth, height = 3.8cm]{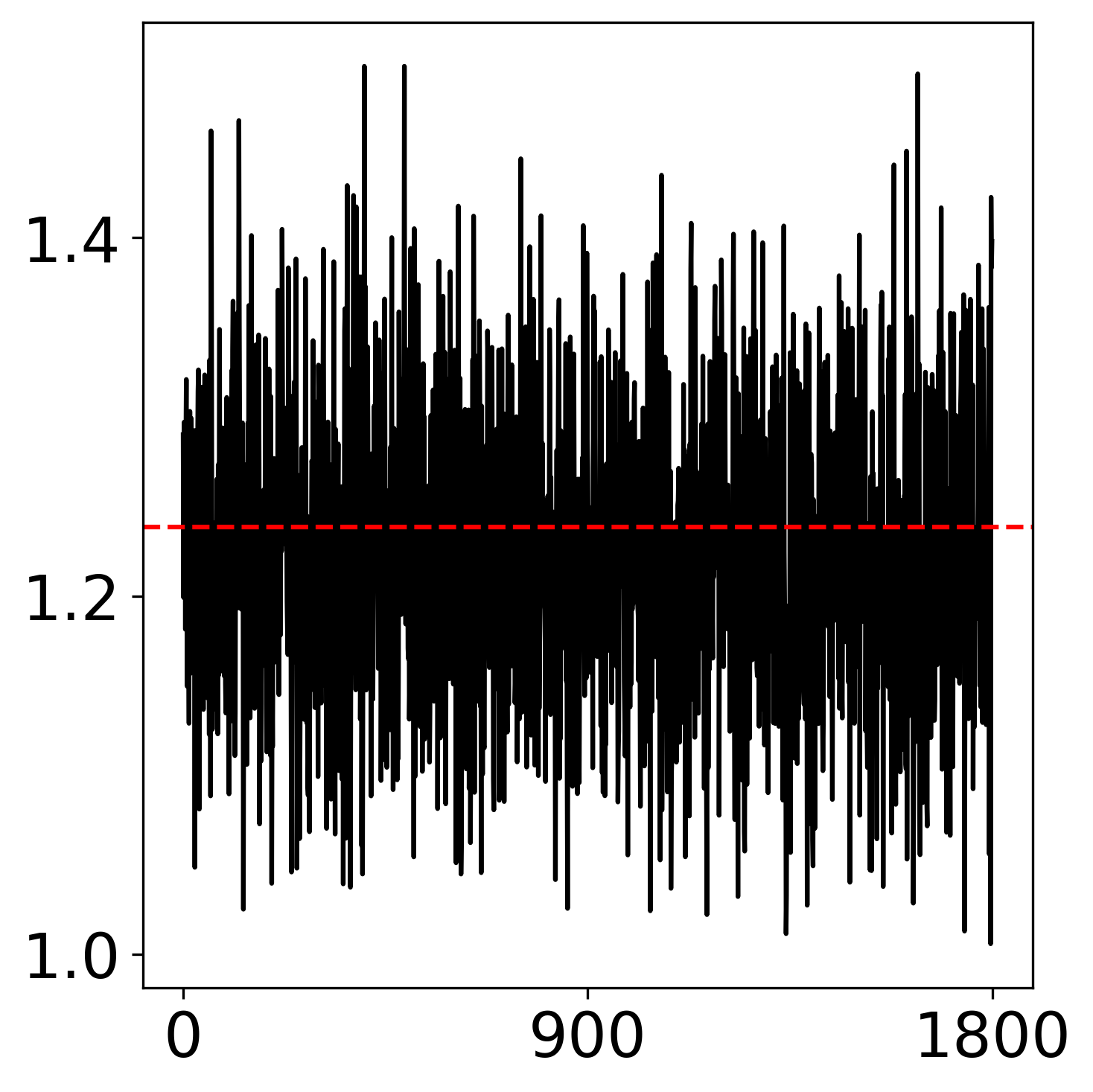}
    \put(44, -4){\scriptsize Iteration}
    \end{overpic}
    \caption{$(d,m) = (2,1)$.}
\end{subfigure}\;
\begin{subfigure}[t]{0.3\textwidth}
    \begin{overpic}[width=\textwidth, height = 3.8cm]{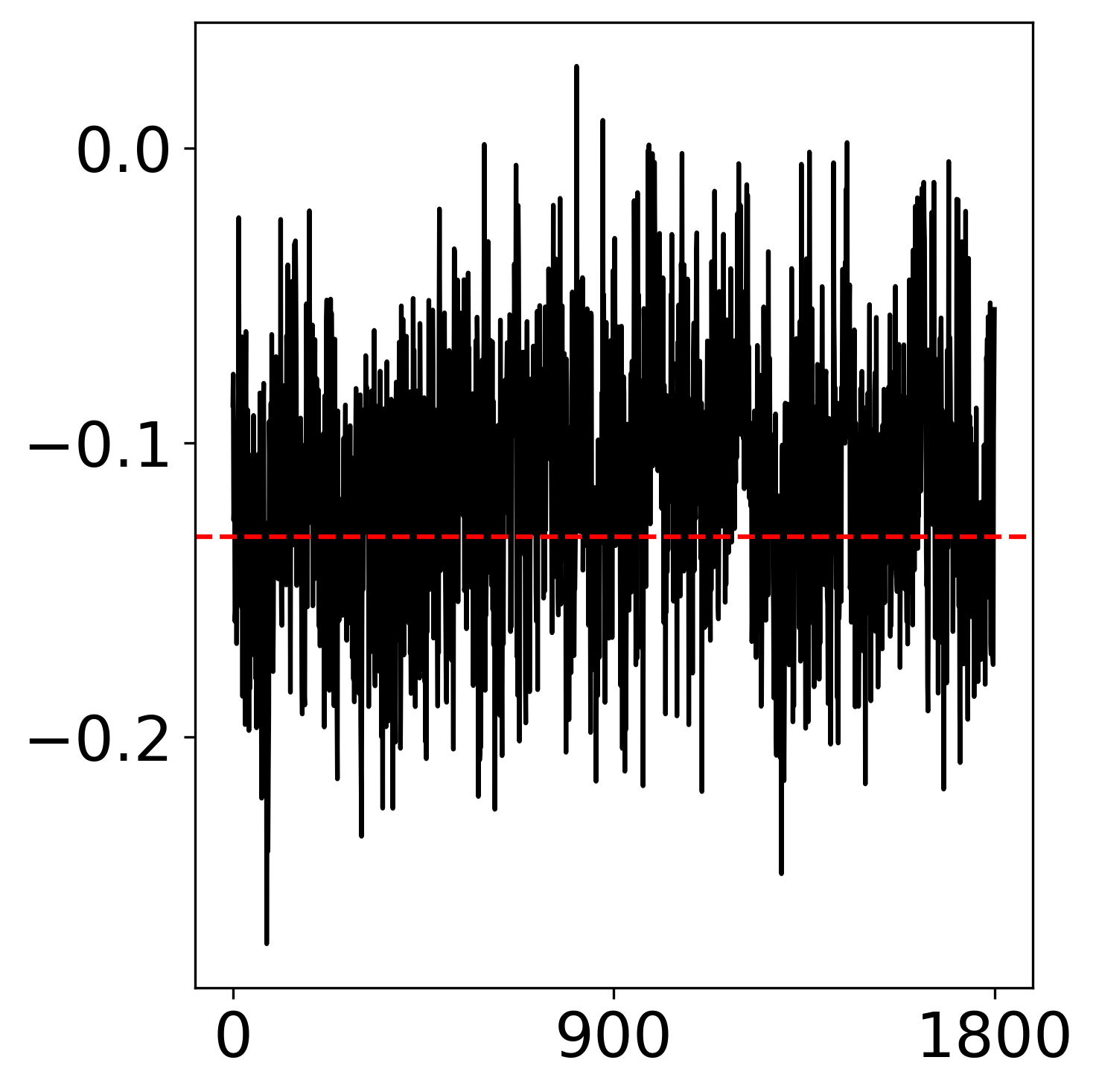}
    \put(44, -4){\scriptsize Iteration}
    \end{overpic}
    \caption{$(d,m) = (500,50)$.}
\end{subfigure}\;
\begin{subfigure}[t]{0.3\textwidth}
    \begin{overpic}[width=\textwidth, height = 3.8cm]{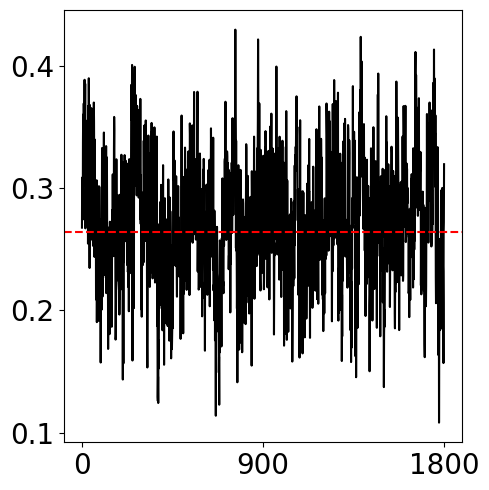}
    \put(44, -4){\scriptsize Iteration}
    \end{overpic}
    \caption{$(d,m) = (10^3,10^2)$.}
\end{subfigure}
\caption{Autocorrelation function (ACF) and trace plots for different values of $(d,m)$. Each box displays the ACF of all $\beta_{i,j}$ parameters. The ACF is computed after thinning the posterior samples by retaining every 25th sample. The trace plots are based on a randomly chosen parameter under each setting.}
\label{plot:beta_ACF_and_trace}
\end{figure}

\begin{figure}[H]
\centering
\begin{subfigure}[t]{0.48\textwidth}
    \begin{overpic}[width=\textwidth,height=4cm]{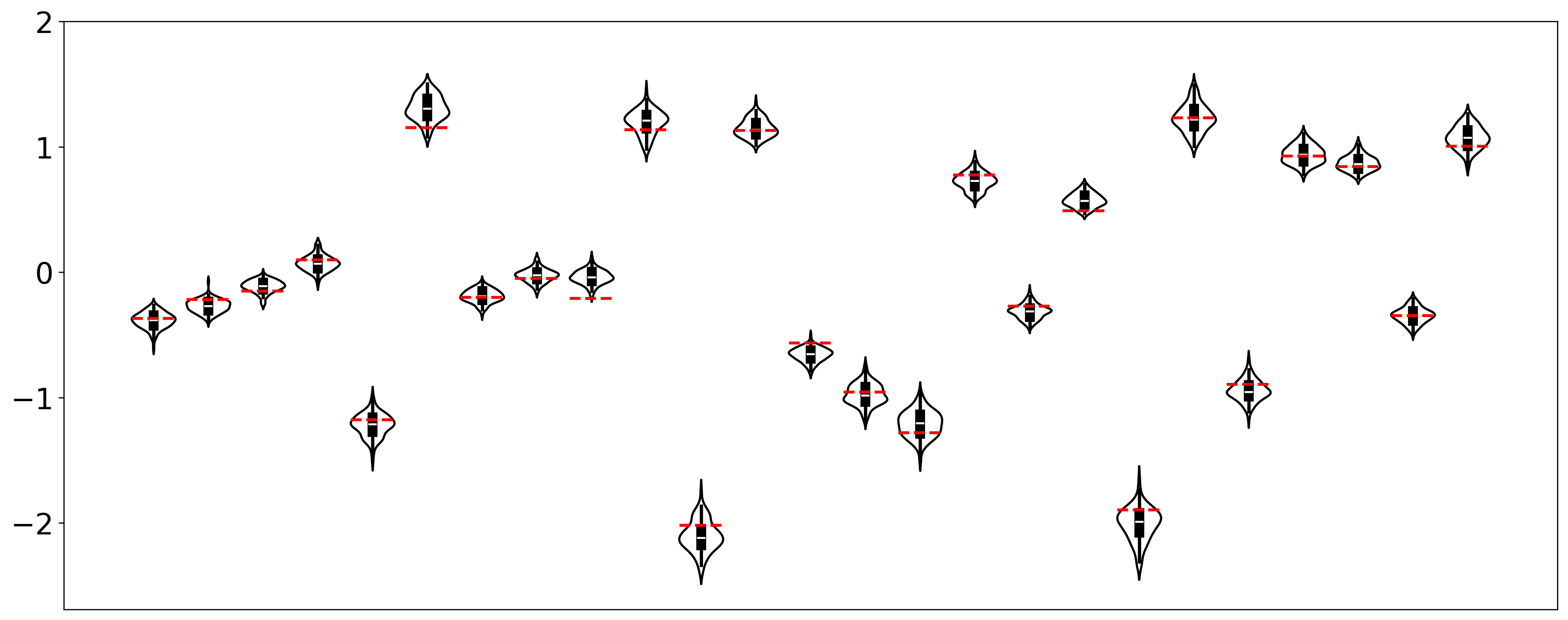}
    \put(43, -4){\scriptsize Parameter}
    \end{overpic}
    \caption{Posterior distributions of $\beta_{i,j}$ for $(d,m) = (5,1)$.}
\end{subfigure}\;
\begin{subfigure}[t]{0.48\textwidth}
    \begin{overpic}[width=\textwidth,height=4cm]{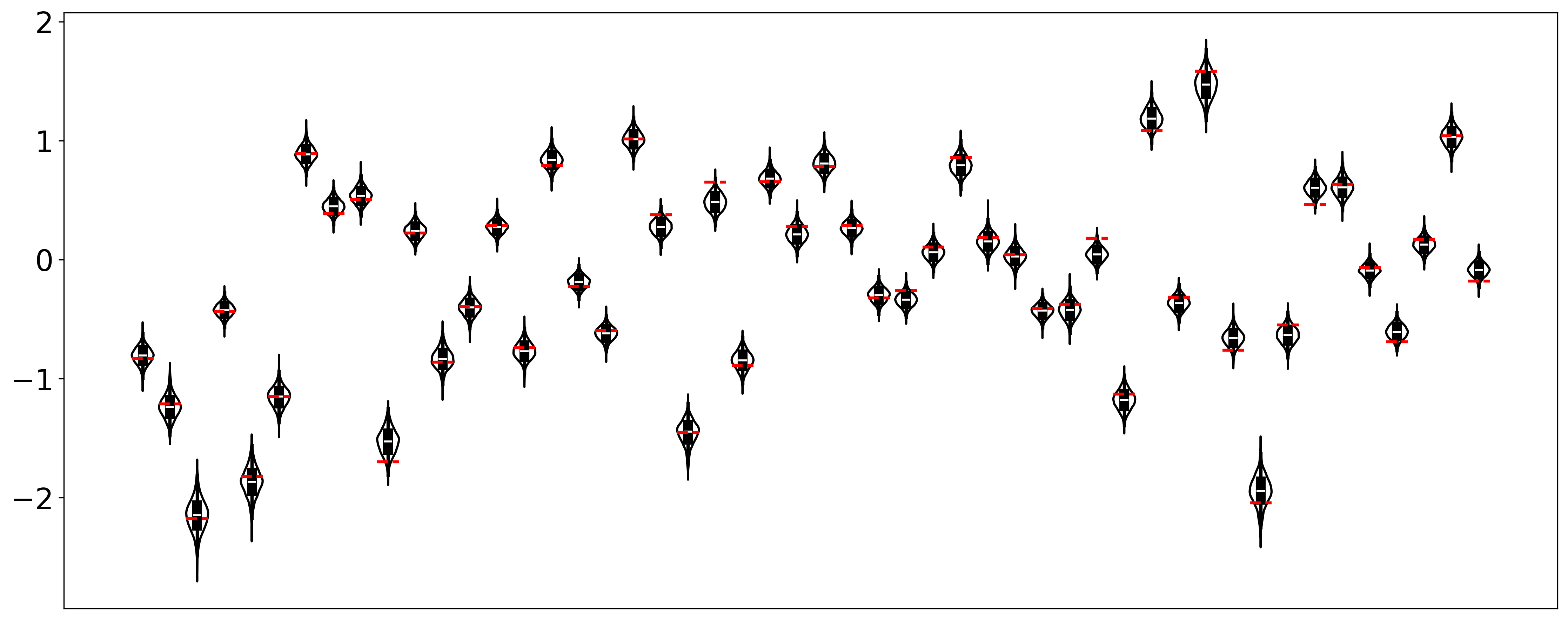}
      \put(43, -4 ){\scriptsize Parameter}
    \end{overpic}
    \caption{Posterior distributions of $\beta_{i,j}$ for $(d,m) = (10,5)$.}
\end{subfigure}
\caption{Violin plots of the posterior distributions of $\beta_{i,j}$ from MCMC sampling for different values of $(d,m)$. The horizontal lines represent the ground truth.}
\label{plot:beta_violin}
\end{figure}

\section{Data analysis on waterfowl matching}
\label{sec:data_application}
We conduct data analysis on the seasonal matching of waterfowl. All waterfowl are monogamous, though their mating patterns differ across species. Geese, swans, and whistling ducks exhibit perennial monogamy with lifelong pair bonds, while ducks practice seasonal monogamy where pair bonds last 4 to 8 months, and they choose new mates each year. The pair formation process follows a clear seasonal pattern: pairing initiates during fall migration in September–-October, most pairs are formed by December--January, and pair bonds continue to strengthen through winter and spring.

Our dataset is the temporal observation of the pairing of male and female ducks from November to March. In the dataset, there are 95 ducks ($|V|=95$) in total, and from 7 different species: American Black Duck, Mallard, Gadwall, Scaup, Redhead, Canvasback, and Ring-necked Duck. Since pairs can only be formed between a male and a female within the same species, the observed matching $M^t=(V,E^t_M)$ at time $t$ is the subgraph of a bipartite graph $G=(V,E_G)$, with 339 possible edges ($|E_G|=339$). As described in Example 4 of Section 2, we use $A$ as the incidence matrix for $G$, and convert each $M^{t}$ into a binary vector $y_{t}\in \{0,1\}^{d}$ with $d=339$, corresponding to whether the $e$-th edge in $G$ is chosen into $M^{t}$. The polyhedron associated with matching is then defined by $Ay_{t}\le 1$.  The maximal number of ones in $y_{t}$ is $46$, which is substantially smaller than $339$ due to the polyhedron constraint.

For the covariates, time is modeled using B-splines corresponding to weeks. We use B-spline bases corresponding to {$\kappa= 5$} columns in the spline matrix. Additionally, we incorporate the weights of both male and female ducks measured at the end of November. Our model for each latent variable $\zeta_t \in \mathbb{R}^d$ is
\(
\zeta_t = \beta_0 1_{d} +   x_{\text{weights}} \beta_1 + C_{\text{species}}\beta_2 x_{\text{splines}:t}  + \epsilon_t,\quad \epsilon_t \sim \text{N}(0,I_d)
\)
where $\beta_0\in \mathbb{R}$ is the common intercept, $x_{\text{weights}}\in \mathbb{R}^{d\times 2}$ has each row as the weights of two ducks on each edge of $G$, $\beta_1\in\mathbb R^2$ is the coefficient vector for duck weights, $x_{\text{splines}:t} \in \mathbb{R}^{\kappa\times 1}$ corresponds to the B-spline basis function value at time $t$, $\beta_2\in \mathbb{R}^{7 \times \kappa}$ is the coefficient matrix corresponding to 7 individual curves for species, $C_{\text{species}} \in \{0,1\}^{d\times 7}$ is the indicator matrix for species in which each row has only one element equal one. For prior specification, we assign $\beta_0,\beta_{1:1},\beta_{1:2}\stackrel{i.i.d.}{\sim}\text{N}(0,\tau^{-1}_\alpha)$ and $\beta_{2:(i,\cdot)}\stackrel{indep}{\sim}\text{N}(0,(\tau_\beta s_i)^{-1}I_\kappa)$ with $s_i$ denoting the number of ducks of species $i$. We use the parameterization with $s_i$ in order for the full conditional of $\beta_{2}$ to be a matrix normal distribution. We set $\tau_\alpha=\tau_\beta=0.1$. The posterior sampling method can be found in the Supplementary Materials S2.

\begin{figure}[H]
  \centering
  \begin{subfigure}[t]{0.24\textwidth}
      \begin{overpic}[width=\textwidth]{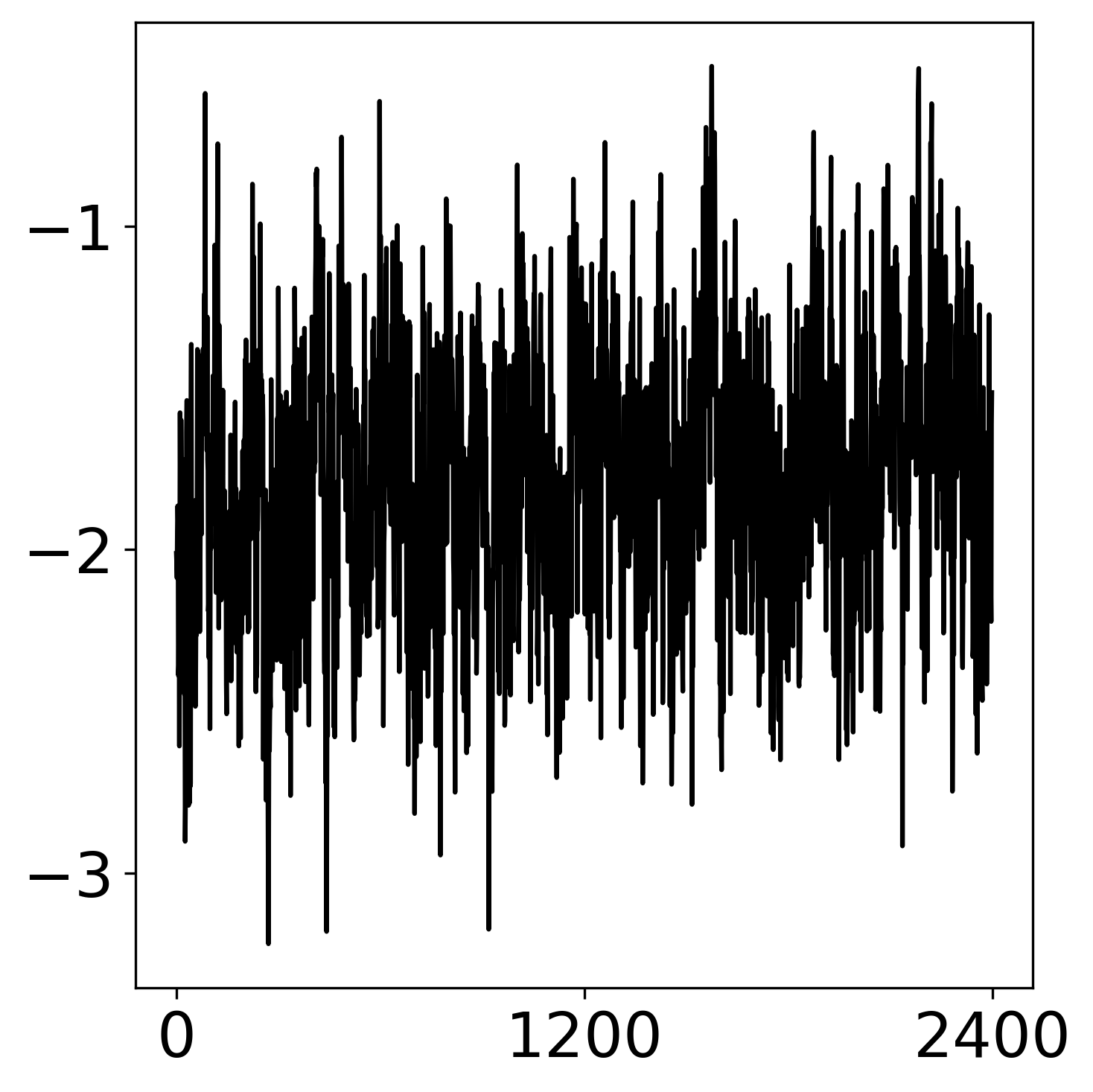}
      \put(44, -4){\scriptsize Iteration}
      \end{overpic}
      \caption{For $\beta_0$.}
  \end{subfigure}
  \begin{subfigure}[t]{0.24\textwidth}
      \begin{overpic}[width=\textwidth]{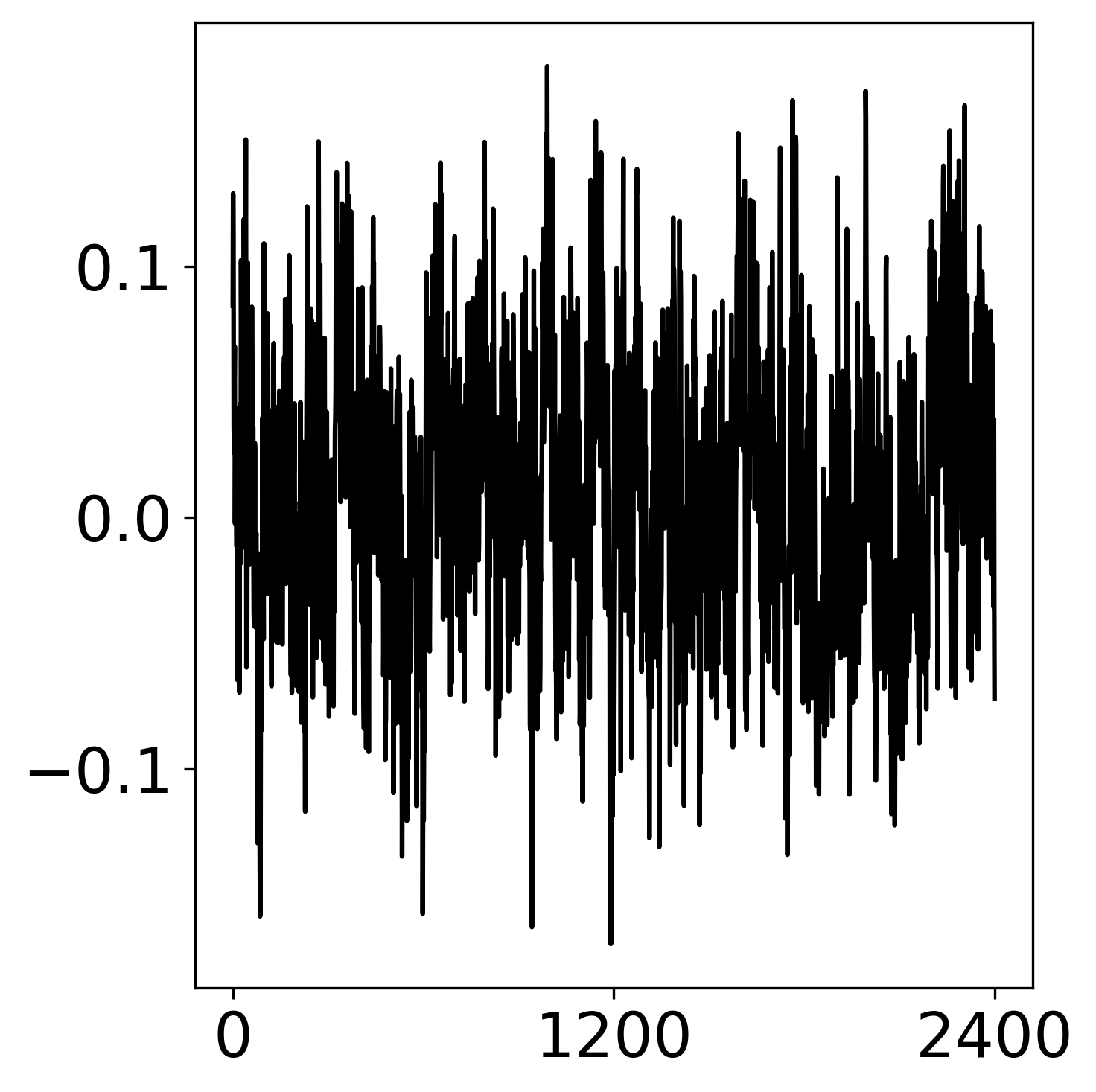}
      \put(44, -4){\scriptsize Iteration}
      \end{overpic}
      \caption{For $\beta_{1:1}$.}
  \end{subfigure}
  \begin{subfigure}[t]{0.24\textwidth}
      \begin{overpic}[width=\textwidth]{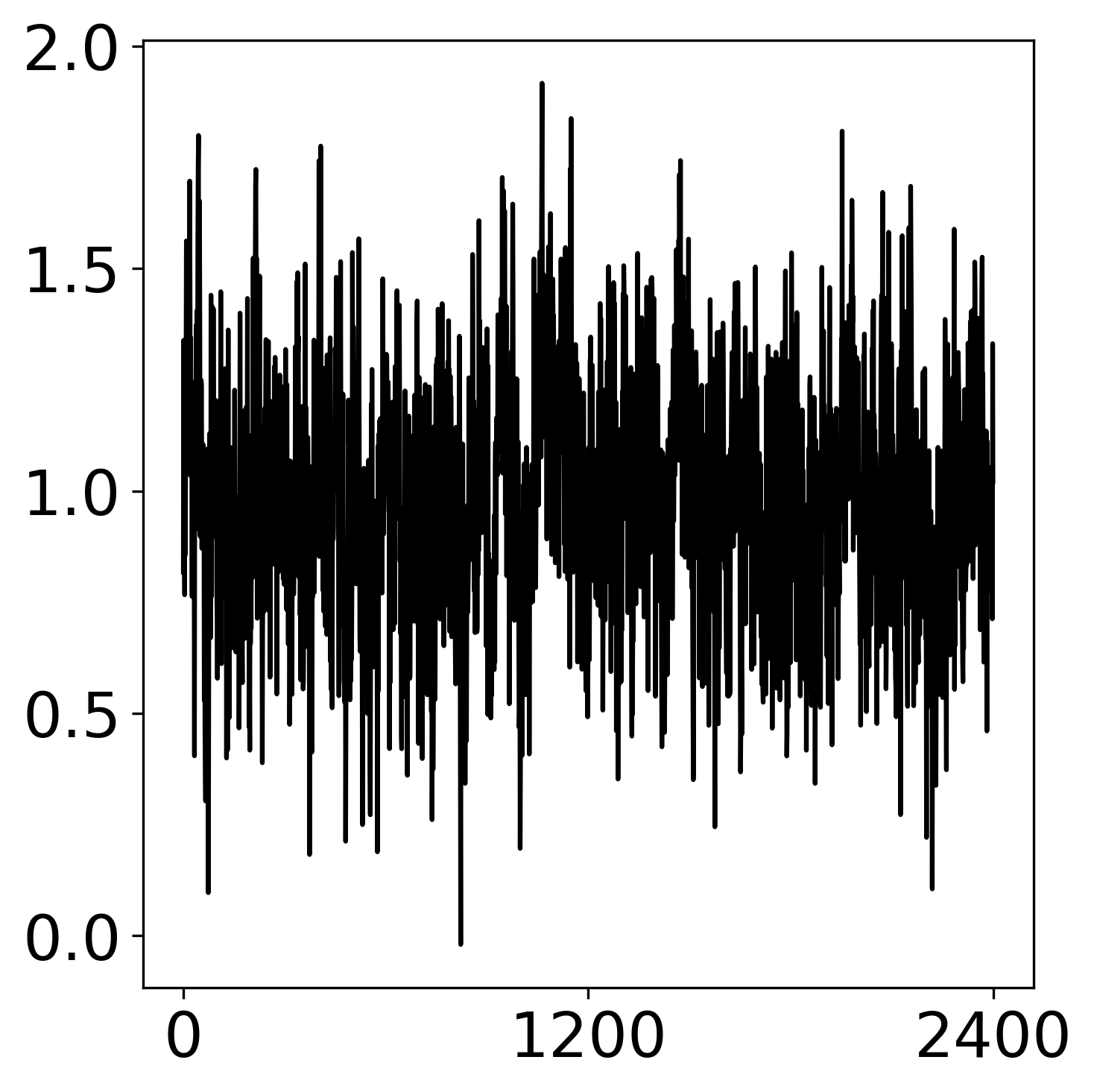}
      \put(44, -4){\scriptsize Iteration}
      \end{overpic}
      \caption{For $\beta_{2:(2,2)}$.}
  \end{subfigure}
  \begin{subfigure}[t]{0.24\textwidth}
      \begin{overpic}[width=\textwidth]{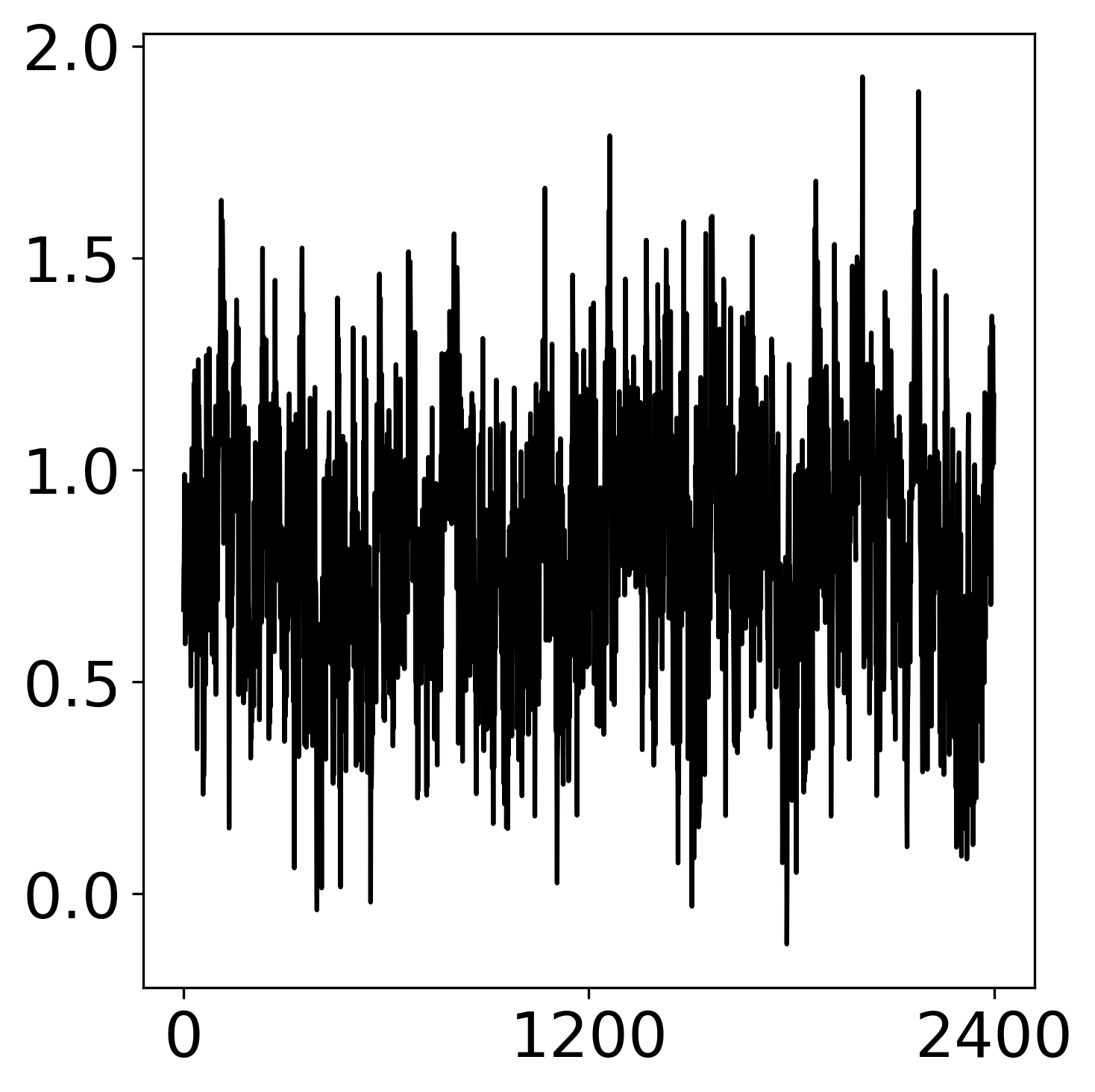}
      \put(44, -4){\scriptsize Iteration}
      \end{overpic}
      \caption{For $\beta_{2:(7,5)}$.}
  \end{subfigure}
  \caption{Trace plots of randomly selected parameters for the duck matching model.}
  \label{plot:duck_trace}
  \end{figure}

We run $50000$ MCMC iterations, with the first $2000$ treated as burn-ins and the rest thinned at every $20$th iteration. 
Figure~\ref{plot:duck_trace} and \ref{plot:duck_ACF} show the trace plots for four randomly selected parameters and the ACF plots, which show excellent mixing.
Table~\ref{tb:est_duck} presents the posterior means and 95\% credible intervals for the coefficient estimates
of $\beta_{1:1}$ and $\beta_{1:2}$, corresponding to the fixed effects of the male and female duck weights, respectively. The effect of female duck weight is statistically significant, while male duck weight is not.
  
   \begin{figure}[H]
   \centering
   \begin{subfigure}[t]{0.45\textwidth}
       \begin{overpic}[width=\textwidth]{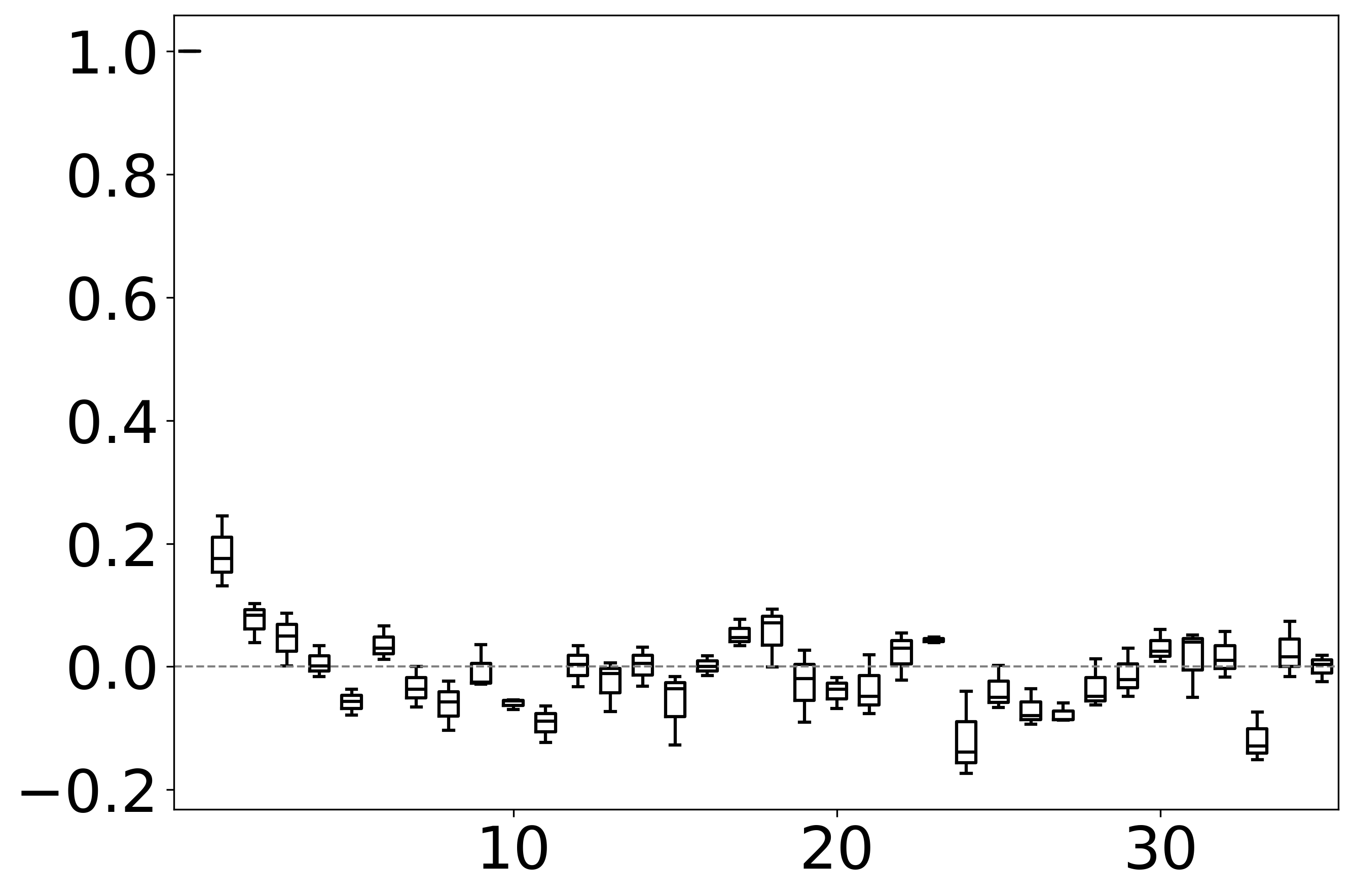}
       \put(50, -2){\scriptsize Lag}
       \end{overpic}
       \caption{ACF for $(\beta_0,\beta_{1:1}, \beta_{1:2})$.}
   \end{subfigure}
   \begin{subfigure}[t]{0.45\textwidth}
       \begin{overpic}[width=\textwidth]{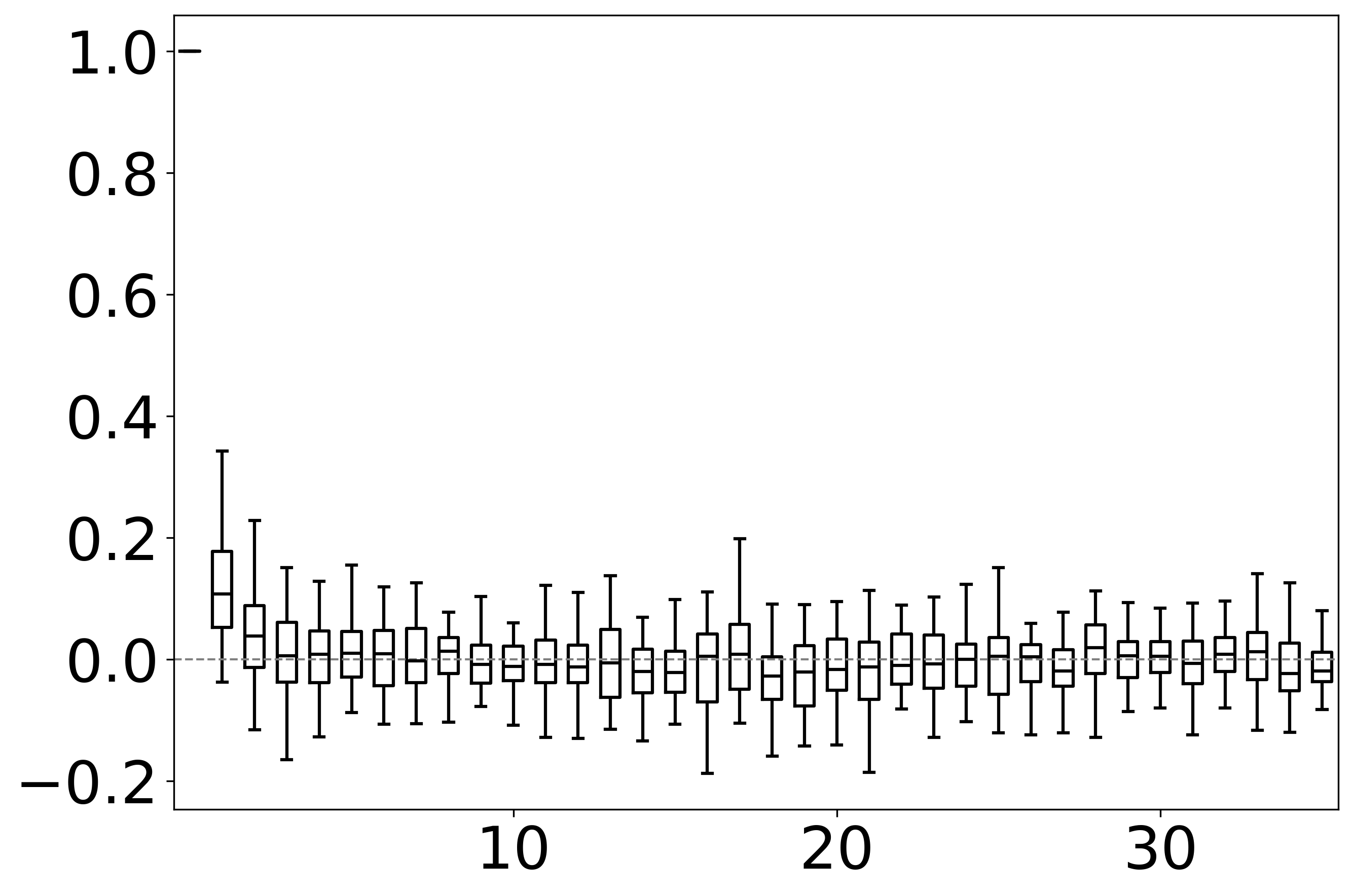}
       \put(50, -2){\scriptsize Lag}
       \end{overpic}
       \caption{ACF for all $\beta_{2:(i,j)}$.}
   \end{subfigure}
   \caption{ACF plots for the ducking matching model. Each box displays the ACF of all parameter components.}
   \label{plot:duck_ACF}
   \end{figure}

\begin{table}[H]
  \footnotesize
\centering
\begin{tabular}{lcc}
\hline
Parameter     & Estimate & Credible Interval \\ \hline
Intercept $\beta_0$     & -1.786   & (-2.566, -1.006)  \\
Male weight $\beta_{1:1}$ & 0.010    & (-0.090, 0.116)   \\
Female weight $\beta_{1:2}$ & 0.187    & (0.012, 0.362)    \\ \hline
\end{tabular}
\caption{Coefficient estimates and 95\% credible intervals for $\beta_0$ and $\beta_1$.}
\label{tb:est_duck}
\end{table}

We now show how the matching probability changes over time for each species.
We plot the probability curves separately for male and female ducks, while fixing the weight to be the medium of the males and females, respectively. The probabilities are computed by Monte Carlo integration, which involves sampling $1000$ $\zeta_t$'s given $\beta$ at each time point $t$ and computing the associated $\tilde y_{t} = T(\zeta_t)$ using linear program. We evaluate the probability of $\sum_{e: e = (i, j^*) \text{ or } (j^*, i)} \tilde y_{t,e} = 1$ ($j^*$ is the index of the medium-weight duck), corresponding to the event that the duck $j^*$ forms a match with another one. We calculate these probabilities for each Markov chain sample of $\beta$, allowing us to obtain point-wise 95\% credible bands. We plot the curves in Figure \ref{plot:matching_probs_error_band}. In addition, we overlay the mean curves between the two sexes and show comparison in the Supplementary Materials S3.

\begin{figure}[H]
  \centering
  \begin{subfigure}[t]{0.93\textwidth}
      \begin{overpic}[width=\textwidth]{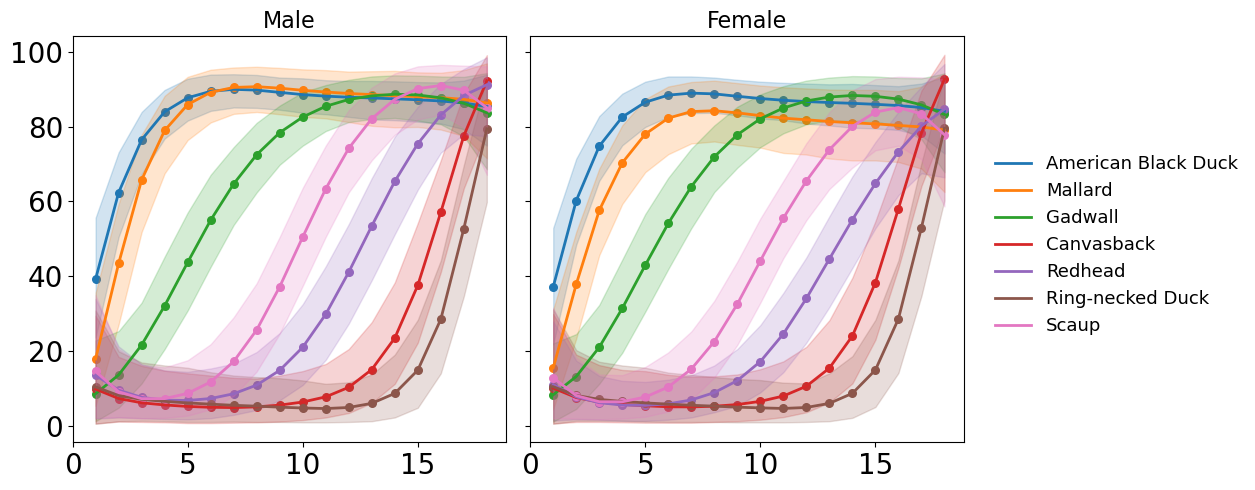}
      \put(40, -2){\scriptsize Week}
      \put(-2,8){\rotatebox{90}{\scriptsize Matching Probability (\%)}}
      \end{overpic}
  \end{subfigure}
  \caption{Matching probabilities over time with pointwise 95\% credible bands.}
  \label{plot:matching_probs_error_band}
  \end{figure}

  We can see a clear difference between dabbling ducks (American Black Duck, Mallard, Gadwall) and diving ducks (Scaup, Redhead, Canvasback, Ring-necked Duck). This is consistent with the biological knowledge that dabbling ducks tend to form pairs much earlier than diving ducks. On the other hand, we can see that Gadwall form matchings slightly later than American Black Duck and Mallard in the year.

 Next, given that there are more females than males for some species in this dataset, we assess the effects of competition on the matching probabilities. To be exact, we compute the conditional probability given that there is at least one opposite-sex duck of the same species that is not matched to any other duck, and view it as a surrogate measure of matching probability if there were no competition. We provide the result in Figure~\ref{plot:matching_probs_female}. For Mallard, the effects of competition start to show as early as in the fourth week, whereas the effects show for Redhead and Scaup later in the season.

\begin{figure}[H]
  \centering
  \begin{subfigure}[t]{0.8\textwidth}
      \begin{overpic}[width=\textwidth]{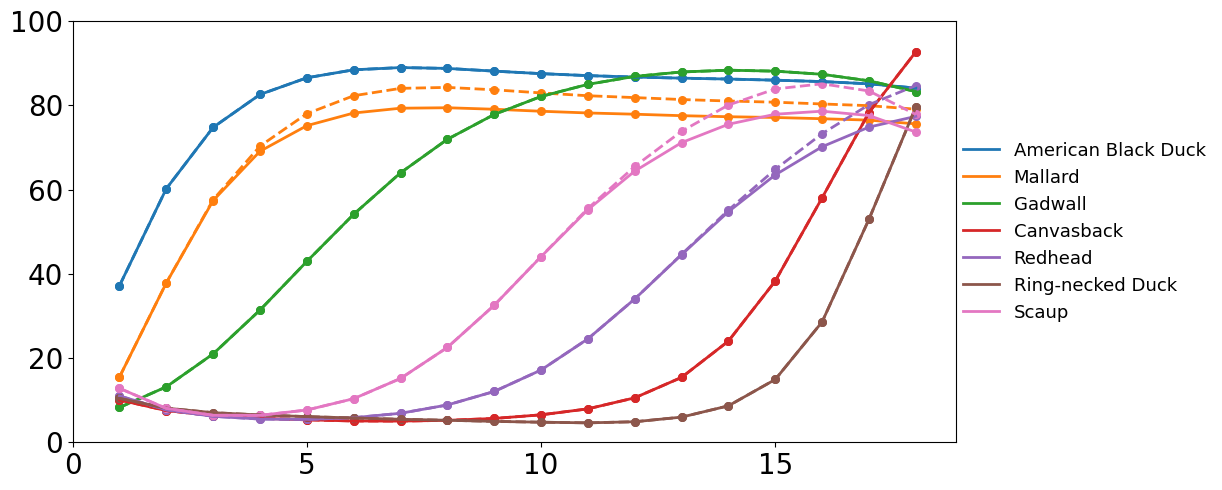}
      \put(42, -1){\scriptsize Week}
      \put(-1,10){\rotatebox{90}{\scriptsize Matching Probability (\%)}}
      \end{overpic}
  \end{subfigure}
  \caption{Mean matching probabilities over time for female ducks. Solid lines correspond to having competition, and dashed lines correspond to no competition.}
  \label{plot:matching_probs_female}
  \end{figure}

\vspace*{-1cm}

\section{Discussion}
In this article, we propose a likelihood-based approach for handling combinatorial response data. We take advantage of the properties of integer linear program, which lead to a data augmentation strategy that enjoys both tractable theoretic analysis and efficient computation. There are several interesting directions worth further pursuing. First, for consistency theory, we consider the case when the ground truth is covered by the family of our specified models. On the other hand, it remains to be seen if one could recover the fitted value functional under some level of model misspecification. Second, we could extend the strategy of augmenting observed data with optimization decision variables beyond linear program. For example, there is a class of optimization approaches based on semi-definite program, which are applicable to inference tasks such as covariance modeling and metric learning.

\noindent \textbf{Disclosure on use of AI tools:} During the preparation of this manuscript, we used ChatGPT 5.2 for code assistance. All outputs were reviewed, verified, and, where necessary, corrected by the authors, who take full responsibility for the content.

\spacingset{1}

 \section*{Supplementary Materials}

\section{Proof of Theorems}

\subsection{Proof for Section 2}

\begin{proof}[Proof of Theorem~\ref{thm:unique_solution_as}]
    We first prove the claim that if the solution to continuous linear program is not unique, then there exists another vertex that is also optimal.
    
    Given any $\zeta_i\in\mathbb R^d$ fixed. Suppose that $z_1\not=z_2\in \mathcal P$ are both optimal with $\zeta_i^\top z_1=\zeta_i^\top z_2=v$, $z_1$ is a vertex of $\mathcal P$, and $z_2$ is not a vertex of $\mathcal P$. Take any vertex $z_3\in\mathcal P$ different than $z_1$. There exist $z_4\in\mathcal P$ and $\lambda\in(0,1)$ such that $z_2=\lambda z_3+(1-\lambda)z_4$. By the optimality of $z_2$,
    \(
    v=\zeta^\top _iz_2=\lambda \zeta^\top _iz_3+(1-\lambda)\zeta^\top _iz_4\le \lambda v+(1-\lambda)v=v,
    \)
    with the equality holds if and only if $\zeta^\top _iz_3=\zeta^\top _iz_4=v$. Therefore, $z_3$ is another vertex of $\mathcal P$ that is also optimal.
    Noting that $z_1-z_3\in\mathbb Z^d$ and $\zeta_i^\top(z_1-z_3)=0$, we have
    \(
    &P(\text{Solution to \eqref{eq:clp} is NOT unique})\\
    \le& P(\text{There exist distinct vertices }z_1,z_3\in\mathcal P\text{ that are both solutions to \eqref{eq:clp}})\\
    \le& \sum_{z\in\mathbb Z^d}P(\zeta_i^\top z=0)=0.
    \)
    \end{proof}
  
    \begin{proof}[Proof of Theorem~\ref{thm:z_and_mu}]
      Rewriting the both sides of strong duality $\eta^\top \tilde z=u^\top b+1^\top (\eta-A^\top u)_+$, we have the sufficient and necessary condition:
      \(
        \sum_j \eta_j 1(\tilde z_j=1)= u^\top b+ \sum_j 1(\tilde z_j=1)(\eta_j-A_{.j}^\top u),
      \)
      rearranging terms yields $u^\top b= \sum_j 1(\tilde  z_j=1)(A_{.j}^\top u)=u^\top A\tilde z$, the first condition.
      
      For sufficiency, note that the value of $\eta_i$ does not impact the strong duality equality, as long as $\eta_j \ge A_{\cdot j}^\top u \text{ for } j: \tilde  z_j=1,$ and $\eta_j \le A_j^\top u \text{ for } j:\tilde  z_j=0$, hence $\tilde z$ is optimal. For necessity, when $\tilde z$ is optimal, we know we cannot have $\eta_j<A_{\cdot j}^\top u$ for $\tilde z_j=1$ as $\hat w_j= A_{.j}^\top \hat u-\eta_j>0$ would violate complementary slackness; similarly, we cannot have $\eta_j>A_{\cdot j}^\top u$ for $\tilde z_j=0$.
    \end{proof}

\subsection{Notations}

We consider a setting in which the response variables $y_i$ take values in the set of vertices, denoted by $\mathcal E$, of an integral polytope $\mathcal P \subset \mathbb R^d$, where $\mathcal P = \{ z \in \mathbb R^d : Az \le b \}$.

For each vertex $e \in \mathcal E$, the associated \emph{normal cone} is
\(
\mathcal N_0(e)
:= \left\{ \zeta \in \mathbb R^d : \zeta^\top e \ge \zeta^\top e' \;\; \text{for all } e' \in \mathcal E \right\}=
\bigcap_{e' \in \mathcal E}
\left\{ \zeta \in \mathbb R^d : \zeta^\top (e-e') \ge 0 \right\}.
\)
The normal cone $\mathcal N_0(e)$ consists of all directions whose inner product with every displacement $e-e'$, $e'\in\mathcal E$, is nonnegative. 
Geometrically, $\mathcal N_0(e)$ is a closed convex polyhedral cone formed by the intersection of finitely many halfspaces and represents the set of directions for which $e$ maximizes the linear functional $\zeta^\top z$ over the polytope $\mathcal P$. The collection $\{\mathcal N_0(e)\}_{e \in \mathcal E}$ forms the \emph{normal fan} of the polytope $\mathcal P$, which induces a polyhedral partition of $\mathbb R^d$.

\begin{remark}[Geometry determines the response distribution]
\label{rmk:geometry}
For any $\mu \in \mathbb R^d$ and $\zeta \sim N(\mu, I_d)$, the distribution of $Y = T(\zeta)$ is completely determined by the polytope $\mathcal P$ and the mean $\mu$. Specifically, for each $e \in \mathcal E$,
\(
P(Y = e \mid \mu) = \int_{\mathcal N_0(e)} \phi_d(\zeta; \mu)\, d\zeta,
\)
where $\phi_d(\zeta; \mu) = (2\pi)^{-d/2} \exp(-\|\zeta - \mu\|_2^2/2)$ denotes the $d$-dimensional Gaussian density.
\end{remark}

Aside from the interior of the normal cones is a Lebesgue-null set, which is comprised of sets where at least two vertex values give the same maximum of $\zeta^Te$. To establish the posterior consistency, we would examine every possible direction along which the parameter is deviated from the ground truth, so those ``tie sets" need to be studied. 

Since in the proof we always focus on vertice, as opposed to the entire face of $\mathcal P$, we use the following simplified face selection mapping:
\(
\mathcal F(v):=\{e^\star\in\mathcal E:v^\top e^\star=\max_{e\in\mathcal E}v^\top e\},
\)
which is equivalent to the intersection between the face selection mapping and the set of vertices defined in the main text.

Here are some simple facts about face selection mapping and face sets. 
Since maximum is always achieved by at least one of $e\in\mathcal E$, a face set must be nonempty. 
Therefore, $\mathcal F(v)$ is a nonempty set of vertices in $\mathcal E$ that satisfies:
\begin{enumerate}
  \item[(a)] $v^\top(e-e')=0$ for all $e,e'\in \mathcal F(v)$;
  \item[(b)] $v^\top(e-e'')>0$ for all $e\in \mathcal F(v), e''\in \mathcal E\setminus\mathcal F(v)$.
\end{enumerate}
When $v\in\mathbb R^d\setminus C$, $T(v)$ is unique, and we have $\mathcal F(v)=\{T(v)\}$, or equivalently, $T(v)\in\mathcal F(v)$. Hence, face selection mapping $\mathcal F$ can be consider as the extension of argmax mapping $T$ onto $\mathbb R^d$.

For a vector $v\in\mathbb R^k(k\in\mathbb N^+)$, $\|v\|_2$ denotes its $L_2$ norm. For $\beta \in \mathbb R^{d\times p}$, we equip the parameter space with the Frobenius norm $\|\beta\|_F = \sqrt{\mathrm{tr}(\beta^\top \beta)}$. For some $\theta_0$ in a normed parameter space with norm $\|\cdot\|$ and $\varepsilon > 0$, we write the open $\varepsilon$-ball centered at $\theta_0$ as 
\(
B_\varepsilon(\theta_0)
:=
\bigl\{ \theta : \|\theta - \theta_0\| < \varepsilon \bigr\}
\)
and the unit sphere as
\(
\mathbb S^{k-1}
:=
\bigl\{ u \in \mathbb R^{k} : \|u\| = 1 \bigr\}.
\)

\subsection{Proof for Section~4.2.1}
\begin{proof}[Proof of Theorem~\ref{thm:law_consistency}]
The support of $P_{\mu}$ is invariant of $\mu$ because only $T$ determines the support. Suppose $\mathcal E_0\subset\mathcal E$ is the support of $P_{\mu}$ for any $\mu\in\mathbb R^d$.
The intercept-only model \eqref{eq:intercept_only} produces i.i.d.\ observations $Y_1, \ldots, Y_n$ from a discrete distribution on the finite set $\mathcal{E}_0$. The map $\mu \mapsto P(Y = e \mid \mu)$ is continuous for each $e \in \mathcal{E}_0$ (as an integral of a Gaussian density over a fixed polyhedral cone). Since $\mathcal{E}_0$ is finite, continuity in total variation follows.

By Condition~\ref{ass:prior}, the prior places positive mass on every neighborhood of $\mu_0$. Since $\mu \mapsto P(Y \mid \mu)$ is continuous, for any $\varepsilon > 0$ there exists $\delta > 0$ such that $\|\mu - \mu_0\|_2 < \delta$ implies $\|P(Y \mid \mu) - P(Y \mid \mu_0)\|_1 < \varepsilon$. Thus the prior places positive mass in every KL-neighborhood of $P(Y \mid \mu_0)$.

Applying Schwartz's theorem \citep{schwartz1965onbayes} for i.i.d.\ observations on a finite space, the posterior is consistent for the law $P(Y \mid \mu_0)$. The $L_1$-convergence of the posterior predictive $\int P(Y \mid \mu) \Pi(\mu \mid Y^n) d\mu$ to $P(Y \mid \mu_0)$ follows from dominated convergence, since all probability vectors on $\mathcal{E}_0$ are bounded.
\end{proof}

\subsection{Additional results and proofs for Section 4.2.2}

Proposition~\ref{prop:uniform_positivity} plays a key role in constructing
exponentially consistent tests for both models
\eqref{eq:intercept_only} and \eqref{eq:regression}.
Its full strength is mainly needed for the regression setting.
For the intercept-only model, a substantially simpler consequence (Corollary~\ref{cor:uniform_positivity}) suffices.

\begin{proposition}[Uniform Positivity of Probability Gains]
\label{prop:uniform_positivity}
Take any proper face set $F \subseteq \mathcal{E}$ and any $\delta>0$. Let $V(F,\delta) \subset \mathbb{R}^d$ be a set of vectors $v_0\in\mathbb R^d$ satisfying:
\begin{itemize}
\item[\textbf{(F1)}] $v_0^T(e-e') = 0$ for all $e,e' \in F$;
\item[\textbf{(F2)}] $v_0^T(e-e'')>\delta$ for all $e \in F$, $e'' \in \mathcal{E} \setminus F$.
\end{itemize}
Then $t \mapsto P(Y\in F|\mu+tv_0)$ is nondecreasing on $[0, \infty)$ for all $v_0\in V(F,\delta)$ and all $\mu\in\mathbb R^d$.
Moreover, for any $t_0 > 0$ and $L_0 > 0$, there exist $r^*=r^*(F,\delta,t_0) > 0$ and $c^*=c^*(F,\delta,t_0,L_0) > 0$ such that for all $t\ge t_0$,
\(
 \inf_{v_0\in V(F,\delta)}\inf_{\|\mu\|_2 \leq L_0} \inf_{\|v - v_0\|_2 < r^*} \big(P(Y\in F|\mu+tv) - P(Y\in F|\mu)\big) \geq c^*.
\)
\end{proposition}

\begin{corollary}
\label{cor:uniform_positivity}
For any $v_0 \in \mathbb{R}^d \setminus \{0\}$ and any $\mu_0\in\mathbb R^d$,
the function
\(
t \mapsto P(Y\in \mathcal F(v_0)\mid \mu_0 + t v_0)
\)
is nondecreasing on $[0,\infty)$.
Moreover, if $\mathcal F(v_0)\subsetneq \mathcal E$, then for any $t_0>0$,
there exists constant $r^*=r^*(v_0, t_0)>0$, $c^*=c^*(v_0,t_0)>0$ such that for all $t\ge t_0$,
\(
\inf_{v:\|v-v_0\|_2\le r^*}[P(Y\in \mathcal F(v_0)\mid \mu_0 + t v_0)
-
P(Y\in \mathcal F(v_0)\mid \mu_0)]
\ge c^*.
\)
\end{corollary}

\begin{proof}
Let $F=\mathcal F(v_0)$. By definition of the face selection mapping,
$v_0$ satisfies conditions (F1) and (F2) for some $\delta>0$.
The monotonicity follows directly from Proposition~\ref{prop:uniform_positivity}.
If $F\subsetneq\mathcal E$, the uniform positivity statement of
Proposition~\ref{prop:uniform_positivity} with $\mu=\mu_0$ yields the claimed
strict lower bound for all $t\ge t_0$.
\end{proof}

\begin{remark}
Corollary~\ref{cor:uniform_positivity} shows that for any nonzero direction
$v_0$, there exists a face set $\mathcal F(v_0)$ whose probability mass
strictly increases as the parameter moves in the direction $v_0$.
Consequently, in the identifiable intercept-only model, any deviation
$\mu-\mu_0$ induces a detectable and uniformly positive change in the
probability of $Y$ falling in the set $\mathcal F(\mu-\mu_0)$.
\end{remark}

\begin{proof}[Proof of Proposition~\ref{prop:uniform_positivity}]
For any $v \in \mathbb{R}^d$ and $t \geq 0$, let
\(
S_F(t; v) := \bigcup_{e \in F} \bigcap_{e' \in \mathcal{E} \setminus F} \{z \in \mathbb{R}^d : (z + tv)^T(e-e'')\geq 0\}.
\)
For $Z_\mu \sim N(\mu, I_d)$, let $p_{v,\mu}(t) := P(Z_\mu \in S_F(t; v))$. We first prove the claim that
\[
\label{eq:claim_P_equal_p}
P(Y\in F|\mu+tv)=p_{v,\mu}(t)
\]

By definition:
\(
P(Y\in F|\mu+tv)
=P\big(Z_\mu \in \bigcup_{e \in F} \bigcap_{e'' \not = e} \{z \in \mathbb{R}^d : (z + tv)^T(e-e'')\geq 0\}\big).
\)
It suffices to prove the following set equivalence:
\(
\bigcup_{e \in F} \bigcap_{e'' \in \mathcal{E} \setminus F} \{z \in \mathbb{R}^d : (z + tv)^T(e-e'')\geq 0\}=\bigcup_{e \in F} \bigcap_{e'' \not= e} \{z \in \mathbb{R}^d : (z + tv)^T(e-e'')\geq 0\}.
\)
Clearly, RHS is a subset of LHS, so we only need to prove LHS is also a subset of RHS. For $z$ belongs to LHS, there exists $e\in F$, $(z+tv)^Te\ge \max_{e''\in\mathcal E\setminus F}(z+tv)^Te''$. In addition, there must exist $e^\star\in F$ such that $(z+tv)^Te^\star= \max_{e\in F}(z+tv)^Te$. Hence, $(z+tv)^T(e^\star-e)\ge0$ for all $e\not=e^\star$, thus $z$ belonging to RHS, finishing the proof of the claim.

Define the subspace
\(
\mathcal S_F:=\text{Span}\{e-e'':\,e\in F,\ e''\in\mathcal E\setminus F\}\subseteq\mathbb R^d
\)
and let $\Pi_F$ be the orthogonal projection onto $\mathcal S_F$.
Since each constraint in $S_F(t;v)$ has normal vector $e-e''\in\mathcal S_F$, we have for every $t\ge0$ and $v\in\mathbb R^d$,
\[
\label{eq:proj_invariance}
S_F(t;v)=S_F(t;\Pi_F v),\qquad\text{and hence}\qquad p_{v,\mu}(t)=p_{\Pi_F v,\mu}(t).
\]

Fix $v_0\in V(F,\delta)$ and $\mu\in\mathbb R^d$. For $e\in F$ and $e''\in\mathcal E\setminus F$, define
\(
H_{e,e''}(t;v_0):=\{z\in\mathbb R^d:(z+tv_0)^\top(e-e'')\ge 0\}.
\)
By (F2), $v_0^\top(e-e'')>0$, so $t_1\le t_2$ implies $H_{e,e''}(t_1;v_0)\subseteq H_{e,e''}(t_2;v_0)$.
Taking intersections over $e''\in\mathcal E\setminus F$ and unions over $e\in F$ yields $S_F(t_1;v_0)\subseteq S_F(t_2;v_0)$.
Therefore $t\mapsto p_{v_0,\mu}(t)$ is nondecreasing on $[0,\infty)$, and so is $t\mapsto P(Y\in F|\mu+tv_0)$.

Next, we prove the uniform positive gap between $P(Y\in F|\mu+tv_0)$ and $P(Y\in F|\mu)$ for all $t\ge t_0$.

Let $M:=\max_{e,e'\in\mathcal E}\|e-e'\|_2<\infty$ and set
\(
r^*:=\frac{\delta}{2M},\qquad \rho:=\frac{t_0\delta}{8M}.
\)
Fix $v_0\in V(F,\delta)$, $\|\mu\|_2\le L_0$, and any $v$ with $\|v-v_0\|_2<r^*$.
Then for all $e\in F$ and $e''\in\mathcal E\setminus F$,
\[
\label{eq:margin_v}
v^\top(e-e'') \ge v_0^\top(e-e'')-\|v-v_0\|_2\|e-e''\|_2
> \delta-r^*M=\delta/2.
\]
Define the center point $z^*:=-\frac{t_0}{2}\,\Pi_F v\in\mathcal S_F$.
Let $z$ satisfy $\|z-z^*\|_2\le \rho$. Using \eqref{eq:proj_invariance} and \eqref{eq:margin_v}, for any $e\in F$ and $e''\in\mathcal E\setminus F$,
\(
(z+t_0 v)^\top(e-e'')&=(z+t_0\Pi_F v)^\top(e-e'')\\
&\ge (z^*+t_0\Pi_F v)^\top(e-e'')-\rho\|e-e''\|_2\\
&= \frac{t_0}{2}(\Pi_F v)^\top(e-e'')-\rho M\\
&\ge \frac{t_0\delta}{4}-\frac{t_0\delta}{8}=\frac{t_0\delta}{8}>0,
\)
so $z\in S_F(t_0;v)$. Similarly,
\(
z^\top(e-e'') \le (z^*)^\top(e-e'')+\rho\|e-e''\|_2
&= -\frac{t_0}{2}(\Pi_F v)^\top(e-e'')+\rho M\\
&\le -\frac{t_0\delta}{4}+\frac{t_0\delta}{8}=-\frac{t_0\delta}{8}<0,
\)
so $z\notin S_F(0;v)$. Hence
\[
\label{eq:ball_in_diff}
B(z^*,\rho)\subset S_F(t_0;v)\setminus S_F(0;v),
\qquad\Rightarrow\qquad
p_{v,\mu}(t_0)-p_{v,\mu}(0)\ge P(Z_\mu\in B(z^*,\rho)).
\]

It remains to lower bound the Gaussian ball probability uniformly.
Define the margin functional on $\mathcal S_F$,
\(
m_F(w):=\min_{e\in F,\ e''\in\mathcal E\setminus F} w^\top(e-e'').
\)
By \eqref{eq:proj_invariance}, $p_{v,\mu}(t)=p_{\Pi_F v,\mu}(t)$, and by \eqref{eq:margin_v} we have $m_F(\Pi_F v)\ge \delta/2>0$.
For any $w\in\mathcal S_F$ with $m_F(w)>0$, define the normalized vector
\(
\bar w := \frac{\delta}{m_F(w)}\,w\in\mathcal S_F,
\quad\text{so that}\quad m_F(\bar w)=\delta.
\)
Then $S_F(t;w)=S_F\!\big(t\,m_F(w)/\delta;\,\bar w\big)$, hence for all $\mu$,
\[
\label{eq:time_scaling}
p_{w,\mu}(t)=p_{\bar w,\mu}\!\Big(t\,\frac{m_F(w)}{\delta}\Big).
\]
Since $m_F(w)/\delta\ge 1/2$ and $t\mapsto p_{\bar w,\mu}(t)$ is nondecreasing, we have in particular
\(
p_{w,\mu}(t_0)-p_{w,\mu}(0)\ge p_{\bar w,\mu}(t_0/2)-p_{\bar w,\mu}(0).
\)
The set of normalized directions
\(
\mathcal K_{F,\delta}:=\Big\{w\in\mathcal S_F:\ w^\top(e-e')=0\ \forall e,e'\in F,\ \ m_F(w)=\delta\Big\}
\)
is compact (it is a bounded slice of the pointed normal cone associated with $F$ inside $\mathcal S_F$).
Therefore the map $(\mu,w)\mapsto p_{w,\mu}(t_0/2)-p_{w,\mu}(0)$ is continuous and strictly positive on
$\{\|\mu\|_2\le L_0\}\times\mathcal K_{F,\delta}$, hence attains a positive minimum:
\(
c^*:=\inf_{\|\mu\|_2\le L_0}\ \inf_{w\in\mathcal K_{F,\delta}}\ \big(p_{w,\mu}(t_0/2)-p_{w,\mu}(0)\big)\ >\ 0.
\)
Combining with \eqref{eq:claim_P_equal_p}, \eqref{eq:proj_invariance} and \eqref{eq:time_scaling} yields the claimed uniform lower bound at time $t_0$
(up to replacing $t_0$ by $t_0/2$ in the construction), completing the proof.

\end{proof}

\begin{proof}[Proof of Theorem~\ref{thm:mu_identifiability}]
We first characterize when two parameters induce the same distribution.

Let $v=\mu-\mu'$. Suppose $v\in \text{Span}\{e-e':e,e'\in\mathcal E\}^{\perp}$. Then for all $e,e'\in\mathcal E$, $v^\top(e-e')=0$.
It follows that for any $e\in\mathcal E$,
\(
P_{\mu'+v}(e)
&=\int \phi(\zeta;\mu'+v,I_d)
    1\{\zeta^\top(e-e')\ge0,\ \forall e'\in\mathcal E\}\,d\zeta \\
&=\int \phi(\zeta;\mu',I_d)
    1\{(\zeta-v)^\top(e-e')\ge0,\ \forall e'\in\mathcal E\}\,d\zeta \\
&=\int \phi(\zeta;\mu',I_d)
    1\{\zeta^\top(e-e')\ge0,\ \forall e'\in\mathcal E\}\,d\zeta \\
&=P_{\mu'}(e),
\)
where the second equality follows from a change of variables and the third from
$v^\top(e-e')=0$ for all $e,e'\in\mathcal E$. Hence $P_{\mu'+v}=P_{\mu'}$.

Conversely, suppose $v\notin \text{Span}\{e-e':e,e'\in\mathcal E\}^{\perp}$. Then there exist $e,e'\in\mathcal E$ such that
$v^\top(e-e')\neq0$. Without loss of generality, assume $v^\top(e-e')>0$. This implies that the face set
$\mathcal F(v)=\arg\max_{e\in\mathcal E} v^\top e$ is a proper subset of $\mathcal E$.
By Corollary~\ref{cor:uniform_positivity}, there exists $c^*>0$ such that
\[
P_{\mu'+v}\big(\mathcal F(v)\big)
-
P_{\mu'}\big(\mathcal F(v)\big)
\ge c^*,
\]
which implies $P_{\mu'+v}\neq P_{\mu'}$. This establishes the stated equivalence.

Condition~\ref{ass:polytope} holds if and only if $\text{Span}\{e-e':e,e'\in\mathcal E\}=\mathbb R^d$.
In this case,
$\text{Span}\{e-e':e,e'\in\mathcal E\}^{\perp}=\{0\}$, and hence
$P_\mu=P_{\mu'}$ implies $\mu=\mu'$. Therefore, the model is identifiable.
If Condition~\ref{ass:polytope} fails, then
$\text{Span}\{e-e':e,e'\in\mathcal E\}^{\perp}\neq\{0\}$, and the model is identifiable only up to shifts in directions belonging to this orthogonal complement.
\end{proof}

Next we prove the posterior consistency theorem for $\mu$. Here is the proof sketch:
For any $\mu\neq\mu_0$, let $v_0=\mu-\mu_0$. By Condition~\ref{ass:polytope}, there exists a proper face set
$\mathcal F(v_0)\subsetneq\mathcal E$ and positive constants $r(\mu)$ and $c(\mu)$ such that $\inf_{v:\|v-v_0\|_2\le r(\mu)}[P_{\mu}\big(\mathcal F(v_0)\big)-P_{\mu_0}\big(\mathcal F(v_0)\big)]>c(\mu)$ for all $\mu=\mu_0+tv_0$, $t\ge 1$.
This yields a test based on the empirical frequency of
$\{Y_i\in\mathcal F(v_0)\}$ with exponentially decaying type-I and type-II
errors.
Together with the prior positivity condition in Condition~\ref{ass:prior},
the posterior consistency follows from the adapted Schwartz theorem.

\begin{proof}[Proof of Theorem~\ref{thm:consistency_mu}]
Fix $\epsilon>0$ and define the complement of an $\epsilon$-ball around $\mu_0$,
\(
B_\epsilon^c:=\{\mu\in\mathbb R^d:\|\mu-\mu_0\|_2\ge \epsilon\}.
\)
We prove that
\(
\Pi(B_\epsilon^c\mid Y^n)\to 0
\)
$P_{\mu_0}^\infty$-almost surely, which implies the stated convergence for any
open neighborhood $V_{\mu_0}$ of $\mu_0$.

Take any $\mu\in B_\epsilon^c$ and let $v_0=\mu-\mu_0\neq 0$.
Under Condition~\ref{ass:polytope}, we have
\(
\text{Span}\{e-e':e,e'\in\mathcal E\}=\mathbb R^d,
\)
hence there exist $e,e'\in\mathcal E$ with $v_0^\top(e-e')\neq 0$ and therefore
\(
\mathcal F(v_0):=\arg\max_{e\in\mathcal E} v_0^\top e
\subsetneq \mathcal E.
\)
By Corollary~\ref{cor:uniform_positivity}, applied with the direction $v_0$ and
base point $\mu_0$, there exists constants $r(\mu)>0$ and $c(\mu)>0$ such that
\(
\inf_{v:\|v-v_0\|_2\le r(\mu)}[P_{\mu}\big(\mathcal F(v_0)\big)-P_{\mu_0}\big(\mathcal F(v_0)\big)]\ge c(\mu)>0.
\)

To obtain a test that is uniform over $B_\epsilon^c$, we use compactness of the
direction space. Let
\(
\mathbb S^{d-1}:=\{u\in\mathbb R^d:\|u\|_2=1\}.
\)
Since $B_\epsilon^c$ can be parameterized by $(t,u)$ with $t\ge\epsilon$ and
$u\in\mathbb S^{d-1}$ via $\mu=\mu_0+tu$, it suffices to construct tests
uniformly over directions $u\in\mathbb S^{d-1}$ and magnitudes $t\ge\epsilon$.

By Condition~\ref{ass:polytope}, for each $u\in\mathbb S^{d-1}$ the face set
$\mathcal F(u)$ is proper. Fix $\delta>0$ and cover $\mathbb S^{d-1}$ by finitely
many balls $\{B(u_j,r_j)\}_{j=1}^J$ such that on each ball there exists a proper
face set $F_j\subsetneq\mathcal E$ with $u_j\in V(F_j,\delta)$ (as in
Proposition~\ref{prop:uniform_positivity}). Applying the uniform positivity part
of Proposition~\ref{prop:uniform_positivity} with $t_0=\epsilon$ and $L_0=\|\mu_0\|_2+1$,
there exist constants $c_\epsilon>0$ and $r_\epsilon>0$ (depending on
$\epsilon$ but not on $j$) such that for any $j\in\{1,\dots,J\}$, any
$u\in B(u_j,r_j)$, and any $t\ge \epsilon$,
\(
P_{\mu_0+tu}(Y\in F_j)-P_{\mu_0}(Y\in F_j)\ge c_\epsilon.
\)

For each $j\in\{1,\dots,J\}$ define the empirical frequency statistic
\(
\widehat p_{n,j}:=\frac1n\sum_{i=1}^n \mathbf 1\{Y_i\in F_j\}.
\)
Let
\(
p_{0,j}:=P_{\mu_0}(Y\in F_j).
\)
Define the test
\(
\phi_{n,j}:=\mathbf 1\Big\{\widehat p_{n,j}-p_{0,j}\ge \frac{c_\epsilon}{2}\Big\}.
\)
Under $P_{\mu_0}$, $\widehat p_{n,j}$ is the mean of i.i.d.\ Bernoulli random
variables with mean $p_{0,j}$, hence by Hoeffding's inequality,
\(
P_{\mu_0}(\phi_{n,j}=1)\le \exp\big(-2n(c_\epsilon/2)^2\big).
\)
Now take any $\mu=\mu_0+tu$ with $t\ge \epsilon$ and $u\in B(u_j,r_j)$.
Then $P_\mu(Y\in F_j)\ge p_{0,j}+c_\epsilon$. Again by Hoeffding's inequality,
\(
P_\mu(\phi_{n,j}=0)
=
P_\mu\Big(\widehat p_{n,j}-P_\mu(Y\in F_j)\le -\frac{c_\epsilon}{2}\Big)
\le \exp\big(-2n(c_\epsilon/2)^2\big).
\)
Therefore, for each patch $j$ the test $\phi_{n,j}$ has exponentially decaying
type-I error and exponentially decaying (uniform) type-II error over the
corresponding alternative set.

\(
\phi_n:=\max_{1\le j\le J}\phi_{n,j}.
\)
By the union bound,
\(
P_{\mu_0}(\phi_n=1)\le J\exp\big(-2n(c_\epsilon/2)^2\big),
\)
which is still exponentially small. Moreover, for any $\mu\in B_\epsilon^c$,
there exists some $j$ such that $\mu=\mu_0+tu$ with $t\ge\epsilon$ and
$u\in B(u_j,r_j)$; hence
\(
P_\mu(\phi_n=0)\le P_\mu(\phi_{n,j}=0)\le \exp\big(-2n(c_\epsilon/2)^2\big).
\)
Thus $\{\phi_n\}$ is an exponentially consistent sequence of tests for
$H_0:\mu=\mu_0$ versus $H_1:\mu\in B_\epsilon^c$.

To sum up, Condition~\ref{ass:prior} implies that $\Pi$ assigns positive mass to every
neighborhood of $\mu_0$, and in particular to the Kullback--Leibler neighborhood
required by the Schwartz theorem (\cite{schwartz1965onbayes}).
Combining this prior positivity with the exponentially consistent tests
constructed above yields
\(
\Pi(B_\epsilon^c\mid Y^n)\to 0
\)
$P_{\mu_0}^\infty$-almost surely. This step of proof is standard Schwartz Theorem application so we omit the details.
Since $\epsilon>0$ is arbitrary, the posterior is consistent at $\mu_0$.
\end{proof}

\subsection{Additional results and proofs for Section 4.2.3}

\begin{lemma}[KL support and denominator control]
\label{lem:KL_condition}
Under Conditions~\ref{ass:prior} and \ref{ass:covariates_bound}, Condition~(ii) in Theorem~\ref{thm:adapt_sch} holds.
\end{lemma}

\begin{proof}[Proof of Lemma~\ref{lem:KL_condition}]
Suppose $\mathcal E_0\subset\mathcal E$ is the support of $P_{\mu}$ for any $\mu\in\mathbb R^d$. Clearly, for any $e\in\mathcal E_0$, $P_\mu(e)$ is bounded away from both zero and one for $\mu$ in a compact set in $\mathbb R^d$.
It follows that both of the following functions, defined on $\mathbb{R}^d\times\mathbb{R}^d$, are continuously differentiable with respect to $(\mu,\mu^*)\in\mathbb{R}^d\times\mathbb{R}^d$:
\(
f_1(\mu,\mu^*)
&=\sum_{e\in\mathcal E_0}P(Y=e\mid \mu^*)\log\frac{P(Y=e\mid \mu^*)}{P(Y=e\mid \mu)},\\
f_2(\mu,\mu^*)&=\sum_{e\in\mathcal E_0}P(Y=e\mid \mu^*)\left[\log\frac{P(Y=e\mid \mu^*)}{P(Y=e\mid \mu)}\right]^2.
\)
Define
\(
\mathcal{K}_{\tau}:=\{v\in\mathbb{R}^d:\|v\|_2\le\tau\} \quad \text{for } \tau\ge0.
\)
Since $\mathcal{K}_\tau\times\mathcal K_\tau$ is compact, both $f_1$ and $f_2$ are Lipschitz continuous on $\mathcal{K}_\tau\times\mathcal K_\tau$. Therefore, there exists a constant $C_\tau>0$ such that for all $(\mu,\mu^*)\in\mathcal{K}_\tau\times\mathcal K_\tau$,
\(
|f_1(\mu,\mu^*)|\vee|f_2(\mu,\mu^*)|&= |f_1(\mu,\mu^*)-f_1(\mu^*,\mu^*)|\vee |f_2(\mu,\mu^*)-f_2(\mu^*,\mu^*)|\\
&\le C_\tau\|\mu-\mu^*\|_2.
\)

Fix $\beta_0\in\mathbb R^{d\times p}$ and $\delta_0>0$. By Condition~\ref{ass:covariates_bound}, $\sup_{i\ge 1}\|\beta x_i-\beta_0x_i\|_2\le \ell\|\beta-\beta_0\|_F$. Taking $\tau_0=\ell(1+\|\beta_0\|_F)$, we have for every $\varepsilon\in(0,1)$: if $\|\beta-\beta_0\|_F<\varepsilon$, then $f_1(\beta x_i,\beta_0x_i)\vee f_2(\beta x_i,\beta_0 x_i)<\varepsilon C_{\tau_0}\ell$ for every $i\ge 1$. Taking $\varepsilon_0=(\delta_0\wedge 1)/(2C_{\tau_0}\ell \vee 1)\in(0,1)$, we obtain
\(
\{\beta:\|\beta-\beta_0\|_F<\varepsilon_0\}\subset\left\{\beta:\sup_{i\ge 1}\mathrm{KL}(Q_{i,\beta_0}\|Q_{i,\beta})<\delta_0,\;\sum_{i=1}^\infty \frac{V(Q_{i,\beta_0}, Q_{i,\beta})}{i^2}<\infty\right\}.
\)
Since the set on the left has positive measure under $\pi$ by Condition~\ref{ass:prior}, condition~(ii) of Theorem~\ref{thm:adapt_sch} is verified.
\end{proof}

\begin{lemma}[Local exponentially consistent tests]
\label{lem:local_exponentially_consistent_tests}
Fix $u_0 \in \mathbb{S}^{dp-1}$ and $t_0>0$.
Under Condition~\ref{ass:anti_boundary}, there exists a radius $r_0>0$, such that there exists an exponentially consistent sequence of tests for testing
\(
H_0:\beta=\beta_0\text{ vs }H_1:\beta=\beta_0+tu,u\in B(u_0,r_0)\cap\mathbb S^{dp-1},t\ge t_0.
\)
\end{lemma}

\begin{proof}[Proof of Lemma~\ref{lem:local_exponentially_consistent_tests}]
Fix $u_0\in\mathbb S^{dp-1}$ and $t_0>0$.
By Condition~\ref{ass:anti_boundary}, there exist a proper face set $F_0\subsetneq\mathcal E$,
constant $\delta_0>0$, such that 
\(
\liminf_{n\to\infty}\frac{1}{n}\sum_{i=1}^n 
1\Big\{
F(u_0x_i)=F_0,\ 
\min_{e\in F_0,\ e''\in\mathcal E\setminus F_0} (u_0x_i)^\top(e-e'')\ge \delta_0
\Big\}>0.
\)
Define
\(
\mathcal X_{u_0}:=\Big\{x_i:\ 
F(u_0x_i)=F_0,\ 
\min_{e\in F_0,\ e''\in\mathcal E\setminus F_0} (u_0x_i)^\top(e-e'')\ge \delta_0
\Big\}.
\)
Then 
\[
\label{eq:informative_subset}
\rho_0:=\liminf_{n\to\infty}n^{-1}\sum_{i=1}^n 1\{x_i\in\mathcal X_{u_0}\}>0.
\]

Let $L_0:=\ell\|\beta_0\|_F$. For any $x\in\mathcal X_{u_0}$, set $v_0:=u_0x$ and $\mu:=\beta_0x$.
By Condition~\ref{ass:covariates_bound}, $\|\mu\|_2\le \|\beta_0\|_F\|x\|_2\le L_0$.
Also, by construction, $v_0\in V(F_0,\delta_0)$, so Proposition~\ref{prop:uniform_positivity} yields constants
$r_0^*>0$ and $c_0^*>0$ such that for any $v$ with $\|v-v_0\|_2<r_0^*$, and any $t\ge t_0$
\(
P(Y\in F_0|\mu+tv)-P(Y\in F_0|\mu)\ge c_0^*.
\)
Set $r_0:=r_0^*/\ell$ and take any $u\in B(u_0,r_0)\cap\mathbb S^{dp-1}$. Then
$\|ux-u_0x\|_2\le \|u-u_0\|_F\|x\|_2<r_0\ell=r^*$, hence
\(
P(Y\in F_0|\beta_0x+tux)-P(Y\in F_0|\beta_0x)\ge c_0^*.
\)
Consider the test $\Psi_{n,u_0}:\mathcal E^n\to[0,1]$:
\(
\Psi_{n,u_0}(Y^n):=1\left(\sum_{i=1}^n\psi_i>\sum_{i=1}^n\alpha_i\right),
\)
where
\(
\psi_i:=1(Y_i\in F_0)1(x_i\in\mathcal X_{u_0}),\quad \alpha_i:=[Q_{i,\beta_0}(F_0)+c^*_0/2]1(x_i\in\mathcal X_{u_0}).
\)
Then by invoking Hoeffding's inequality \citep[pp.~14]{Dudley1999uniform} twice on $\Psi_{n,u_0}$ and $1-\Psi_{n,u_0}$, respectively, we get
\(
\mathbb E_{Q^n_{\beta_0}}[\Psi_{n,u_0}(Y^n)]=&Q^n_{\beta_0}\left(\sum_{i=1}^n[\psi_i-\mathbb E_{Q_{i,\beta_0}}\psi_i]> |\mathcal X_{u_0}|c^*_0/2\right)\\
\le&\exp\left(-\frac{1}{4}n(c_0^*)^2\left(\frac{|\mathcal X_{u_0}|}{n}\right)^2\right),
\)
and for $\forall \beta=\beta_0+tu$ with $u\in B(u_0,r_0)\cap\mathbb S^{dp-1}$ and $t\ge t_0$,
\(
\mathbb E_{Q^n_{\beta}}[1-\Psi_{n,u_0}(Y^n)]=&Q^n_{\beta}\left(\sum_{i=1}^n[(1-\psi_i)-(1-\mathbb E_{Q_{i,\beta}}\psi_i)]>\sum_{i=1}^n[-\alpha_i+\mathbb E_{Q_{i,\beta}}\psi_i]\right)\\
\le&Q^n_{\beta}\left(\sum_{i=1}^n[(1-\psi_i)-(1-\mathbb E_{Q_{i,\beta}}\psi_i)]>|\mathcal X_{u_0}|c^*_0/2\right)\\
\le&\exp\left(-\frac{1}{4}n(c_0^*)^2\left(\frac{|\mathcal X_{u_0}|}{n}\right)^2\right).
\)
Let $b(u_0):=(c_0^*\rho_0)^2/4>0$ with $\rho_0$ being defined in \eqref{eq:informative_subset}. Then for sufficiently large $n$,
\(
\mathbb E_{Q^n_{\beta_0}}[\Psi_{n,u_0}(Y^n)]\le \exp(-nb(u_0)),\qquad \inf_{\beta:\beta=\beta_0+tu}\mathbb E_{Q^n_{\beta}}[1-\Psi_{n,u_0}(Y^n)]\le\exp(-nb(u_0)).
\)

\end{proof}

With both Lemma~\ref{lem:KL_condition} and Lemma~\ref{lem:local_exponentially_consistent_tests}, we can prove the posterior consistency for the regression model.
\begin{proof}[Proof of Theorem~\ref{thm:consistency_beta}]
Invoking Lemma~\ref{lem:local_exponentially_consistent_tests}, we get an opn cover of the compact sphere $\mathcal S^{dp-1}$:
\(
\{B(u_0,r(u_0))\cap\mathbb S^{dp-1}:u_0\in\mathbb S^{dp-1}\},
\)
where on each open cover, there exists an exponentially consistent sequence of tests $\{\psi_{n,u_0}\}$ such that $\mathbb E_{Q^n_{\beta_0}}[\psi_{n,u_0}(Y^n)]\le \exp(-nb(u_0))$ and $\mathbb E_{Q^n_{\beta}}[1-\psi_{n,u_0}(Y^n)]\le\exp(-nb(u_0))$ for all $\beta=\beta_0+tu$ with $t\ge t_0$ and $u\in B(u_0,r(u_0))\cap\mathbb S^{dp-1}$. Hence, there exist $u_1,\ldots,u_M\in\mathbb S^{dp-1}$ such that
\(
\mathbb S^{dp-1}\subset\bigcup_{m=1}^M(B(u_m,r(u_m))\cap\mathbb S^{dp-1}).
\)
Write $r_m:=r(u_m)$, $\psi_{n,m}:=\psi_{n,u_m}$, and $b:=\min_{1\le m\le M}b(u_m)$.

Define the global test $\Psi_n:=\max_{1\le m\le M}\psi_{n,u_m}$.
Then by the union bound and the local type I error bounds,
\(
\mathbb E_{\beta_0}[\Psi_n]\le \sum_{m=1}^M \mathbb E_{\beta_0}[\psi_{n,m}]
\le M\max_{1\le m\le M}\exp(-b(u_m)n)\ \le\ M e^{-bn}.
\)

For the type II error, fix any $u\in\mathbb S^{dp-1}$. Choose $m$ such that $u\in B(u_m,r_m)\cap\mathbb S^{dp-1}$.
Then $\Psi_n\ge \psi_{n,u_m}$ and therefore
\(
\sup_{t\ge t_0}\mathbb E_{\beta_0+t u}[1-\Psi_n]\le \sup_{t\ge t_0}\mathbb E_{\beta_0+t u}[1-\psi_{n,m}]
\le \exp(-b(u_m)n)\ \le\ C e^{-bn},
\)
Taking supremum over $u\in\mathbb S^{dp-1}$ gives the uniform bound:
\(
\sup_{\beta:\|\beta-\beta_0\|_F\ge t_0}\mathbb E_{Q^n_{\beta}}[1-\Psi_n]\le \exp(-bn).
\)
Hence, we have verified condition (ii) in Theorem~\ref{thm:adapt_sch}. Combining it with Lemma~\ref{lem:KL_condition}, we have proved the posterior consistency of $\beta$ for the regression model \eqref{eq:regression}.
\end{proof}

\section{Posterior sampling in data application}

We propose a posterior sampling scheme for the following hierarchical model, which will be applied in Section~\ref{sec:data_application} to sample the parameters $\beta_0$, $\beta_1$, and $\beta_2$:
\(
&\zeta \sim \text{Mat-N}(\mu, I_n, I_d), \\
&\mu = \sum_{m=1}^M \alpha_m W_m + B\beta C, \\
&\alpha=(\alpha_1,\ldots,\alpha_M)^\top \sim \text{N}(0, V_\alpha), \\
&\beta \sim \text{Mat-N}\left(0, I_\kappa, (\tau_\beta C^\top C)^{-1} \right),
\)
where $W_m \in \mathbb{R}^{n \times d}$, $B \in \mathbb{R}^{n \times \kappa}$, $\beta \in \mathbb{R}^{\kappa \times K}$, and $C \in \mathbb{R}^{K \times d}$ is a matrix whose columns are one-hot vectors. 

For each $k = 1, \ldots, K$, define the index set $J_k := \{j : C_{k,j} = 1\}$, and let $s_k := |J_k|$ denote its cardinality. Then, $C^\top C = \text{diag}(s_1, \ldots, s_K)$, and we assume $s_k > 0$ for all $k$. Our goal is to infer the parameters $(\alpha, \beta)$.

We employ a MH-Within-Gibbs sampling algorithm that alternately samples the regression coefficients $\alpha$ and the matrix $\beta$ from their respective full conditional distributions. The steps are as follows:

\begin{itemize}
\item Sample $\beta$ given $\alpha$: \\
Define the matrix 
\(
Q := (B^\top B + \tau_\beta I_\kappa)^{-1} B^\top \left(\zeta - \sum_{m=1}^M \alpha_m W_m\right),
\)
and let $Q_j$ denote the $j$-th column of $Q$. Then, the columns of $\beta = (\beta_1, \ldots, \beta_K)$ are conditionally independent and sampled as
\(
\beta_k \sim \text{N} \left( \frac{1}{s_k} \sum_{j \in J_k} Q_j, \, \frac{1}{s_k}(B^\top B + \tau_\beta I_\kappa)^{-1} \right).
\)

\item Sample $\alpha$ given $\beta$: \\
Let $\Gamma \in \mathbb{R}^{M \times M}$ be the Gram matrix with entries $\Gamma_{s,t} = \operatorname{trace}(W_s^\top W_t)$, and define the vector $\gamma \in \mathbb{R}^M$ with entries 
\(
\gamma_m = \operatorname{trace}\left(W_m^\top (\zeta - B \beta C)\right).
\)
Then, $\alpha$ is sampled from
\(
\alpha \sim \text{N} \left((\Gamma + V^{-1}_\alpha)^{-1} \gamma, \, (\Gamma + V^{-1}_\alpha)^{-1} \right).
\)
\end{itemize}

\section{Additional results for the data application}
Figure~\ref{plot:matching_probs_2in1} shows the estimated matching probabilities over time for both sexes to facilitate comparison. Overall, the matching probabilities for males are generally higher than or equal to those for females.

\begin{figure}[H]
\centering
\begin{subfigure}[t]{\textwidth}
    \begin{overpic}[width=\textwidth]{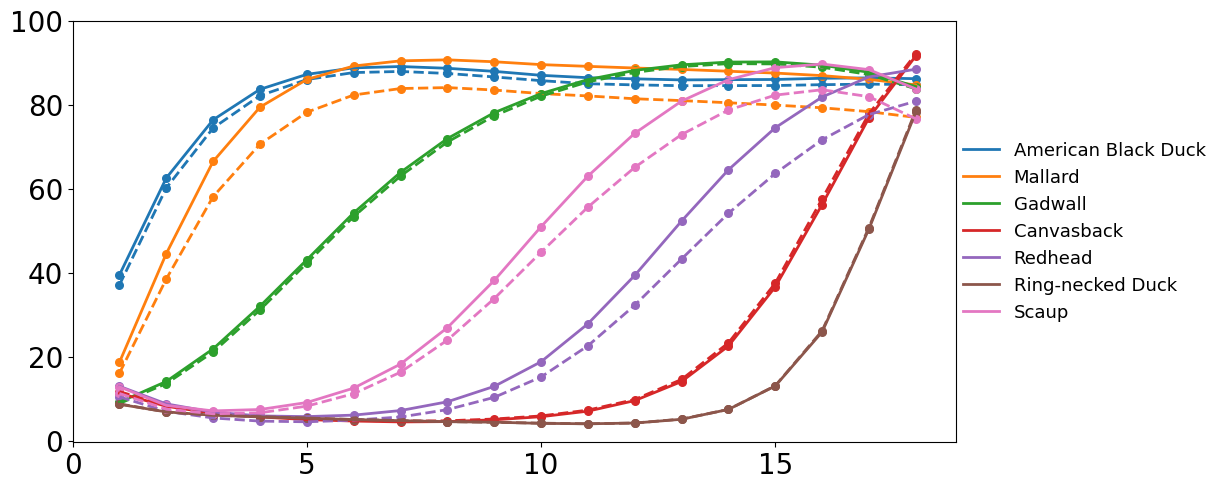}
    \put(42, -1){\scriptsize Week}
    \put(-1,10){\rotatebox{90}{\scriptsize Matching Probability (\%)}}
    \end{overpic}
\end{subfigure}
\caption{Matching probabilities over time. Solid lines correspond to male ducks of medium weight, while dashed lines correspond to female ducks of medium weight.}
\label{plot:matching_probs_2in1}
\end{figure}

Figure~\ref{plot:reduced_matching_probs_with_error_band} shows the estimated matching probabilities over time for the two-group model. Dabbling ducks tend to form matches earlier than diving ducks.

\begin{figure}[H]
\centering
\begin{subfigure}[t]{\textwidth}
    \begin{overpic}[width=\textwidth]{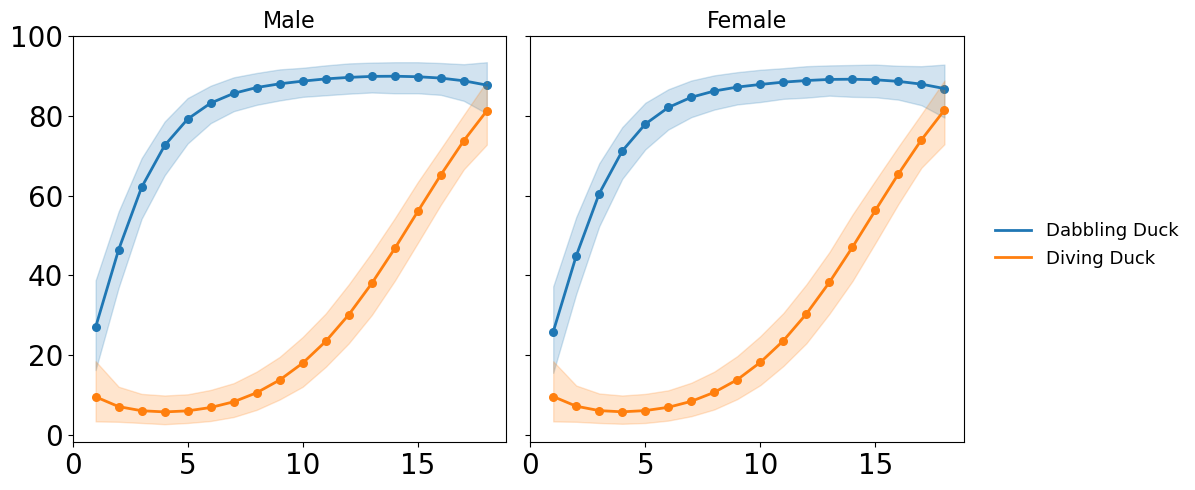}
    \put(42, -1){\scriptsize Week}
    \put(-1,10){\rotatebox{90}{\scriptsize Matching Probability (\%)}}
    \end{overpic}
\end{subfigure}
\caption{Matching probabilities over time with 95\% credible bands for the two-group model.}
\label{plot:reduced_matching_probs_with_error_band}
\end{figure}

\section{Additional results for the simulation}
Figure~\ref{plot:beta_violin_more} contains the violin plots of the posterior distributions of $\beta_{i,j}$ in two relatively high-dimensional settings.

\begin{figure}[H]
\centering
\begin{subfigure}[t]{\textwidth}
    \begin{overpic}[width=\textwidth,height=4cm]{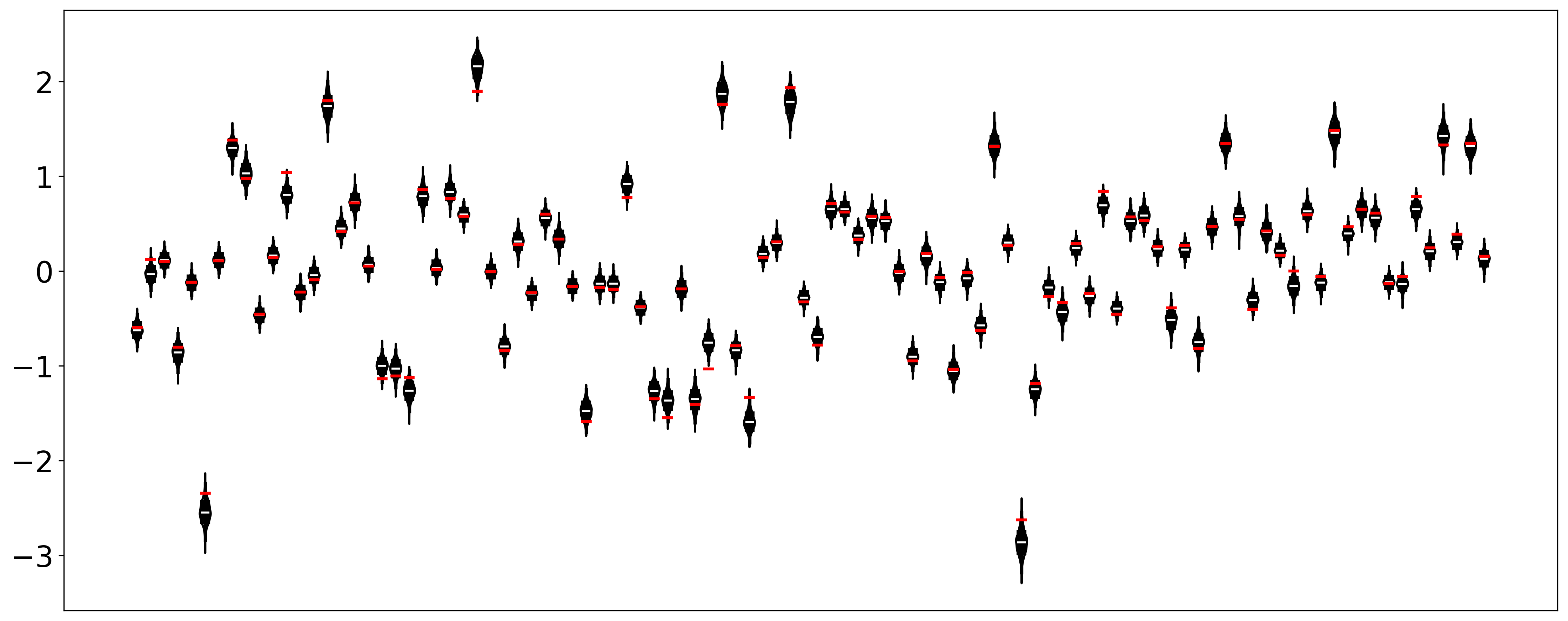}
    \put(43, -1){\scriptsize Parameter}
    \end{overpic}
    \caption{Posterior distributions of $\beta_{i,j}$ under $(d,m) = (20,10)$.}
\end{subfigure}\;
\begin{subfigure}[t]{\textwidth}
    \begin{overpic}[width=\textwidth,height=4cm]{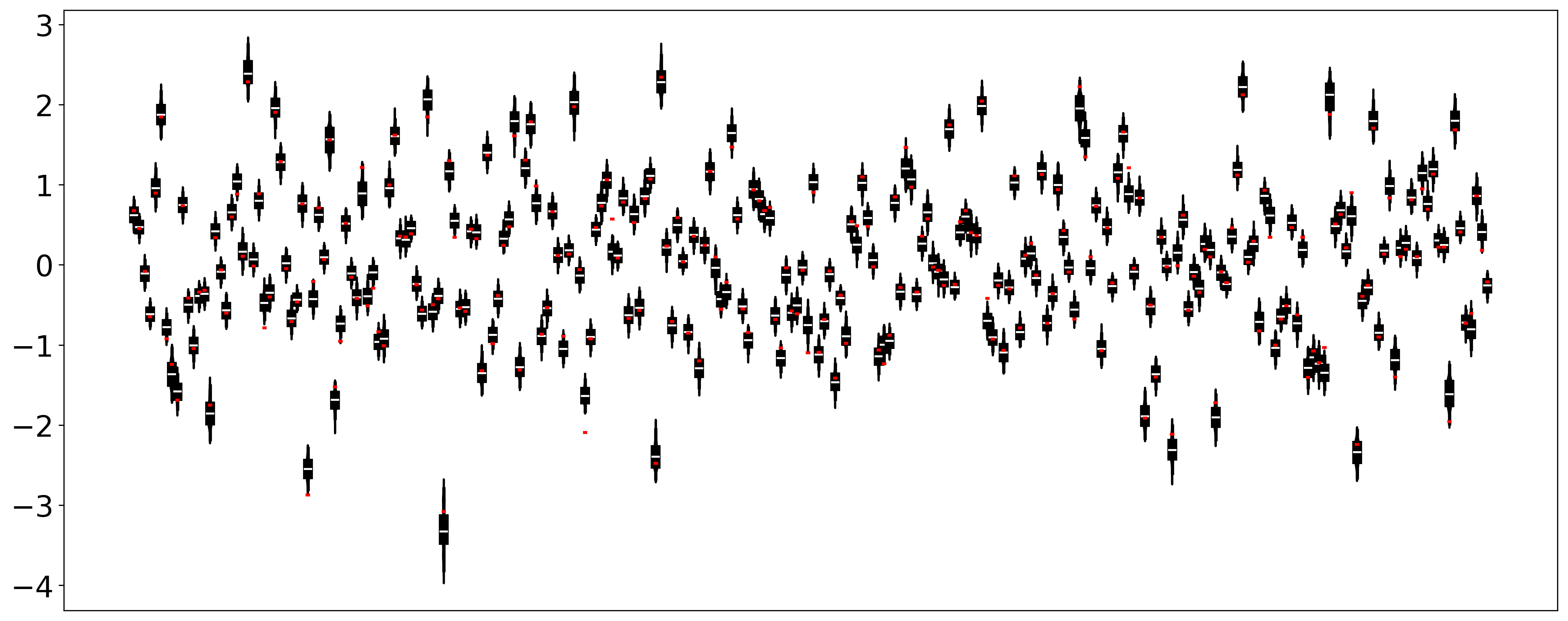}
    \put(43, -1){\scriptsize Parameter}
    \end{overpic}
    \caption{Posterior distributions of $\beta_{i,j}$  under $(d,m) = (50,20)$.}
\end{subfigure}\;
\caption{Violin plots of the posterior distributions of $\beta_{i,j}$ from MCMC sampling for different values of $(d,m)$. The horizontal lines represent the ground truth.}
\label{plot:beta_violin_more}
\end{figure}

\section{Alternative MCMC samplers}

\subsection{Hamiltonian Monte Carlo}\label{sec:alg_nuts}
Hamiltonian Monte Carlo (HMC) is a popular MCMC algorithm for sampling from high-dimensional posterior distributions.
However, for our model, a direct application of HMC is complicated by (i) the constraint-induced latent structure and (ii) the fact that an intractable normalization constant that depends on the latent variables.
Below we describe a practically implementable HMC scheme that addresses (i) via a reparameterization that satisfies the constraints, and (ii) via an approximate posterior that ignores the normalization constant.

Using the reparameterization $v_i=\zeta_i-A^\top u^{(i)}$ hence $\zeta_i=v_i+A^\top u^{(i)}$, we obtain
\(
\mathcal L(y_i,v_i,\zeta_i;\mu_i)
&=
\frac{
\mathcal F(v_i+A^\top u^{(i)};\mu_i)\,\pi_u(u^{(i)})
}{
  c[\mathcal U(y_i,\,v_i+A^\top u^{(i)})]
}
1\{u^{(i)}\in\mathcal U(y_i)\}
\prod_{j=1}^d 1\{s_{i,j}v_{i,j}>0\}.
\)
with $s_{i,j}:=2y_{i,j}-1$ for simplicity. Here the supports of $u^{(i)}$ and $v_i$ become independent:
\[\label{eq:U_V_supp}
& u^{(i)}\in\mathcal U(y_i):=
\{u:
u_k=0 \text{ for }k:(Ay_i-b)_k<0,\; u_k>0 \text{ for }k:(Ay_i-b)_k=0
\},\\
& v_{i,j}\in\{v:s_{i,j}v>0\},
\]
On the other hand, the likelihood remains intractable because the normalization $c[(\mathcal U(y_i,\,v_i+A^\top u^{(i)})]$ depends jointly on $u^{(i)}$ and $v_i$.

To enable tractable computation via HMC, we consider a different posterior by dropping the normalizing constant $c(y_i,\,v_i+A^\top u^{(i)})$. The resulting posterior can be viewed as an approximation to the one presented in the main text, and takes the form
\(
\pi^{\mathrm{HMC}}(\beta, U, V)
\propto
\pi(\beta)
\prod_{i=1}^n
\mathcal F(v_i + A^\top u^{(i)}; \beta x_i)
\;\pi_u(u^{(i)}),
\)
with support \eqref{eq:U_V_supp}. To improve mixing and computational efficiency, we marginalize out the auxiliary $V$, yielding
\[
\label{eq:drop_c_marginal}
\pi^{\text{HMC}}_{\text{marginal}}(\beta, U) \propto \pi(\beta) \prod_{i=1}^{n} \left[\pi_u(u^{(i)}) \prod_{j=1}^{d} \Phi(s_{i,j}(\mu_{ij} - (UA)_{ij}))\right].
\]

 As shown in a later section, the No-U-Turn Hamiltonian Monte Carlo (NUTS) algorithm based on \eqref{eq:drop_c_marginal} shows excellent computational efficiency and can be applied to high dimension settings.

{
\subsection{Pseudo-Marginal Metropolis--Hastings}\label{subsec:pmmh}
We also consider a pseudo-marginal Metropolis--Hastings (PMMH) algorithm \citep{andrieu2009pseudo}, which targets the marginal posterior distribution of the regression coefficient matrix $\beta$:
\(
\Pi(\beta|y_1,\ldots,y_n)\propto \pi(\beta)\prod_{i=1}^nP(y_i \mid \beta, x_i)
= \pi(\beta)\mathbb E_{\zeta_i \sim \mathcal N(\beta x_i, I_d)}
\left[1(T(\zeta_i) = y_i )\right].
\)
PMMH algorithm replaces this quantity with an unbiased Monte Carlo estimator. 
Specifically, for each $i$, we draw $M$ independent Gaussian perturbations
\(
\zeta_i^{(t)} = \beta x_i + \varepsilon_i^{(t)}, 
\qquad \varepsilon_i^{(t)} \sim N(0, I_d), 
\quad t = 1, \ldots, M,
\)
and estimate the likelihood contribution by the empirical proportion
\(
\widehat P(y_i \mid \beta, x_i)
= \frac{1}{M} \sum_{t=1}^M  1(T(\zeta_i^{(t)}) = y_i ).
\)

Let $\widehat L(\beta)$ denote the product of the estimated likelihood contributions across observations. Given a proposal $\beta' \sim q(\beta' \mid \beta)$, the PMMH acceptance probability takes the standard form
\(
\alpha(\beta, \beta') 
= \min\left\{1,\;
\frac{\widehat L(\beta') \, \pi(\beta')}
     {\widehat L(\beta) \, \pi(\beta)}
\frac{q(\beta \mid \beta')}{q(\beta' \mid \beta)} \right\},
\)
where $\pi(\beta)$ denotes the prior distribution.

Empirically, we found that the PMMH algorithm works well for low-dimensional settings with $d\le 10$. As $d$ increases further, zero-likelihood estimates $\widehat P(y_i \mid \beta, x_i)$ become increasingly common, hence requiring  substantially larger $M$.

}

\subsection{Simulation using Hamiltonian Monte Carlo}
{
We follow the same simulation setup as in Section 5.2 of the main text, and test the performance of the HMC algorithm. In each setting, we run NUTS-HMC sampler for 2000 warm-up iterations and 8000 post-warm-up iterations.

As discussed above, this HMC sampler targets a different posterior than the true posterior because the normalizing constant related to the latent variables cannot be evaluated. On the other hand, since the constraint-induced dependency between latent variables $U$ and $\zeta$ is now removed thanks to the reparameterization, HMC can mix more efficiently in higher dimension than the MH-Within-Gibbs sampler. In addition to the settings we presented for the MH-Within-Gibbs sampler, we conduct additional simulations for HMC with $d\in\{2000,3000\}$. 

Across all simulated settings, the HMC chains were stable and showed no divergent transitions. Trace plots and ACF plots in Figure~\ref{plot:beta_ACF_and_trace_NUTS} show that most coordinates of $\beta$ mix reasonably well in both low and moderate dimensions, and the autocorrelation typically drops close to zero immediately. 
On average, it takes $1.9$ seconds to run 1000 iterations for $(d,m)=(2,1)$, $245$ seconds for $(d,m)=(500,100)$, $1589$ seconds for $(d,m)=(3000,100)$.

\begin{figure}[H]
\centering
\begin{subfigure}[t]{0.3\textwidth}
    \begin{overpic}[width=\textwidth, height = 3.8cm]{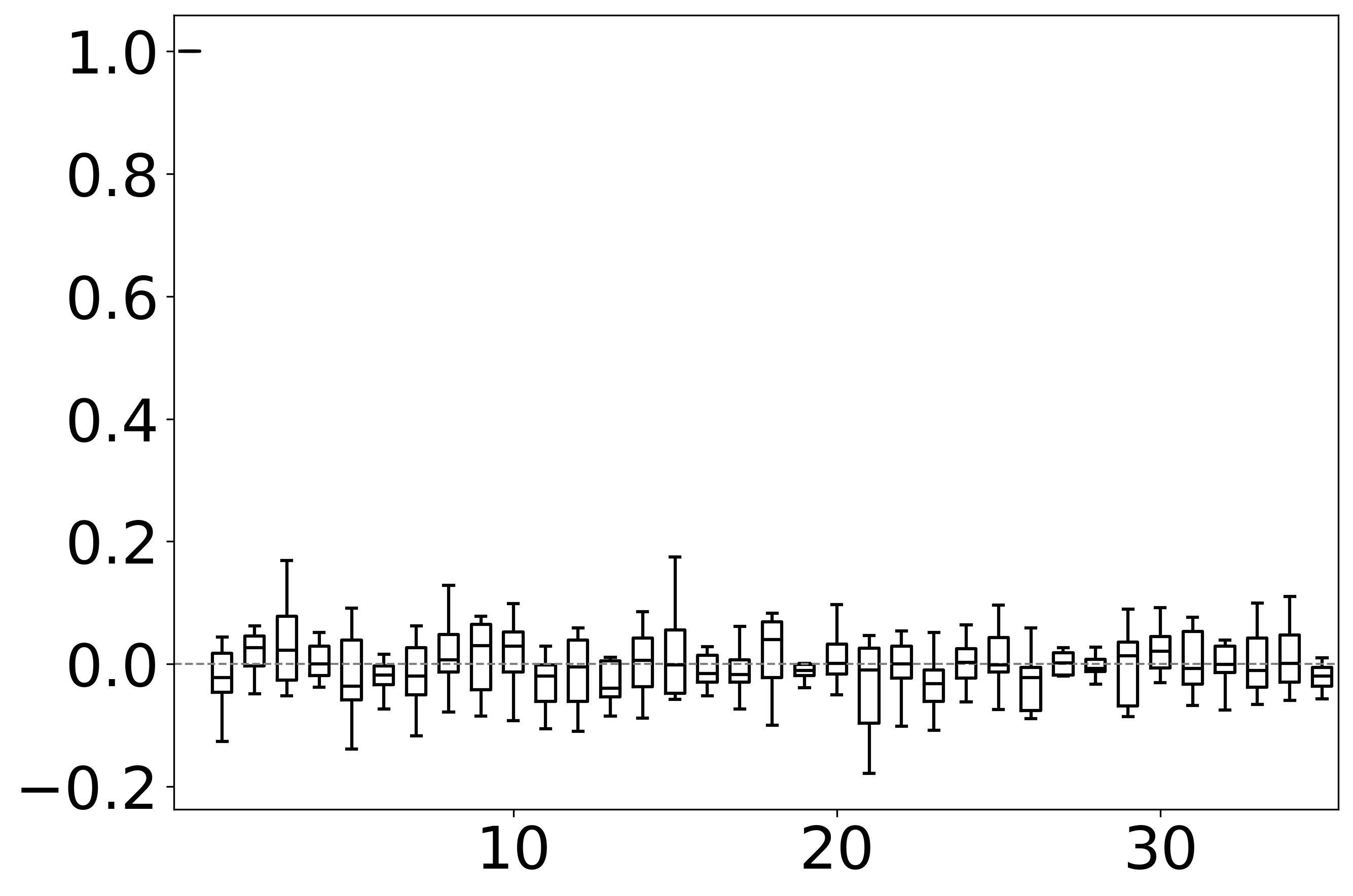}
    \put(50, -4){\scriptsize Lag}
    \end{overpic}
    \caption{$(d,m) = (2,1)$.}
\end{subfigure}\;
\begin{subfigure}[t]{0.3\textwidth}
    \begin{overpic}[width=\textwidth, height = 3.8cm]{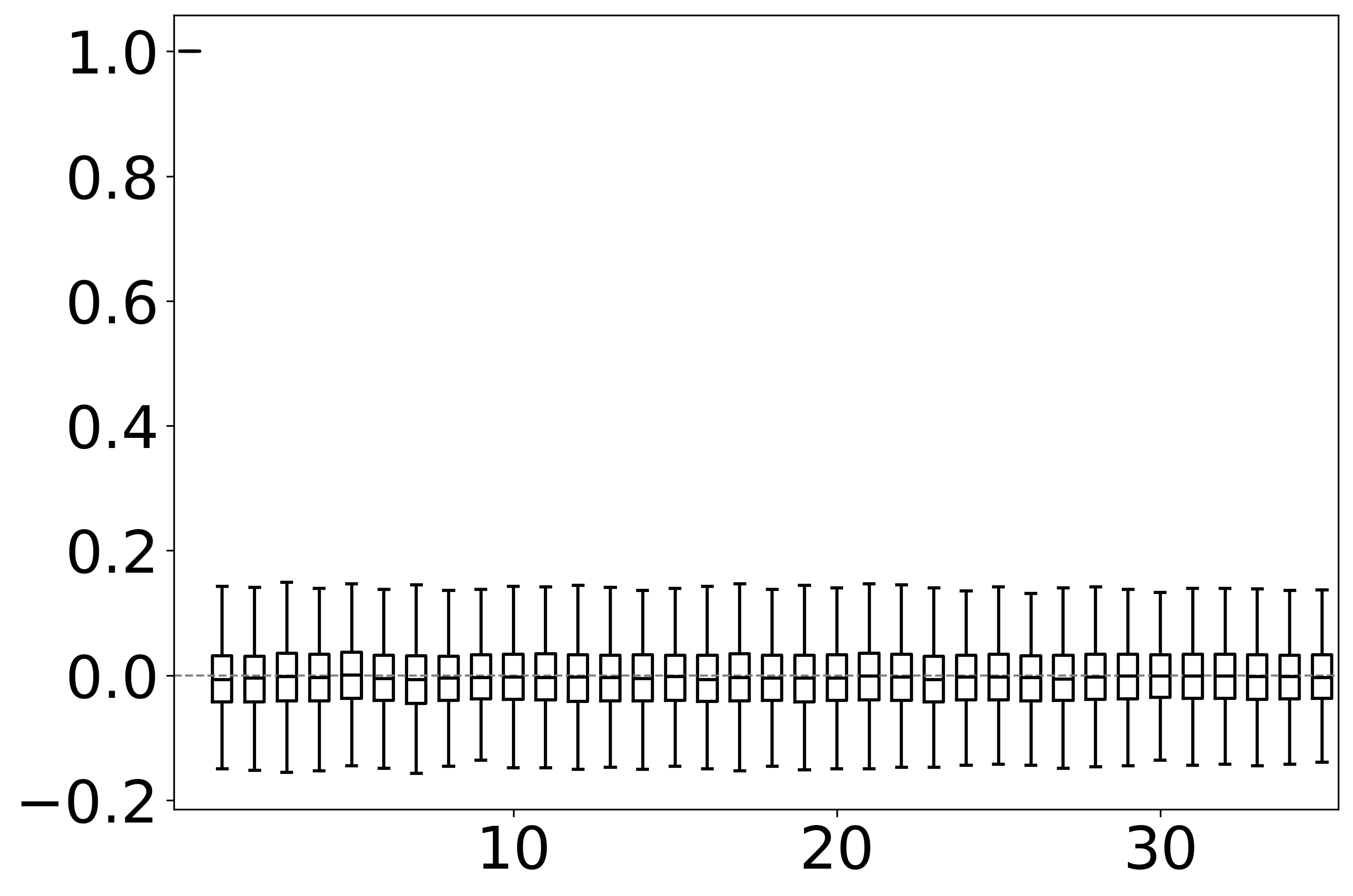}
    \put(50, -4){\scriptsize Lag}
    \end{overpic}
    \caption{$(d,m) = (500,100)$.}
\end{subfigure}\;
\begin{subfigure}[t]{0.3\textwidth}
    \begin{overpic}[width=\textwidth, height = 3.8cm]{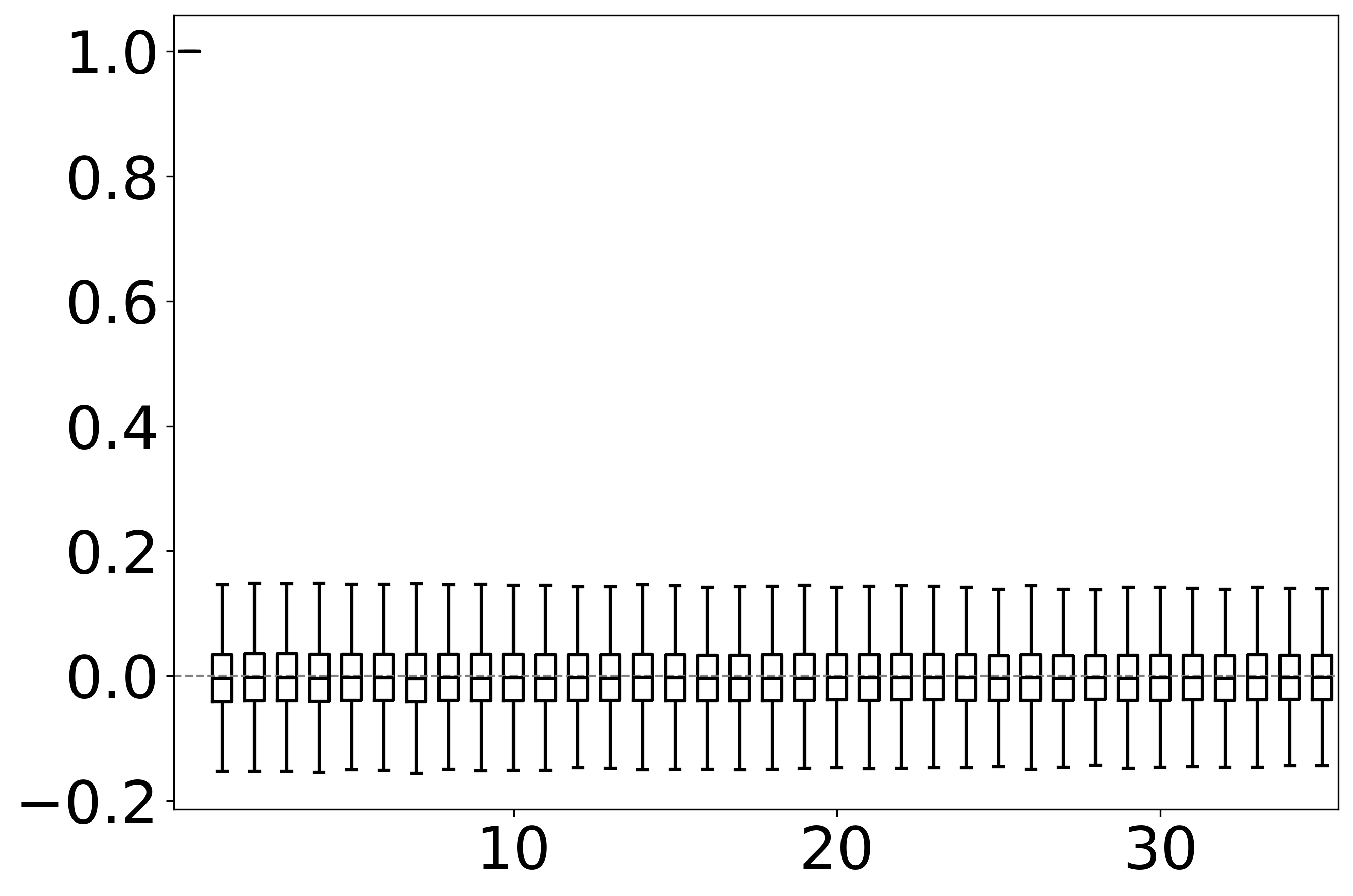}
    \put(50, -4){\scriptsize Lag}
    \end{overpic}
    \caption{$(d,m) = (3000,100)$.}
\end{subfigure}\\
\begin{subfigure}[t]{0.3\textwidth}
    \begin{overpic}[width=\textwidth, height = 3.8cm]{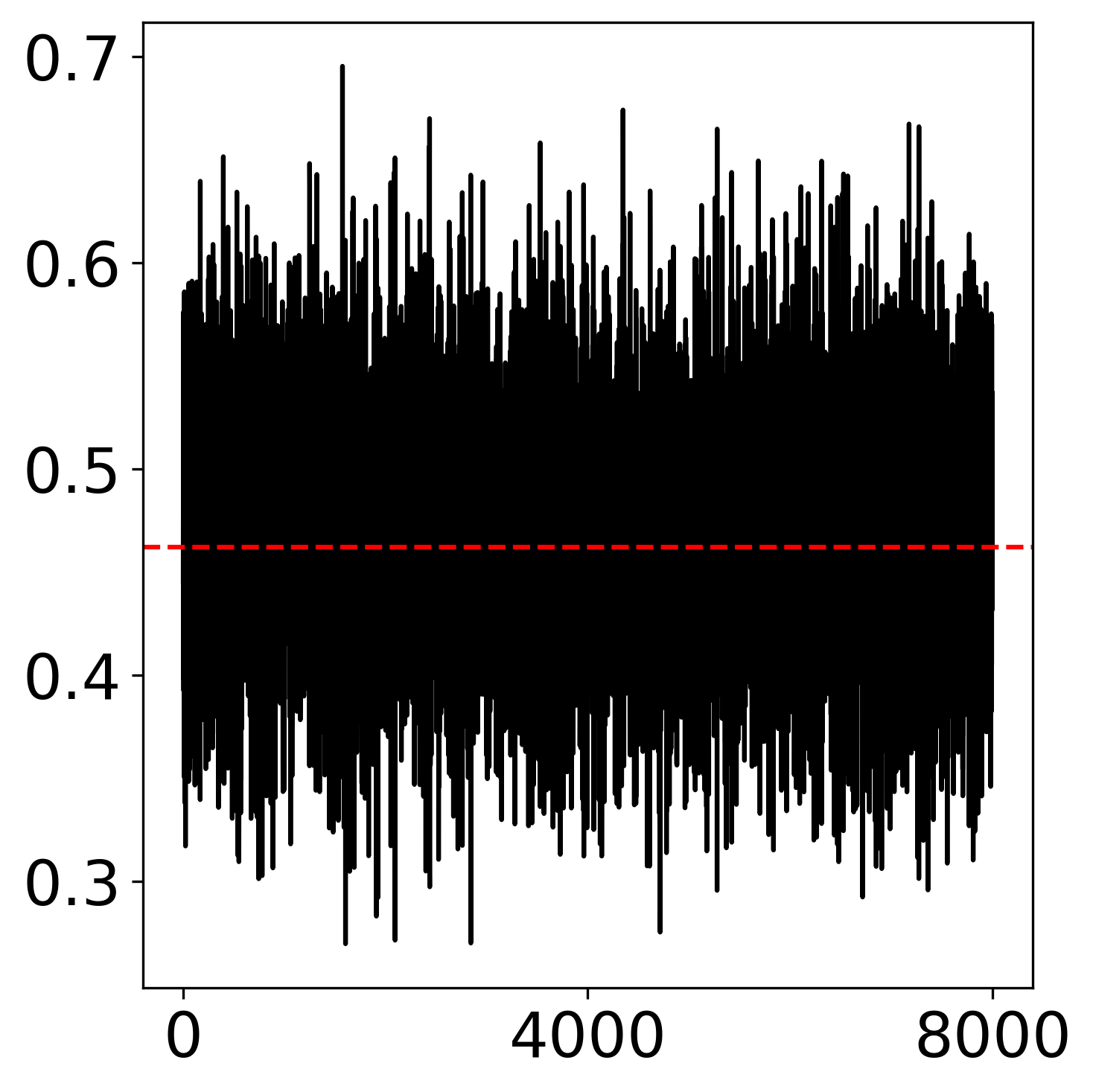}
    \put(44, -4){\scriptsize Iteration}
    \end{overpic}
    \caption{$(d,m) = (2,1)$.}
\end{subfigure}\;
\begin{subfigure}[t]{0.3\textwidth}
    \begin{overpic}[width=\textwidth, height = 3.8cm]{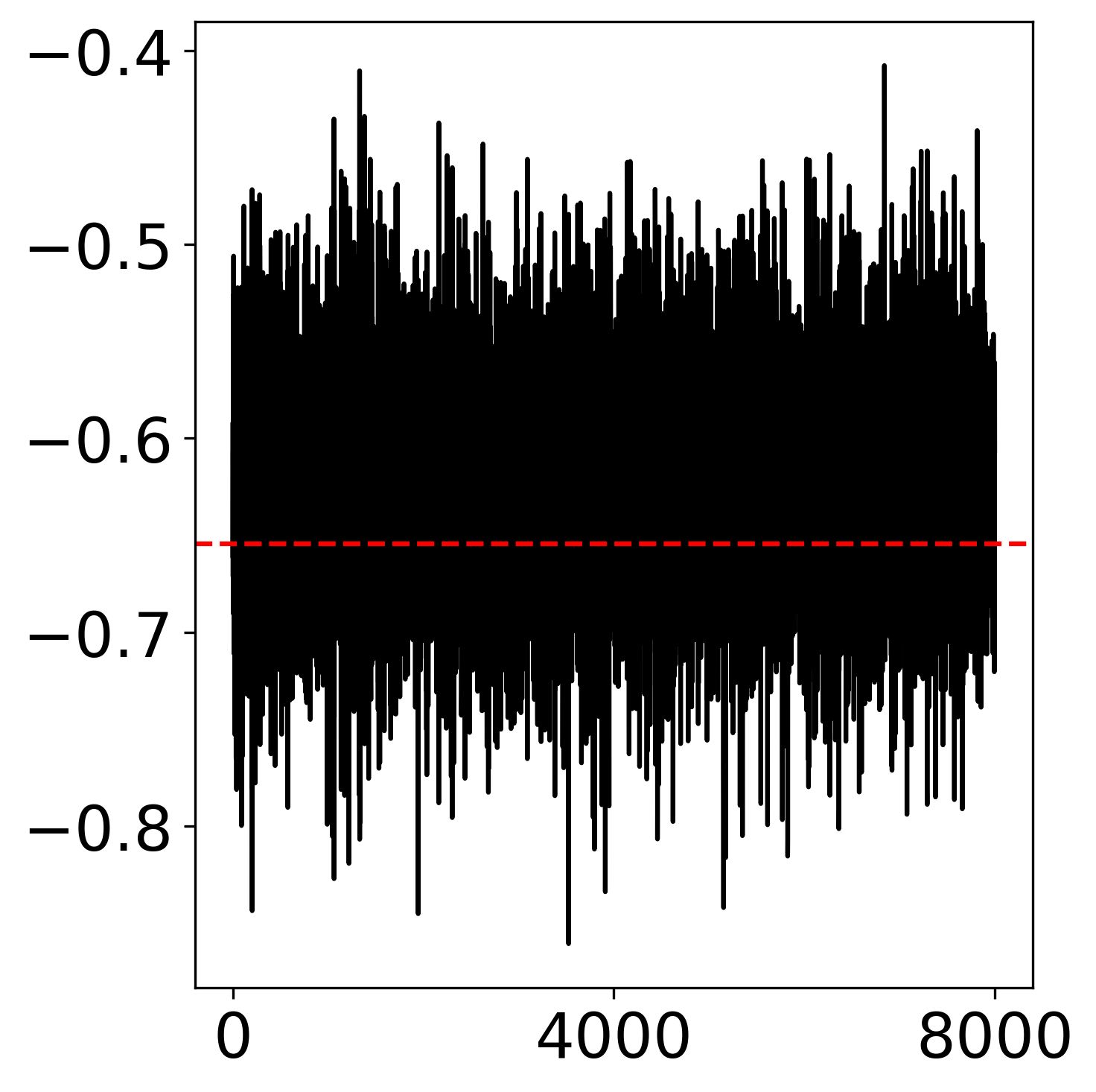}
    \put(44, -4){\scriptsize Iteration}
    \end{overpic}
    \caption{$(d,m) = (500,100)$.}
\end{subfigure}\;
\begin{subfigure}[t]{0.3\textwidth}
    \begin{overpic}[width=\textwidth, height = 3.8cm]{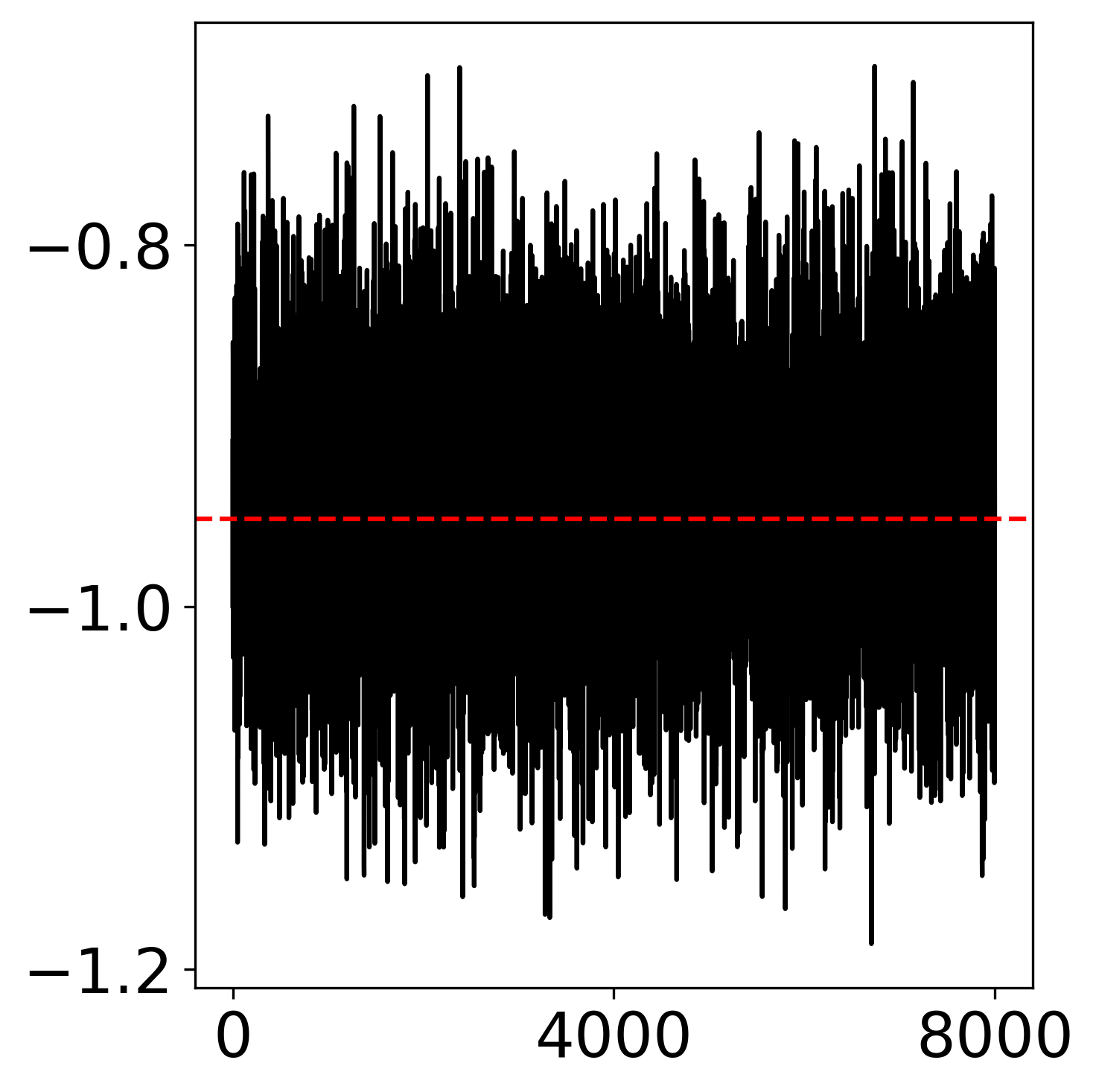}
    \put(44, -4){\scriptsize Iteration}
    \end{overpic}
    \caption{$(d,m) = (3000,100)$.}
\end{subfigure}
\caption{Autocorrelation function (ACF) and trace plots for different values of $(d,m)$ for HMC. Each box displays the ACF of all $\beta_{i,j}$ parameters. The ACF is computed after thinning the posterior samples by retaining every 25th sample. The trace plots are based on a randomly chosen parameter under each setting.}
\label{plot:beta_ACF_and_trace_NUTS}
\end{figure}

Since ignoring the normalizing constant alters the target distribution, posterior summaries such as the posterior mean of $\beta$ are less accurate than those from the MH-Within-Gibbs sampler and often exhibit larger RMSE than the MH-Within-Gibbs sampler. We report the RMSE for estimating $\beta$ in Table~\ref{tb:RMSE_NUTS}.
When $m$ is relatively small compared to $d$, the HMC algorithm achieves a comparable RMSE to that of the MH-Within-Gibbs sampler.
When $m$ increases, RMSE increases as the bias of ignoring the normalizing constant becomes non-negligible.

\begin{table}[H]
  \centering
  \small
  \begin{tabular}{ccccccccccc}
    \hline
        & \multicolumn{9}{c}{d}                                                 \\ \cline{2-11} 
    m   & 2     & 10     & 20    & 50    & 100    & 200   & 500   & 1000   & 2000 & 3000 \\ \hline
    1   & 0.113 & 0.161 & 0.112 & 0.091 & 0.089 & 0.086 & 0.083 & 0.082 & 0.080 & 0.080 \\
    5   & -     & 0.304 &  0.212& 0.114 & 0.096 & 0.087 & 0.081 & 0.084 & 0.082  &0.082 \\
    10  & -     & 0.379     & 0.299 & 0.189 & 0.139 & 0.116 & 0.093 & 0.087 & 0.083 & 0.083\\
    20  & -     & 1.175     & 0.595     &  0.217    & 0.179 & 0.131 & 0.105 & 0.094 & 0.088 & 0.085 \\
    50  & -     & -     & -     & -     & 0.283     &  0.211     & 0.136 & 0.108 & 0.093 & 0.092 \\
    100 & -     & -     & -     & -     & -     & 0.267    & 0.175    & 0.136     & 0.107 & 0.098 \\ \hline
    \end{tabular}
    \caption{Root mean squared error (RMSE) for estimating $\beta$ using posterior mean under different dimensionality $(d,m)$ for NUTS-HMC.}
  \label{tb:RMSE_NUTS}
  \end{table}
}

{
\subsection{Simulation using Pseudo-Marginal MCMC}\label{subsec:pmmh_simulation}
We next compare the PMMH algorithm with the MH-Within-Gibbs sampler. 
Specifically, we consider configurations $(d,m) \in \{(2,1), (5,2), (10,5)\}$ and we fix $n=1,000$ and $p=5$. We run $50,000$ iterations and discard the first $5,000$ as warm-up. We take $M=1,000$ for estimating the log-likelihood for each observation. 
A Gaussian random-walk proposal is employed for $\beta$, the scale of which is chosen from a grid, $\{0.1,0.05,0.01,0.005\}$, so that the acceptance rate is away from 0. 

Due to the higher computational cost, PMMH takes more time than the MH-Within-Gibbs sampler to run. Specifically, the running time per 1000 iterations of PMMH for $(d,m)=(2,1)$ is 1.5 minutes, for $(d,m)=(5,2)$ is 2.2 minutes, for $(d,m)=(10,5)$ is 17.3 minutes.
Figure~\ref{plot:beta_ACF_and_trace_PMMH} reports representative trace plots and ACF plots for PMMH. For $(d,m)=(10,5)$, the mixing is very slow. For comparison, we add one more column (plot (d) and plot (h)) for a trace plot and an ACF plot at $(d,m)=(10,5)$ for our MH-Within-Gibbs sampler. The results show that the MH-Within-Gibbs sampler has faster mixing than PMMH. For comparison of the estimator accuracy, we use RMSE for estimating $\beta$ and show it in Table~\ref{tb:RMSE_PMMH}. 
When $d$ exceeds $10$, the estimate of the likelihood becomes zero for most observations due to high rejection rate, leading to non-ideal performance of the PMMH algorithm, hence we do not report the results for $(d,m)=(10,5)$ for PMMH. Potential remedies for this issue may be found with some alternative unbiased Monte Carlo estimator, which is beyond the scope of this paper. Readers are recommended to  \citet{andrieu2009pseudo,lee2019unbiased} and the references therein for more details.

\begin{figure}[H]
\centering
\begin{subfigure}[t]{0.23\textwidth}
    \begin{overpic}[width=\textwidth, height = 3cm]{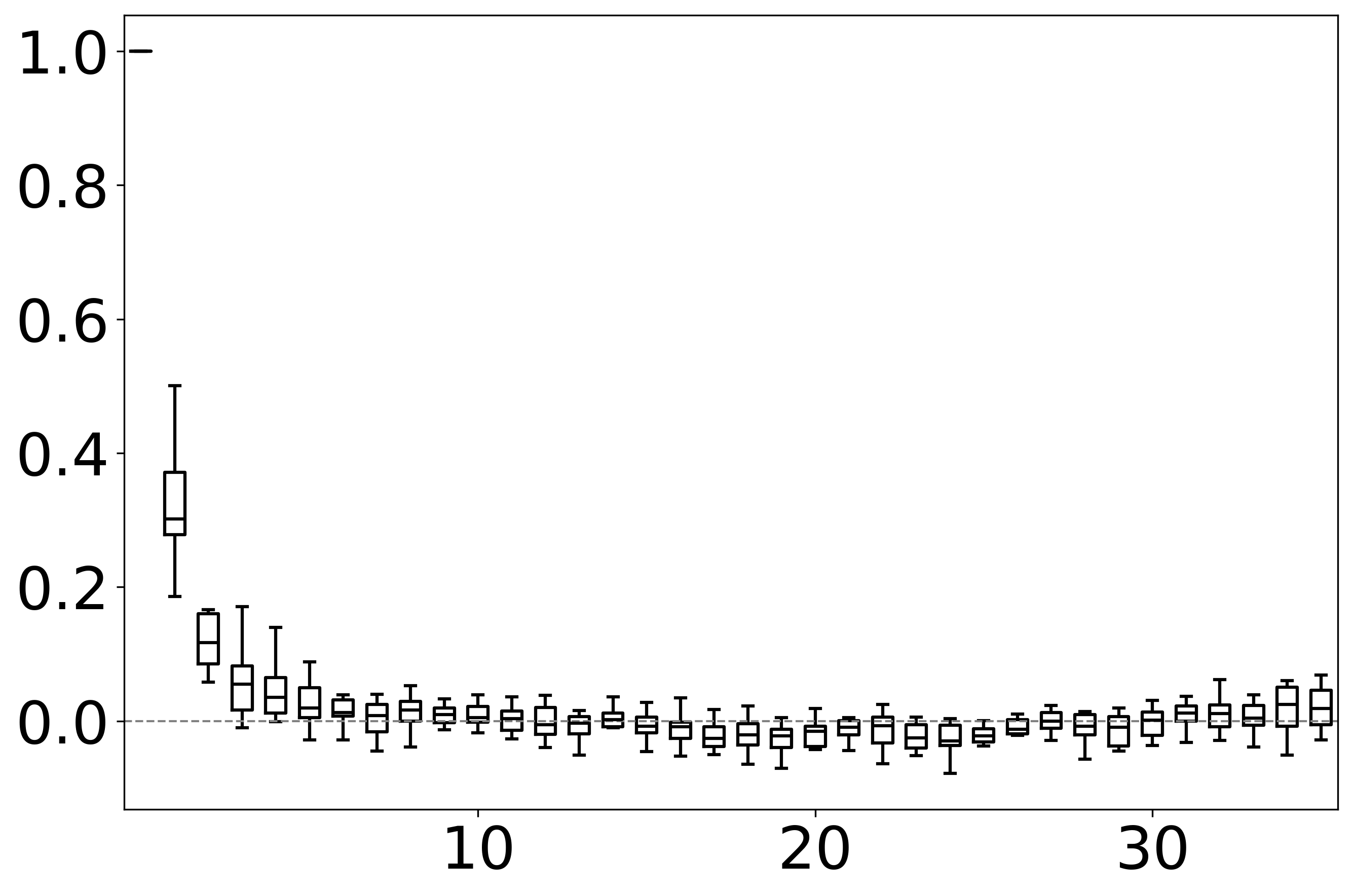}
    \put(50, -4){\scriptsize Lag}
    \end{overpic}
    \caption{$(d,m) = (2,1)$.}
\end{subfigure}\;
\begin{subfigure}[t]{0.23\textwidth}
    \begin{overpic}[width=\textwidth, height = 3cm]{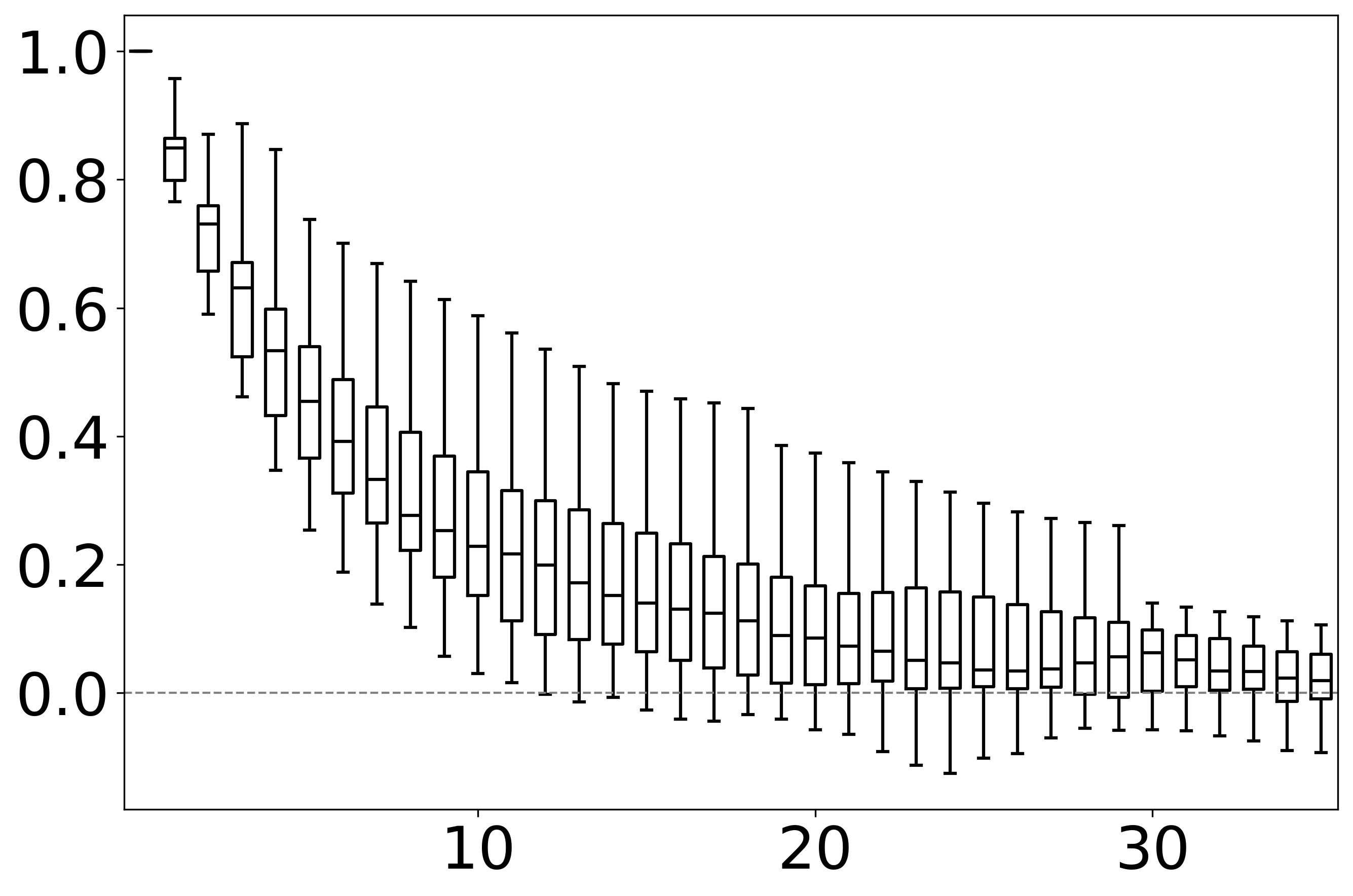}
    \put(50, -4){\scriptsize Lag}
    \end{overpic}
    \caption{$(d,m) = (5,2)$.}
\end{subfigure}\;
\begin{subfigure}[t]{0.23\textwidth}
    \begin{overpic}[width=\textwidth, height = 3cm]{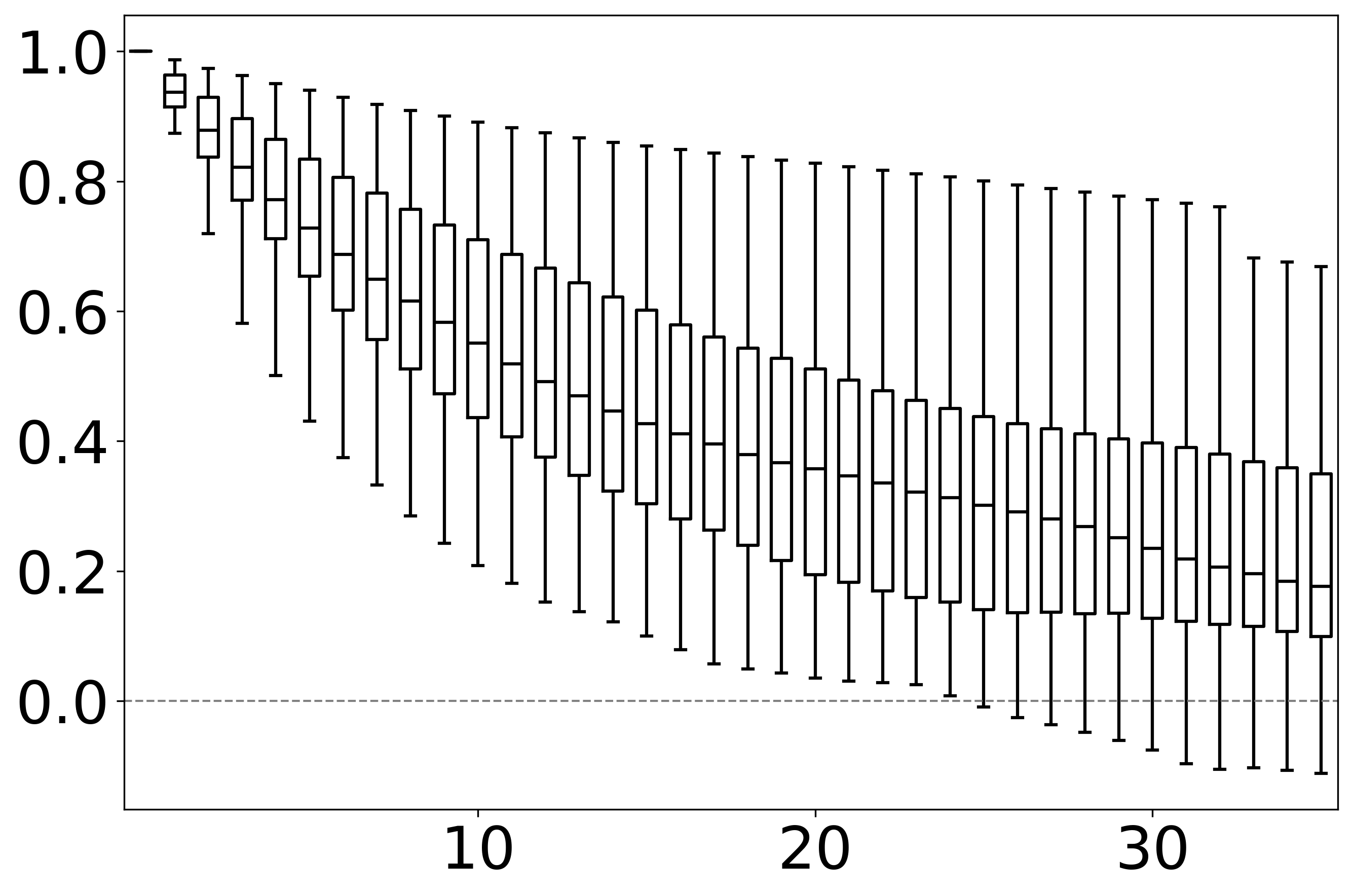}
    \put(50, -4){\scriptsize Lag}
    \end{overpic}
    \caption{$(d,m) = (10,5)$.}
\end{subfigure}
\begin{subfigure}[t]{0.23\textwidth}
    \begin{overpic}[width=\textwidth, height = 3cm]{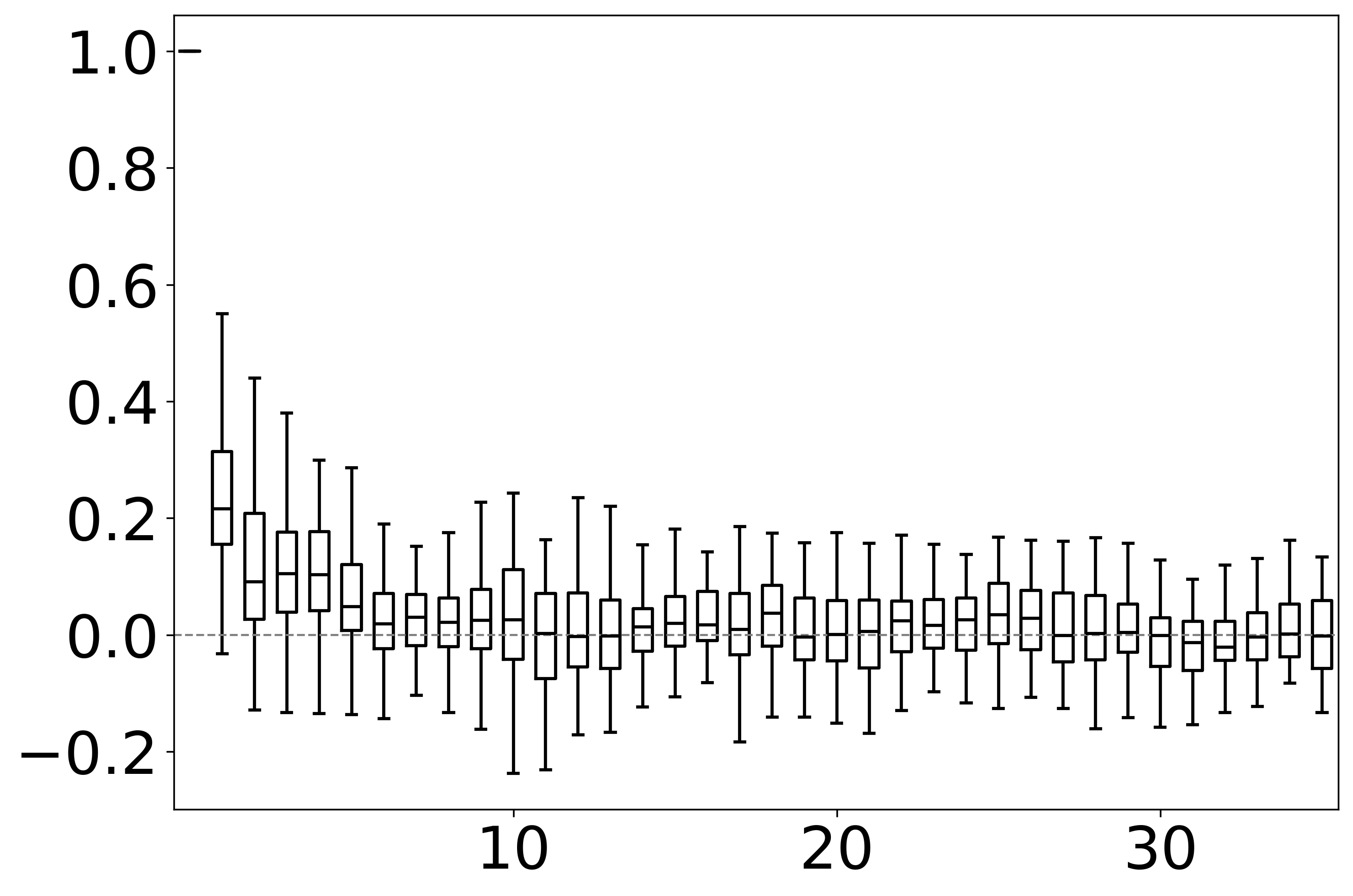}
    \put(50, -4){\scriptsize Lag}
    \end{overpic}
    \caption{$(d,m) = (10,5)$.}
\end{subfigure}\\
\begin{subfigure}[t]{0.23\textwidth}
    \begin{overpic}[width=\textwidth, height = 3cm]{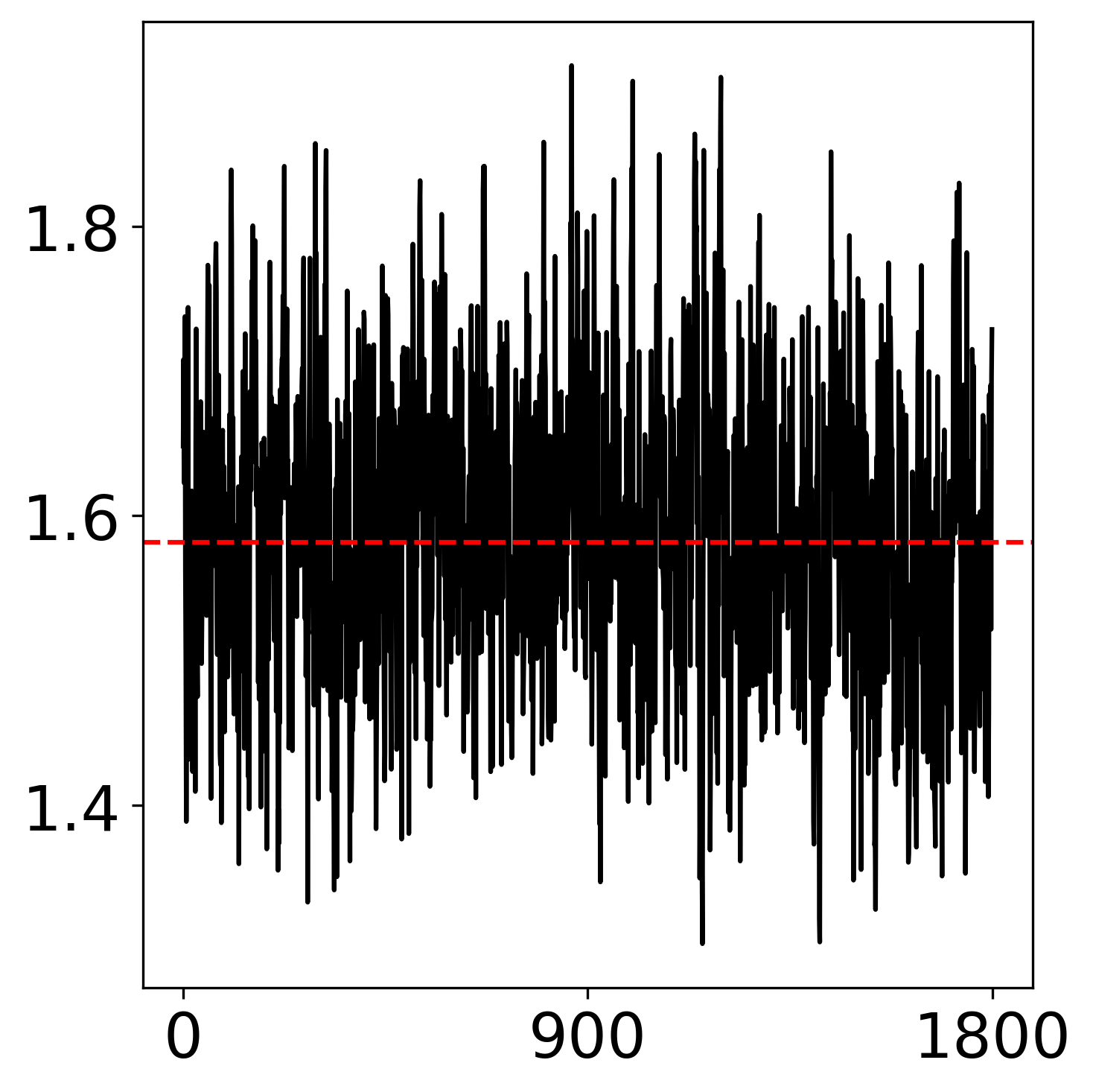}
    \put(44, -4){\scriptsize Iteration}
    \end{overpic}
    \caption{$(d,m) = (2,1)$.}
\end{subfigure}\;
\begin{subfigure}[t]{0.23\textwidth}
    \begin{overpic}[width=\textwidth, height = 3cm]{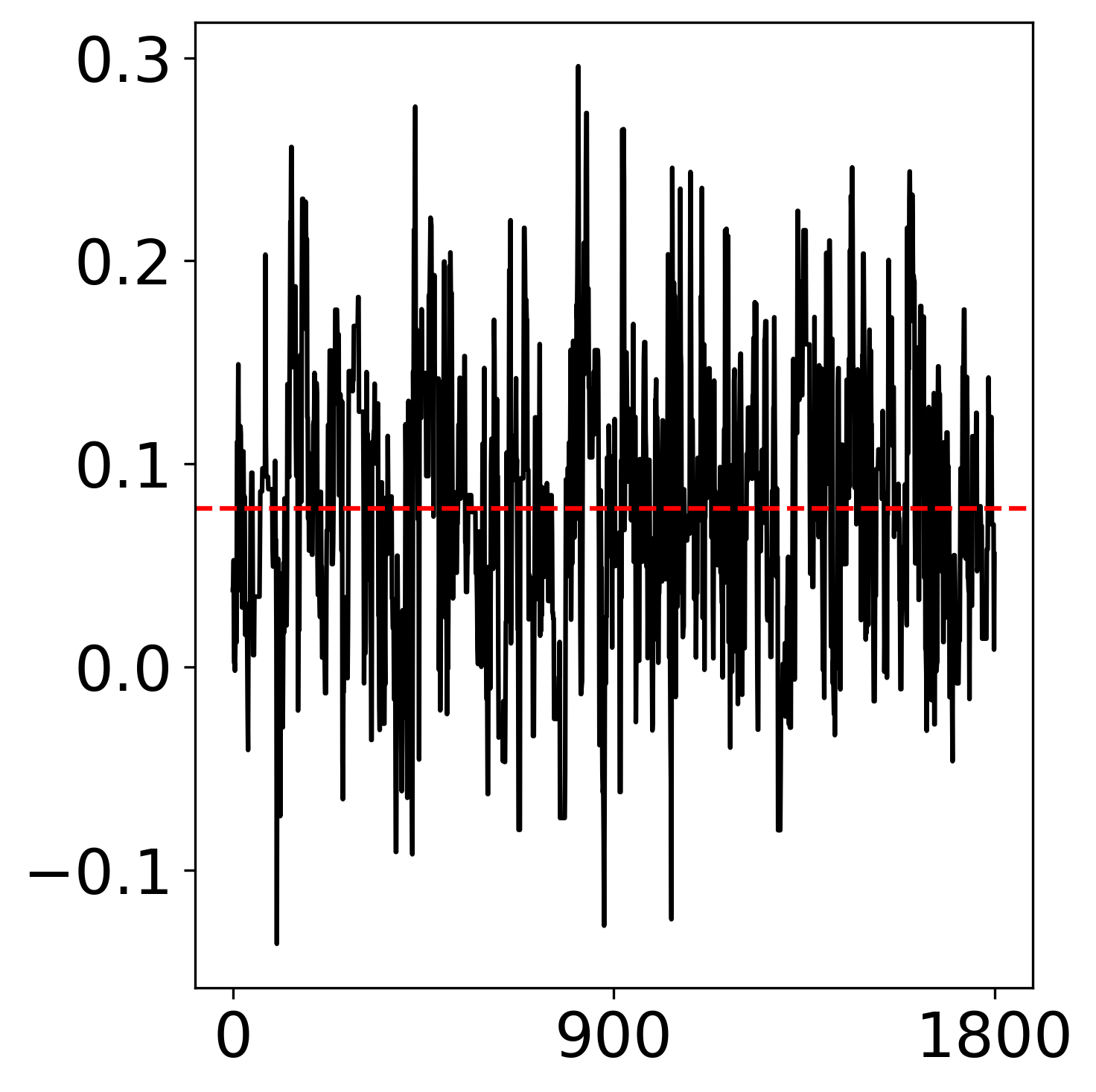}
    \put(44, -4){\scriptsize Iteration}
    \end{overpic}
    \caption{$(d,m) = (5,2)$.}
\end{subfigure}\;
\begin{subfigure}[t]{0.23\textwidth}
    \begin{overpic}[width=\textwidth, height = 3cm]{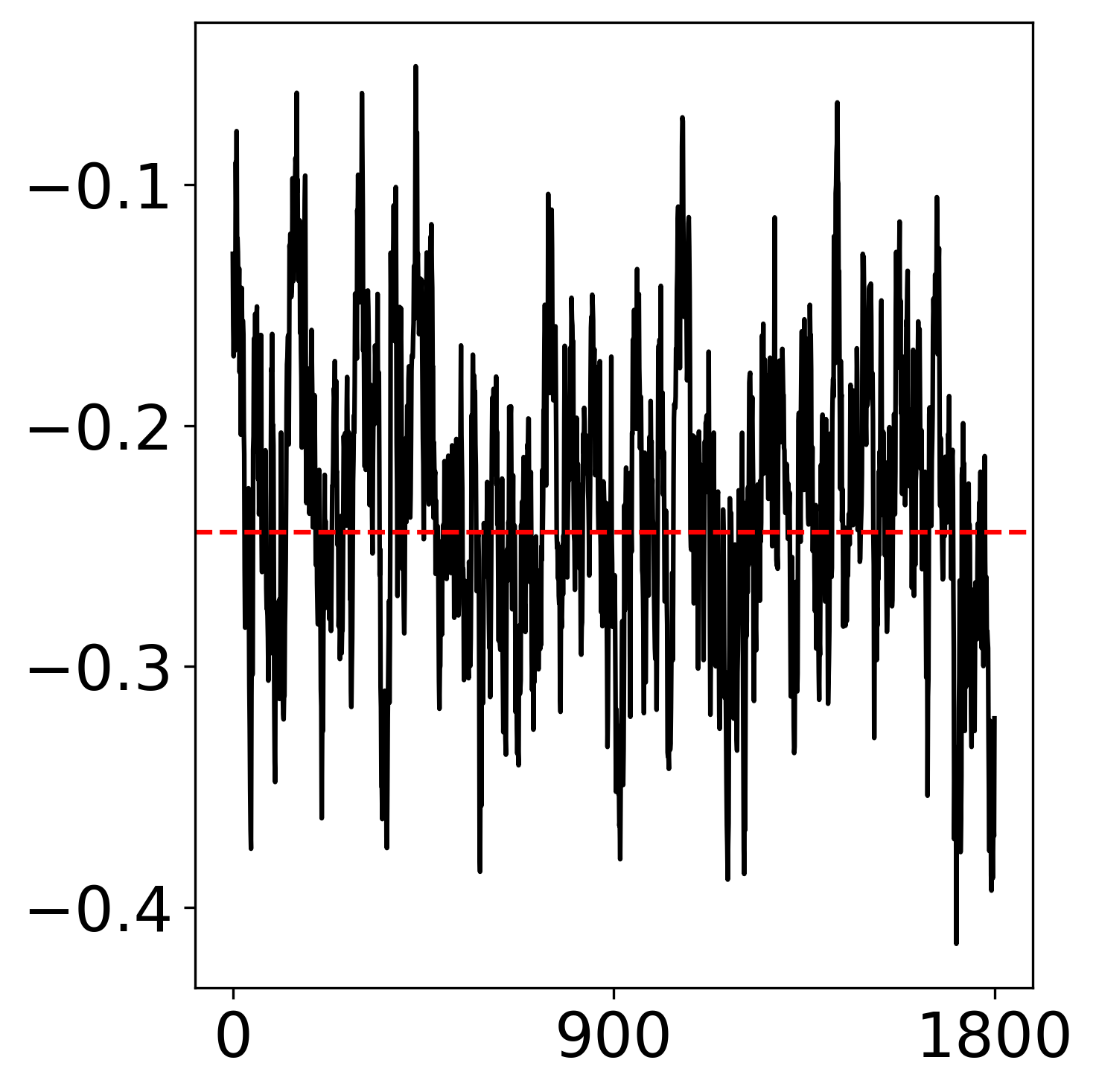}
    \put(44, -4){\scriptsize Iteration}
    \end{overpic}
    \caption{$(d,m) = (10,5)$.}
\end{subfigure}\;
\begin{subfigure}[t]{0.23\textwidth}
    \begin{overpic}[width=\textwidth, height = 3cm]{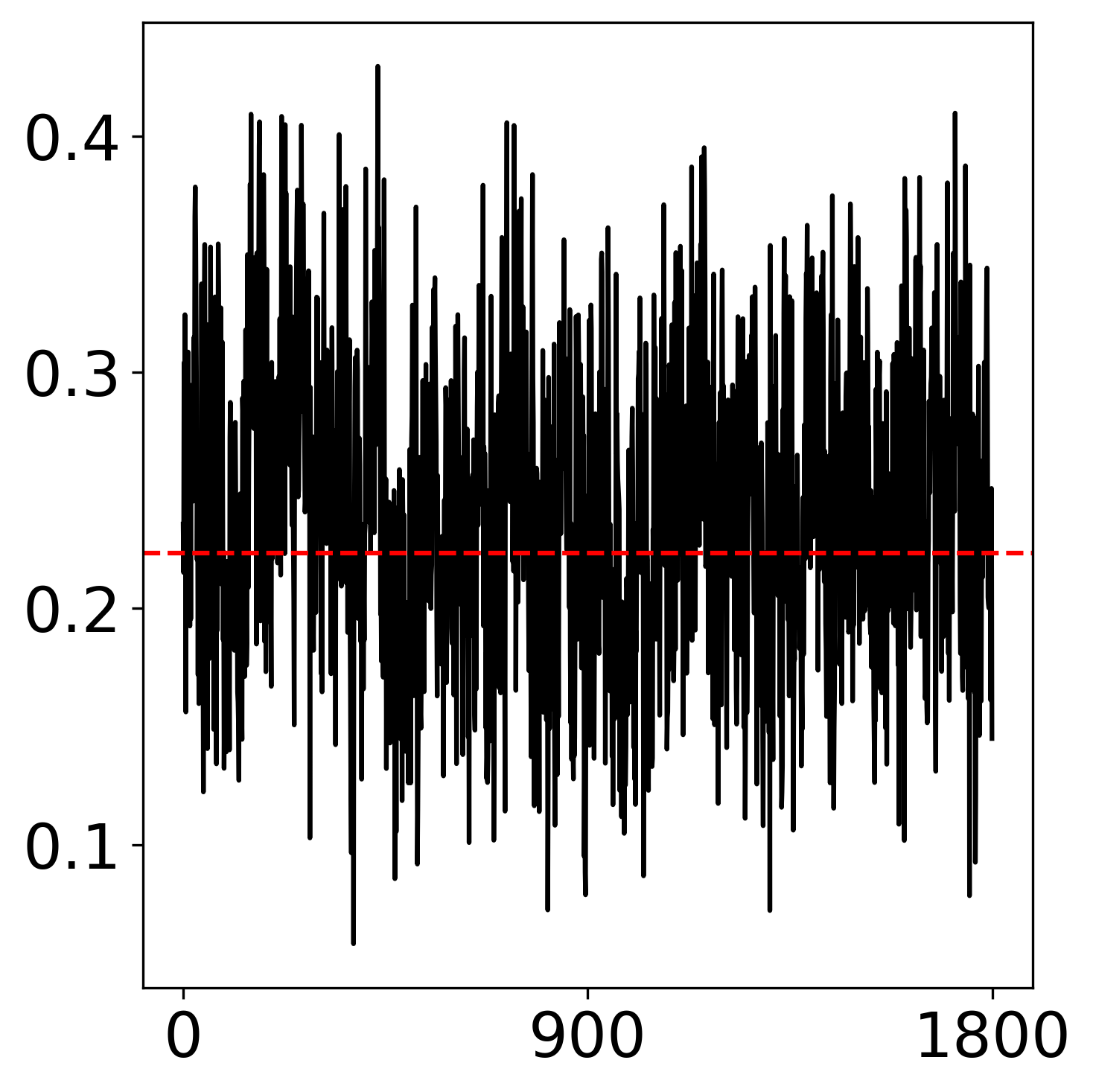}
    \put(44, -4){\scriptsize Iteration}
    \end{overpic}
    \caption{$(d,m) = (10,5)$.}
\end{subfigure}
\caption{Autocorrelation function (ACF) and trace plots for different values of $(d,m)$ for PMMH (plots (a)-(c),(e)-(g)) and for MH-Within-Gibbs sampler (plots (d) and (h)). Each box displays the ACF of all $\beta_{i,j}$ parameters. The ACF is computed after thinning the posterior samples by retaining every 25th sample. The trace plots are based on a randomly chosen parameter under each setting.}
\label{plot:beta_ACF_and_trace_PMMH}
\end{figure}

\begin{table}[H]
    \centering
    \begin{tabular}{cccc}
    \hline
         & $(d,m)=(2,1)$ & $(d,m)=(5,2)$ & $(d,m)=(10,5)$ \\\hline
       PMMH  & 0.089 & 0.080 & Fail to converge\\
       MH-Within-Gibbs  & 0.046 & 0.066 & 0.066 \\\hline
    \end{tabular}
    \caption{Root mean squared error (RMSE) for estimating $\beta$ using posterior mean under different low dimensionality $(d,m)$ for PMMH and MH-Wihtin-Gibbs.}
    \label{tb:RMSE_PMMH}
\end{table}

}

{
\section{Alternative kernels for the auxiliary variable $u$}

This section investigates whether the exponential kernel used for the auxiliary variables $u^{(i)}$ in the MH-within-Gibbs sampler is essential, and how alternative choices may affect sampler performance. We study this question empirically by replacing the exponential kernel with a half-Gaussian kernel and comparing the resulting sampler with the original implementation. More specifically, we consider the following conditional distribution of $u^{(i)}$ given $(y_i,\zeta_i)$:
\(
\Pi_u(u^{(i)}\mid y_i,\zeta_i)
=\frac{\prod_{k:(Ay_i-b)_k=0}\exp\!\left[-\tfrac{1}{2}(u^{(i)}_k)^2\right]}
{c'[\mathcal U(y_i,\zeta_i)]},
\quad u^{(i)}\in\mathcal U(y_i,\zeta_i),
\)
where $c'[\mathcal U(y_i,\zeta_i)]$ denotes the normalizing constant.

The hit-and-run update for the half-Gaussian case closely parallels the exponential-kernel sampler described in the main paper, differing only in the one-dimensional sampling step along a randomly chosen direction. In particular, in step (ii) of the algorithm, we draw $\alpha^*$ from a truncated normal distribution supported on $(\tilde a,\tilde b)$, with mean $-\sum_k \omega_k v_k / \sum_k v_k^2$ and variance $1 / \sum_k v_k^2$, where $(\tilde a,\tilde b)$ is the longest interval such that $\omega + \alpha v \in \mathcal U(y_i,\zeta_i)$ for all $\alpha \in (\tilde a,\tilde b)$.

We evaluate the impact of the kernel choice using the same constraint structure and data-generating design ($n=1{,}000,p=5$) as in Section~5.2 of the main paper. For each combination of $d$ and $m$, we run the MH-within-Gibbs sampler using either the exponential kernel or the half-Gaussian kernel for $u^{(i)}$, while keeping all other components identical. Each chain is run for $10{,}000$ iterations, with the first $5{,}000$ iterations discarded as warm-up.

Across all settings considered, diagnostic measures—including root mean squared error (RMSE) for estimating $\beta$, effective sample sizes, trace plots, autocorrelation functions, and total runtime—show very similar behavior under the two kernels. In particular, we observe no systematic differences in mixing behavior or computational cost. For brevity, we only report RMSE for estimating $\beta$ for several representative combinations of $d$ and $m$ in Table~\ref{tb:kernel_compare}.

\begin{table}[H]
  \centering
  \small
  \begin{tabular}{ccc}
\toprule
$(d,m)$ &  Exponential kernel & Half-Gaussian kernel \\
\midrule
$(5,1)$ & 0.065 & 0.067\\
$(20,5)$ & 0.066 & 0.065\\
$(50,10)$ & 0.113 & 0.105\\
$(200,50)$ & 0.167 & 0.175\\
$(1000,100)$ & 0.128 & 0.133\\
\bottomrule
\end{tabular}
    \caption{Root mean squared error (RMSE) for estimating $\beta$ using posterior mean under different dimensionality $(d,m)$ for MH-Within-Gibbs samplers with exponential kernel and half-Gaussian kernel.}
  \label{tb:kernel_compare}
\end{table}

Overall, these findings suggest that the exponential kernel used in the main paper is not essential for good sampler performance; alternative choices such as the half-Gaussian kernel yield comparable statistical accuracy and computational efficiency.

}

\bibliographystyle{chicago}
\bibliography{ref_fixed}

\end{document}